\documentclass{article}
\usepackage{kr}

\usepackage{times}
\usepackage{soul}
\usepackage{url}
\usepackage[hidelinks]{hyperref}
\usepackage[utf8]{inputenc}
\usepackage[small]{caption}
\usepackage{graphicx}
\usepackage{amsmath}
\usepackage{amsthm}
\usepackage{booktabs}
\usepackage{algorithm}
\usepackage{algorithmic}
\usepackage{lineno}
\newtheorem{theorem}{Theorem}

\title{Partially Finite Model Reasoning in Description Logics\\ Extended Version}

\author{
Tomasz~Gogacz\textsuperscript{\rm 1}
\and
Filip~Murlak\textsuperscript{\rm 1}
\and
Marcin~Przyby{\l}ko\textsuperscript{\rm 1}
\and
Alexandra~Rogova\textsuperscript{\rm 1}
\And 
Micha{\l}~Skrzypczak\textsuperscript{\rm 1} \\
\affiliations
\textsuperscript{\rm 1}University of Warsaw\\
\emails
\{t.gogacz, fmurlak, m.przybylko, rogova, mskrzypczak\}@mimuw.edu.pl
}

\usepackage{booktabs}
\usepackage{amsthm}
\usepackage{amsmath}
\usepackage{amssymb}
\usepackage{amsopn}
\usepackage{mathptmx}
\usepackage{stmaryrd}
\usepackage{thm-restate}
\newtheorem{definition}[theorem]{Definition}
\newtheorem{lemma}[theorem]{Lemma}
\newtheorem{claim}[theorem]{Claim}
\newtheorem{fact}[theorem]{Fact}

\usepackage[dvipsnames]{xcolor}

\usepackage[shortcuts]{extdash}

\usepackage{xspace}

\usepackage{graphicx}
\usepackage{subcaption}

\usepackage{tikz}
\usetikzlibrary {automata,positioning,graphs, calc, decorations.pathmorphing, decorations.shapes, patterns}

\usepackage{pifont}
\usepackage{ mathdots }

\newcommand{\distN}{k}

\newcommand{\Del}{D}
\newcommand{\Nel}{E}

\newcommand{\ignore}[1]{}

\newcommand{\TT}{{\mathcal T}}

\newcommand{\EE}{{\mathcal E}}

\newcommand{\KK}{{\mathcal K}}
\newcommand{\QQ}{{\mathcal Q}}
\newcommand{\II}{{\mathcal I}}
\newcommand{\JJ}{{\mathcal J}}
\newcommand{\LL}{{\mathcal L}}
\newcommand{\UU}{{\mathcal U}}

\newcommand{\GG}{{\mathcal G}}

\newcommand{\HH}{{\mathcal H}}
\DeclareMathOperator{\Ima}{Im}

\newcommand{\Slogic}{\ensuremath{\mathcal{S}}\xspace}

\newcommand{\ConNameSet}{\ensuremath{\mathtt{N_C}}\xspace}
\newcommand{\RolNameSet}{\ensuremath{\mathtt{N_R}}\xspace}
\newcommand{\TRolNameSet}{\ensuremath{\mathtt{N^t_R}}\xspace}
\newcommand{\NTRolNameSet}{\ensuremath{\mathtt{N^{nt}_R}}\xspace}
\newcommand{\IndNameSet}{\ensuremath{\mathtt{N_I}}\xspace}
\newcommand{\VarSet}{\ensuremath{\mathtt{N_V}}\xspace}
\newcommand{\TypesSet}[1]{\ensuremath{\mathtt{Tp}(#1)}\xspace}

\newcommand{\dom}{\mathrm{dom}\xspace}
\newcommand{\range}{\mathrm{rg}\xspace}

\newcommand{\KB}{\ensuremath{\mathcal{K}}\xspace}
\newcommand{\TBox}{\ensuremath{\mathcal{T}}\xspace}
\newcommand{\ABox}{\ensuremath{\mathcal{A}}\xspace}
\newcommand{\Fin}{\ensuremath{F}\xspace}
\newcommand{\intp}{\ensuremath{\mathcal{I}}\xspace}
\newcommand{\intb}{\ensuremath{\mathcal{J}}\xspace}
\newcommand{\intq}{\ensuremath{\mathcal{Q}}\xspace}

\newcommand{\query}{\ensuremath{Q}\xspace}

\newcommand{\from}{\colon}
\newcommand{\N}{\mathbb{N}\xspace}

\newcommand{\conceptsIn}[1]{\ensuremath{\mathtt{CN}(#1)}\xspace}
\newcommand{\rolesIn}[1]{\ensuremath{\mathtt{Rol}(#1)}\xspace}

\newcommand{\indsIn}[1]{\ensuremath{\mathtt{Ind}(#1)}\xspace}
\newcommand{\types}[2]{\ensuremath{\mathtt{tp}^{#1}(#2)}\xspace}

\newcommand{\restr}{\upharpoonright}

\newcommand{\eqdef}{\stackrel{\text{def}}=}

\newcommand{\rch}[2]{\mathtt{rch}^{#2}_{#1}}

\newcommand{\exptime}{\ensuremath{\mathrm{ExpTime}}\xspace}

\definecolor{myred}{HTML}{E8584B}
\definecolor{mypurple}{HTML}{663B91}
\definecolor{myorange}{HTML}{DE9502}
\definecolor{mybrown}{HTML}{3D2901} 
\definecolor{mydarkblue}{HTML}{0402A2}
\definecolor{mylightblue}{HTML}{9C9AFE}
\definecolor{mylightgreen}{HTML}{7CCB95}
\definecolor{mylightred}{HTML}{EC796F}
\newcommand{\mypurple}[1]{{\color{mypurple}#1}} 
\newcommand{\myorange}[1]{{\color{myorange}#1}} 
\newcommand{\mygreen}[1]{{\color{ForestGreen}#1}}
\newcommand{\myblue}[1]{{\color{RoyalBlue}#1}}

\DeclareMathAlphabet{\mathcal}{OMS}{cmsy}{m}{n}

\begin{document}

\maketitle

\begin{abstract}
Aiming to harmonise finite and infinite model reasoning, we initiate the study of \emph{partially finite models}, where the reasoning task comes with a formula that specifies a part of the model that must be finite. We focus on the problem of \emph{partially finite query entailment} in description logics (DLs): given a~knowledge base (KB), a~query, and a~distinguished concept, decide whether the query holds in all models of the KB that interpret the distinguished concept as a finite set. To break the ground, we work with the DL $\Slogic$, an extension of the basic DL $\mathcal{ALC}$ with transitive roles, which is one of the simplest cases where finite and infinite query entailment diverge. Generalising previous results on the finite and infinite cases,  we show that also partially finite entailment of conjunctive queries is in 2\=/\exptime 
for $\Slogic$. The solution involves sophisticated infinite model surgery and goes far beyond combining the arguments for the two special cases. 
As a direct application, we show how the problem of query containment in the presence of closed predicates can be solved by reduction to partially finite query entailment.
\end{abstract}

\section{Introduction}
\label{sec:intro}





Query answering is one of the fundamental tasks in knowledge\=/base related reasoning, where a~query language -- often, \emph{conjunctive queries} -- is enhanced by an ontology -- a~set of inference rules. The ontology
augments the query by facilitating the unification of heterogeneous data, allowing reconstruction of incomplete data, or allowing domain knowledge injection~\cite{cali-framework-2012,calvanese-ontologies-and-databases}.

Classically, in ontology\=/mediated data access one 
asks for so\=/called \emph{certain answers}, that is, answers valid in every, possibly infinite, model of the knowledge base. This \emph{infinite model reasoning} mode
has been the~subject of a~rich and intense study for many ontology languages, including existential rules and description logics, e.g.~\cite{existential-rules-cali,DBLP:journals/jcss/CalvaneseGLR12,DBLP:conf/sdb/Perez-UrbinaMH08,DBLP:conf/ijcai/AmendolaLMV18}.
In the alternative mode of 
\emph{finite model reasoning}, inspired by applications in database theory, 
only finite models
are considered; see  \cite{Murlak23} for an overview. 

For many ontology and query languages the above modes of query entailment coincide, which is known as  \emph{finite controllability}. Some languages, however, are not finitely controllable even with respect to conjunctive queries: two well-known examples among DLs are $\Slogic$ and $\mathcal{ALCIF}$. For an in-depth study see, e.g.,~\cite{amendola-finite-control,BednarczykK22,figueira-control-20,gogacz-control-17}.

This motivates us to consider a~hybrid mode of \emph{partially finite model reasoning}  that combines the universality of infinite models with the real\=/world restrictions of finite domains.
In this mode, we are additionally provided with a~distinguished concept and consider only models where this  concept is interpreted as finite set of individuals.
This can be seen as a~relaxation of query entailment with  \emph{closed predicates}, see e.g.~\cite{DBLP:journals/lmcs/LutzSW19}, where considered models might extend the interpretation of a closed predicate but must keep it finite. 

\begin{figure*}
\begin{subfigure}{0.3\textwidth}
            \begin{tikzpicture}[scale=0.6, sm/.style={scale=0.7, circle, fill=black, inner sep=0pt,minimum size=5pt, outer sep=2pt}, baseline=(A1.base)]
                \node[label=above:\myorange{$A$}, sm, fill=myorange] (A1) {};
                \node[label=above:\myblue{$B$}, sm, fill=RoyalBlue] (B1) [right of=A1] {};
                \node[label=above:\myorange{$A$}, sm, fill=myorange] (A2) [right of=B1] {};
                \node[label=above:\myblue{$B$}, sm, fill=RoyalBlue] (B2) [right of=A2] {};
                \node[label=above:\myorange{$A$}, sm, fill=myorange] (A3) [right of=B2] {};
                \node[right of=A3] (next) {$\cdots$};
                
                \node[label=below:\mygreen{$F$}, sm, fill=ForestGreen] (DA1) [below of=A1] {};
                \node[label=below:\mygreen{$F$}, sm, fill=ForestGreen] (DB1) [right of=DA1] {};
                \node[label=below:\mygreen{$F$}, sm, fill=ForestGreen] (DA2) [right of=DB1] {};
                \node[label=below:\mygreen{$F$}, sm, fill=ForestGreen] (DB2) [right of=DA2] {};
                \node[label=below:\mygreen{$F$}, sm, fill=ForestGreen] (DA3) [right of=DB2] {};
                \node[right of=DA3] {$\cdots$};
                
                \draw[->] (A1) -- (B1);
                \draw[->] (B1) -- (A2);
                \draw[->] (A2) -- (B2);
                \draw[->] (B2) -- (A3);
                \draw[->] (A3) -- (next);
                \draw[->] (A1) -- (DA1);
                \draw[->] (B1) -- (DB1);
                \draw[->] (A2) -- (DA2);
                \draw[->] (B2) -- (DB2);
                \draw[->] (A3) -- (DA3);
            \end{tikzpicture}
    \caption{Unrestricted (infinitely many $F$ nodes).}
    \label{fig:ex-1-unrestricted-model}
\end{subfigure} \hfill
\begin{subfigure}{0.3\textwidth}
            \begin{tikzpicture}[scale=0.6, sm/.style={scale=0.7, circle, fill=black, inner sep=0pt,minimum size=5pt, outer sep=2pt}, baseline=(A1.base)]
                \node[label=above:\myorange{$A$}, sm, fill=myorange] (A1) {};
                \node[label=above:\myblue{$B$}, sm, fill=RoyalBlue] (B1) [right of=A1] {};
                \node[label=above:\myorange{$A$}, sm, fill=myorange] (A2) [right of=B1] {};
                \node[label=above:\myblue{$B$}, sm, fill=RoyalBlue] (B2) [right of=A2] {};
                \node[label=above:\myorange{$A$}, sm, fill=myorange] (A3) [right of=B2] {};
                \node[right of=A3] (next) {$\cdots$};
                
                \node[label=below:\mygreen{$F$}, sm, fill=ForestGreen] (D) [below of=A2] {};
                
                \draw[->] (A1) -- (B1);
                \draw[->] (B1) -- (A2);
                \draw[->] (A2) -- (B2);
                \draw[->] (B2) -- (A3);
                \draw[->] (A3) -- (next);
                \draw[->] (A1) -- (D);
                \draw[->] (B1) -- (D);
                \draw[->] (A2) -- (D);
                \draw[->] (B2) -- (D);
                \draw[->] (A3) -- (D);
            \end{tikzpicture}
    \caption{Partially finite (single $F$ node).}
    \label{fig:ex-1-mixed-model-single-D}
\end{subfigure} \hfill
\begin{subfigure}{0.3\textwidth}
            \begin{tikzpicture}[scale=0.6, sm/.style={scale=0.7, circle, fill=black, inner sep=0pt,minimum size=5pt, outer sep=2pt}, baseline=(A1.base)]
                \node[label=above:\myorange{$A$}, sm, fill=myorange] (A1) {};
                \node[label=above:\myblue{$B$}, sm, fill=RoyalBlue] (B1) [right of=A1] {};
                \node[label=above:\myorange{$A$}, sm, fill=myorange] (A2) [right of=B1] {};
                \node[label=above:\myblue{$B$}, sm, fill=RoyalBlue] (B2) [right of=A2] {};
                \node[label=above:\myorange{$A$}, sm, fill=myorange] (A3) [right of=B2] {};
                \node[right of=A3] (next) {$\cdots$};
                
                \node[label=below:\mygreen{$F$}, sm, fill=ForestGreen] (D1) [below right of=B1, node distance=1.2cm] {};
                \node[label=below:\mygreen{$F$}, sm, fill=ForestGreen] (D2) [below right of=A2, node distance=1.2cm] {};
                
                \draw[->] (A1) -- (B1);
                \draw[->] (B1) -- (A2);
                \draw[->] (A2) -- (B2);
                \draw[->] (B2) -- (A3);
                \draw[->] (A3) -- (next);
                \draw[->] (A1) -- (D1);
                \draw[->] (B1) -- (D2);
                \draw[->] (A2) -- (D1);
                \draw[->] (B2) -- (D2);
                \draw[->] (A3) -- (D1);
            \end{tikzpicture}
    \caption{Partially finite (two $F$ nodes).}
    \label{fig:ex-1-mixed-model-multiple-D}
\end{subfigure}
\caption{Unrestricted and partially finite interpretations.}
\end{figure*}
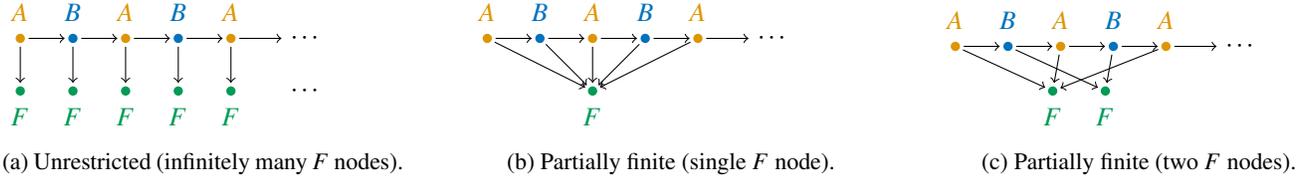

\paragraph{No universal models.}
Consider the knowledge base consisting of the ABox $\{A(a)\}$ and
the TBox 
\[\{A  \sqsubseteq \exists r.B,~ B \sqsubseteq \exists r.A, ~ A \sqcup B \sqsubseteq \exists r.F\}\,.\]
In every model of this knowledge base, there is an infinite chain of alternating $A$\=/, and $B$\=/labelled nodes,  each linked to an~$F$\=/labelled node. 
A~natural infinite  model for this TBox is represented in Figure~\ref{fig:ex-1-unrestricted-model}. In the partially finite reasoning mode, however, we might be required to ensure that there are only finitely many $F$-labelled nodes. There are several ways to make the model conform to this new requirement. The most obvious choice, shown in Figure~\ref{fig:ex-1-mixed-model-single-D}, is to merge all $F$-labelled nodes into one and redirect all edges to this single representative. 
This, however,  might not always be a~good solution. Consider the query 
\[
\exists x\, \exists y\, \exists z \; r(x,y) \wedge r(x,z) \wedge r(y,z)\,.
\]
which checks if some node and its successor share a common successor.
The  model in Figure~\ref{fig:ex-1-unrestricted-model} shows that this query is not entailed over  infinite models. But the model in 
Figure~\ref{fig:ex-1-mixed-model-single-D} does satisfy the query. Is there one that does not? Instead of merging all $F$\=/nodes in the model from Figure~\ref{fig:ex-1-unrestricted-model} into one, separately merge those connected to an $A$-labelled node and those connected to a $B$-labelled node, as shown in Figure~\ref{fig:ex-1-mixed-model-multiple-D}. The resulting  model does not satisfy the query. Hence, we can conclude that the query is not entailed. 

The partially finite model in Figure~\ref{fig:ex-1-mixed-model-multiple-D} is relatively simple, but it is tailored to a specific query -- unlike the unrestricted model in  Figure~\ref{fig:ex-1-unrestricted-model}, which can serve as a universal countermodel for all conjunctive queries that are not entailed. For \[
\exists x\, \exists y\, \exists z\,  \exists w \;  r(x,y) \wedge r(y,z) \wedge r(x,w) \wedge r(z,w)
\] we would need a different partially finite countermodel. And indeed, no partially finite model can serve as a counter-model for all conjunctive queries that are not entailed.

\paragraph{Contribution.}
We begin the investigation of partially finite reasoning 
by inspecting the problem of partially finite entailment for the logic $\Slogic$, 
a~description logic that extends $\mathcal{ALC}$ with role transitivity. As a~relatively small and syntactically simple logic without finite controllability, it is a~prime candidate for this initial study.  As our main contribution, we show that for $\Slogic$ partially finite entailment of conjunctive queries is 2\=/\exptime-complete, matching the complexities of both finite~\cite{garcia-finite-s-arxiv} and infinite entailment~\cite{eiter-unrestricted-s,shiq-infinite-s,shoq-infinite-s}. Additionally, to showcase the usefulness of partially finite entailment, we apply it to solve query containment in the presence of closed predicates.


We start the presentation with the necessary preliminaries in Section~\ref{sec:prelim}. In Section~\ref{sec:single-role}, we describe a~solution for a~single transitive role. In Section~\ref{sec:PartiallyfiniteTheorem}, we develop a model-theoretic tool that allows us to solve the general problem in  Section~\ref{sec:multiple-roles}. In Section~\ref{sec:application-to-closed-predicates}, we discuss the connection to  query containment with closed predicates. Finally, we conclude with a discussion of potential future work  in Section~\ref{sec:ending}. An~appendix with full proofs can be found in the supplementary material. 

\paragraph{Related work.}

In the 1980's and 1990's, a~line of work investigated \emph{circumscription}~\cite{mccarthy1980circumscription,DBLP:journals/jair/BonattiLW09}, and \emph{negation as failure rules}~\cite{10.1016/0004-3702(86)90001-9,DBLP:journals/tocl/DoniniNR02,DBLP:journals/ngc/Katsuno90} as alternatives to the big divide of Closed vs Open World Assumptions. In a~\emph{circumscribed knowledge base}, a~set of predicates (or concepts) must be minimal (but not necessarily finite), whereas in a~language with \emph{negation as failure rules}, a~set of predicates can be marked to be interpreted as false unless provably true. \citeauthor{gelfond1989relationship} (\citeyear{gelfond1989relationship}), show that these two approaches can be made equivalent under some natural constraints. More recently, combinations of the classical DL open-world semantics and closed-world rules have been considered \cite{DBLP:conf/ijcai/Lifschitz91,DBLP:journals/jacm/MotikR10,DBLP:journals/ai/KnorrAH11}: unlike partially finite entailment, they suffer from increased  complexity or even undecidability. 

The complexity of query answering in the presence of \emph{closed predicates}, has been studied in the context of description logics, ranging from DL-Lite and $\mathcal{EL}$ \cite{DBLP:journals/lmcs/LutzSW19,DBLP:conf/aaai/AndreselOS20,DBLP:conf/kr/NgoOS16}, to very expressive DLs like $\mathcal{ALCHI}$ \cite{DBLP:conf/aaai/BajraktariOS18,DBLP:journals/lmcs/LutzSW19} and $\mathcal{ALCHQIO}$ \cite{DBLP:conf/amw/AhmetajOS16,DBLP:conf/aaai/LukumbuzyaOS20}, as well as database\=/flavoured existential rules \cite{DBLP:conf/kr/GogaczLOS20,DBLP:conf/aaai/0002025,benedikt2018pspace,DBLP:conf/ijcai/BienvenuB19}. 

The idea of restricting a~part of the model to be finite also appears in practice, notably in the SUMA system \cite{DBLP:journals/dase/QinZYWFX21}, where only the parts of the model which are relevant for querying are materialised in~memory. Experiments on query answering show that SUMA performs comparably or~better than similar systems, which corresponds to our complexity upper\=/bounds.

\section{Preliminaries}
\label{sec:prelim}

We fix disjoint infinite countable sets \ConNameSet of concept names, 
\RolNameSet of role names, \IndNameSet of individuals, and 
\VarSet of variables.

\paragraph{The description logic \Slogic.}

We consider the description logic  {\Slogic} which extends $\mathcal{ALC}$ with role transitivity. We assume that $\RolNameSet$ is split into the sets $\TRolNameSet$ and $\NTRolNameSet$ of \emph{transitive} and \emph{non-transitive role names}. \Slogic \emph{concepts}, usually denoted by $C$ or $D$, are then defined just like for $\mathcal{ALC}$ by the grammar
\[
C, D := \bot ~|~ \top ~|~ A ~|~ \neg C ~|~ C \sqcap D
~|~ C \sqcup D ~|~ \exists r.\,C ~|~ \forall r.\,C 
\]
where $A$ ranges over {concept names} from \ConNameSet and $r$ ranges over {role names} from \RolNameSet. 
A \emph{concept inclusion} is a~formula of the shape $C \sqsubseteq D$ where $C$ and $D$ are \Slogic~concepts. 
%
%
A  {TBox} is a~finite set of  {concept inclusions}.

Without loss of generality we can assume that TBoxes are normalised and only use concept inclusion of the forms
\[ \textstyle
\bigsqcap_i A_i \sqsubseteq \bigsqcup_j B_j\,, \qquad 
A \sqsubseteq \forall r.B\,, \qquad
A \sqsubseteq \exists r.B\,,
\]
with empty conjunction and disjunction treated as $\top$ and $\bot$  \cite{DBLP:journals/ai/Gutierrez-Basulto23}. 

An  {ABox} is a~finite set of  {concept assertions} of the form $A(a)$ and  {role assertions} of the form $r(a,b)$ where $A$ is a~{concept name}, $r$ is a~{role name} and $\{a,b\}$ are  {individuals}.

\paragraph{Interpretations.}

An interpretation \intp consists of a~(possibly infinite)  {domain} $\Delta^{\intp}$ and  a {function} $\cdot^{\intp}$ that maps each concept name $A$ to  $A^\intp \subseteq \Delta^{\intp}$, and each role name $r$ to $r^\intp \subseteq \Delta^{\intp}\times\Delta^{\intp}$. We extend $\cdot^{\intp}$ to $\Slogic$ concepts as follows:
\begin{gather*}
\bot^\intp = \emptyset\:, \quad 
\top^\intp = \Delta^{\intp}\,,\quad
(\neg C)^\intp = \Delta ^{\intp} \setminus C^{\intp}\,,\\ 
(C \sqcap D)^\intp = C^{\intp} \cap D^{\intp},\quad (C \sqcup D)^\intp = C^{\intp} \cup D^{\intp}\,,\\
(\exists r.C)^\intp = \{a \in \Delta^{\intp} ~|~ \exists b.\  (a,b) \in r^{\intp} \land b \in C^{\intp}\}\,,\\
(\forall r.C)^\intp = \{a \in \Delta^{\intp} ~|~ \forall b.\  (a,b) \in r^{\intp} \implies b \in C^{\intp}\}\,.
\end{gather*}

We call a role $r \subseteq \Delta \times \Delta$ \emph{transitive} if for all $a, b, c \in \Delta$, $(a,b)\in r$ and $(b,c)\in r$ implies $(a,c)\in r$. 
The \emph{transitive closure} of $r\subseteq \Delta\times\Delta$ (with respect to $\Delta$), written as $r^*$ is the least transitive $r' \subseteq \Delta \times \Delta$ such that $r \subseteq r'$.
An interpretation $\intp$ is \emph{transitive} if $t^\intp$ is transitive for all $t\in\TRolNameSet$.
The \emph{transitive closure} of an interpretation $\intp$, written as $\intp^*$ is the interpretation such that $\Delta^{\intp^*} = \Delta^\intp$, $A^{\intp^*} = A^\intp$,  $r^{\intp^*} = r^\intp$, and $t^{\intp^*} = \big(t^\intp\big)^*$ for all $A\in\ConNameSet$, $r\in\NTRolNameSet$, and $t\in\TRolNameSet$. 

An interpretation \intp  \emph{satisfies}: 
a~concept inclusion $C \sqsubseteq D$ if $C^{\intp} \subseteq D^{\intp}$; 
a~{concept assertion} $A(a)$ if $a \in A^{\intp}$; 
a~{role assertion} $r(a,b)$ if $(a,b) \in r^{\intp}$; 
a~TBox $\TBox$, written as $\intp\models \TBox$, if it is transitive and satisfies each concept inclusion in $\TBox$; 
and an ABox $\ABox$, written as $\intp\models\ABox$, if it satisfies each assertion in $\ABox$. (Note that we adopt the Standard Name Assumption; in particular,  $\Delta^\intp$ includes all individuals used in $\ABox$.)

An {interpretation} $\intp$ is a~\emph{subinterpretation} of an~{interpretation} $\intb$, written as  $\intp \subseteq \intb$, if $\Delta^{\intp} \subseteq \Delta^{\intb}$ and $A^{\intp} \subseteq A^{\intb}, r^{\intp} \subseteq r^{\intb}$ for all $A \in  \ConNameSet$, and $r \in  \RolNameSet$. A subinterpretation $\intp$ of $\intb$ is \emph{induced} by $\Delta\subseteq\Delta^{\intb}$, written as $\intp = \intb \restr \Delta$, if $\Delta^{\intp} = \Delta$, $A^{\intp} = A^{\intb} \cap \Delta$ and $r^{\intp} = r^{\intb} \cap \Delta \times \Delta$ for all $A \in \ConNameSet$, and $r \in  \RolNameSet$.  For $\Sigma \subseteq \RolNameSet$, an interpretation \intp is \emph{over $\Sigma$} if  $r^\intp = \emptyset$ for all  $r\notin\Sigma$. If $\Sigma = \{r\}$ we say simply that $\intp$ is over $r$. 
A~concept name $A$ \emph{occurs} in  $\intp$ if $A^\intp \neq \emptyset$. We consider only interpretations in which finitely many concept names occur.

A homomorphism $h$ from $\intp$ to $\intb$, written as $h:\intp\to\intb$ is a~function from $\Delta^\intp$ to $\Delta^{\intb}$ such that $e \in A^\intp$ implies $h(e) \in A^{\intb}$, $(e,e') \in r^\intp$ implies $(h(e), h(e')) \in r^{\intb}$, and $h(a) = a$ for all $e,e'\in\Delta^\intp$, 
$A\in\ConNameSet$, 
$r\in\RolNameSet$, 
$a\in\IndNameSet$. By $\range(h)\subseteq\Delta^{\intb}$ we denote the set of values of $h$, i.e.,~$\{h(e)\mid e\in\Delta^\intp\}$.

An interpretation $\intp$ can be seen as directed multigraph $G_\intp$ whose nodes are the elements of $\intp$ and edges are obtained as the disjoint union of the interpretations of role names in $\intp$. Nodes of $G_\intp$ are labelled with sets of concept names and edges are labelled with role names (note that parallel edges have different labels). Throughout the paper we apply the standard graph-theoretic terminology directly to interpretations. In particular, we will speak of reachability and \emph{strongly\=/connected components} (SCCs) of interpretations, which are maximal sets of elements in which every element is reachable from every other element.

\paragraph{Knowledge Bases.}


A~{knowledge base} (KB) $\KB$ is a~pair $(\TBox, \ABox)$ composed of a~{TBox} $\TBox$ and an {ABox} $\ABox$. An  {interpretation} \intp is a~\emph{model} of a~{KB} $\KB = (\TBox, \ABox)$, written as $\intp \models \KB$, if $\intp\models\TBox$, $\intp\models\ABox$, and $\intp$ is transitive.

We write $\|\KB\|$ for the total size of  concept inclusions, transitivity declarations, and assertions in $\KB$. By ${\conceptsIn{\KB}}$, $ {\rolesIn{\KB}}$
, and $ {\indsIn{\KB}}$, we denote the sets of all  {concept names},  {role names}
~and  {individuals} that appear in \KB. $\KB$ is \emph{over $\Sigma \subseteq \RolNameSet$} if $\rolesIn{\KB} \subseteq \Sigma$. We extend this notation and terminology in the natural way to  {TBoxes} and  {ABoxes}.


A  \emph{unary type} is a~set of {concept names}. A unary type over $\Gamma \subseteq \ConNameSet$ is a~set of concept names from $\Gamma$.  For a~{KB} \KB, we write $ {\TypesSet{\KB}}$ for the set of all unary types over $\conceptsIn{\KB}$. We write $\types{\intp}{e}$ for the \emph{type of element $e$ in interpretation $\intp$}, defined as  $\types{\intp}{e} = \{A\in\ConNameSet \bigm| e \in A^\intp\}$. 

\paragraph{Conjunctive Queries.}

We consider \emph{conjunctive queries} (CQs), i.e.~queries of the form $\exists \bar{x}. ~ q_1(\bar{y}_1) \wedge q_2(\bar{y}_2)\wedge \ldots\wedge q_n(\bar{y}_n)$ where $\bar{x}$ is a~tuple of variables, $\bar{y}_1,\bar{y}_2,\dots \bar{y}_n$ are tuples of  {variables} from $\bar x$, and each $q_i$ is either a~unary atom  of the form $A(x)$ for some  $A\in\ConNameSet$ and $x\in\VarSet$  or a~binary atom of the form $r(x,x')$ for some $r\in\RolNameSet$ and $x, x'\in\VarSet$.

A  \emph{match} of a~CQ $q$ in an {interpretation} \intp is a~function $\pi$ that maps each  {variable} in $q$ to an element of $\Delta^{\intp}$ such that $\pi(x) \in A^{\intp}$ for each atom in $q$ of the shape $A(x)$ and $(\pi(x), \pi(x')) \in r^{\intp}$ for each atom in $q$ of the shape $r(x, x')$. We say that an  {interpretation} $\intp$ {satisfies} a~ CQ $q$, written as $\intp \models q$, if there is a {match}  of $q$ in $\intp$. 

A \emph{union of conjunctive queries} (UCQ) is a finite set $\query$ of CQs. An interpretation $\intp$ satisfies $\query$, written as $\intp\models \query$,  if  $\intp\models q$ for some $q\in \query$. It is well known that if $\intp\models\query$ and there is a homomorphism from $\intp$ to $\JJ$, then $\JJ\models\query$. 

A query $Q$ is \emph{over $\Sigma \subseteq \RolNameSet$} if uses role names only from $\Sigma$.


\paragraph{Partially Finite Entailment.}

A {KB} $\KB$ \emph{entails} a~query $\query$, written as $\KB \models \query$, if every  {model} of $\KB$ {satisfies} $\query$.  A  \emph{counter-model} is a~{model} of $\KB$ that does not  {satisfy} $\query$. Similarly, $\KB$ \emph{finitely entails} $\query$, written as $\KB \models_{\mathit{fin}} Q$, if every \emph{finite} {model} of $\KB$  {satisfies} $\query$.
We focus on a~common generalisation of  these problems, where the implication must hold for models where a~specified concept is finite.  Given a~concept $\Fin$, we say that \KB \emph{$\Fin$-entails} \query, written 
$\KB \models_\Fin Q$,
if every model \intp of \KB with $\Fin^{\intp}$ finite, satisfies \query. For simplicity, we assume $F$ is a~concept name.

\smallskip

\noindent \fbox{
  \begin{tabular}{r l}     \multicolumn{2}{c}{\large \textsc{Partially finite entailment}}  \\
     \textbf{Input:} & A  {KB} \KB, a~concept name $\Fin$, and a~ {query} \query \\
     \textbf{Output:} & True iff $\KB  \models_\Fin \query$
\end{tabular}}

\smallskip

Throughout the paper we assume that  $F$ is the distinguished concept to be interpreted as a finite set. We call an element $e$ of interpretation $\intp$ \emph{critical} if $e\in F^\intp$.


\section{Single Transitive Role}
\label{sec:single-role}

We first deal with the case of a~single transitive role $t$: throughout the section we  assume that the interpretations of all remaining role names are empty. To make this subcase easier to use in the general solution, we work with rooted and labelled interpretations. 

An \emph{$X$\=/labelled interpretation} $(\intp,\lambda)$ is an interpretation $\intp$ along with a~labelling function $\lambda :\Delta^\intp\to X$ that assigns to each element of $\intp$ a~label from $X$. The set $X$ will vary, depending on the context. A homomorphism $h\from (\intp, \lambda) \to (\intp', \lambda')$ is a~homomorphism $h\from\intp \to \intp'$ such that $\lambda(a) = \lambda'(h(a))$ for all $a \in \Delta^\intp$.

A \emph{rooted interpretation} $(\intp, a)$ is an interpretation $\intp$ along with a~distinguished element $a\in\Delta^\intp$ called the \emph{root}. A \emph{homomorphism} $h\from(\intp,a) \to (\intp', a')$ is a~homomorphism $h\from\intp \to \intp'$ such that $h(a)=a'$.

A \emph{rooted $X$\=/labelled interpretation} is $(\intp, a, \lambda)$ where $(\intp, a)$ is a~rooted interpretation and $(\intp, \lambda)$ is an $X$\=/labelled interpretation. A homomorphism $h\from(\intp, a, \lambda)\to (\intp', a', \lambda')$ is a~homomorphism $h\from\intp \to \intp'$ such that both $h\from(\intp, a)\to (\intp', a')$ and $h\from(\intp, \lambda)\to (\intp',\lambda')$ are homomorphisms.

This additional labelling in $X$ will be used inside our inductive bottom\=/up construction to store the already\=/simplified parts of our interpretation as external values (labels), allowing us to focus on the currently modified part.

Intuitively, we would like to replace a~given (rooted $X$\=/labelled) interpretation with a~simpler one, but preserving certain properties. Eventually, we will care about satisfying a~given KB and not satisfying a~given query. For now, we use an abstract sufficient condition, formulated purely in terms of interpretations.

For an interpretation $\intp$, an element $u \in \Delta^\intp$, and a~transitive role name $t$, we write $\rch{t}{\intp}(u)$ for the set of concept names reachable from $u$, i.e. all $A$ such that $A^\intp\setminus \{u\}$ contains a~node $t$-reachable from $u$. 

A homomorphism $h\from\intp \to \intp'$ is \emph{$t$\=/strong} 
if \[\types{\intp}{a} = \types{\intp'}{h(a)} \quad  \text{and} \quad \rch{t}{\intp}(a) = \rch{t}{\intp'}(h(a))\] for all $a\in\Delta^\intp$. Clearly, the composition of two $t$\=/strong homomorphisms is a~$t$\=/strong homomorphism.

\begin{restatable}{lmm}{lemStrong}
\label{lem:strong}
Consider a~TBox $\TBox$ over $t$, a~UCQ $Q$, and transitive interpretations $\intp$ and $\intp'$ such that there is a~$t$\=/strong homomorphism $h\from \intp \to \intp'$. 
\begin{itemize}
\item If  $\intp'\models \TBox$ and $\intp'\not\models Q$,  then $\intp\models \TBox$ and $\intp\not\models Q$.
\item If $h$ is surjective and $\intp\models \TBox$, then $\intp' \models \TBox$.
\end{itemize}
\end{restatable}

The notion of $t$\=/strong homomorphism lifts naturally to rooted and $X$\=/labelled interpretations: we call a~homomorphism $h\from (\intp, a, \lambda) \to (\intp', a', \lambda')$  $t$\=/strong if the underlying homomorphism from $\intp$ to $\intp'$ is $t$\=/strong.

Our goal is the following:
for a~transitive rooted $X$\=/labelled interpretation $(\intp,a, \lambda)$ over $t$, construct a~rooted $X$\=/labelled interpretation over $t$ that has a~small finite representation
and maps to $ (\intp,a, \lambda)$ via a~$t$\=/strong homomorphism.

Next, we make this goal more concrete by explaining how we represent infinite interpretations using finite ones.  

\subsection{Quasi-unravelling}

We modify the standard unravelling to keep the set of critical elements finite: the \emph{quasi\=/unravelling} procedure builds a~tree-like model by making a separate copy of an element for each path leading to this element, except for critical elements, of which only one copy is kept (see Figure~\ref{fig:quasiunravelling}).

Let \intp be an interpretation over a~single transitive role name $t$. We refer to SCCs of $\intp$ as \emph{clusters}. Notice that since we have a~single transitive role, a~cluster is either a~single element (with or without a~loop) or a~clique. We call a~cluster in \intp \emph{critical} if it contains a~critical element.
We write  $\Delta^\intp_F$ for the union of all critical clusters in interpretation \intp. For an element $a \in \Delta^\intp$, we let $K^\intp_a$ be the cluster in $\intp$ that contains $a$ if this cluster is critical, and $K^\intp_a=\{a\}$ otherwise.

\begin{figure}
    \begin{tikzpicture}[scale=0.25, sm/.style={scale=0.5, circle, fill=black, inner sep=0pt,minimum size=5pt, outer sep=2pt}]
        \node[label=above:\myorange{$A$}, sm, fill=myorange] (A1) {};
        \node[label=above:\myblue{$B$}, sm, fill=RoyalBlue] (B1) [above right of=A1] {};
        \node[label=below:\myblue{$B$}, sm, fill=RoyalBlue] (B2) [below right of=A1] {};
        \node[label=above:\mypurple{$C$}, sm, fill=mypurple] (C1) [below right of=B1] {};
        \node[label=right:\mygreen{$F$}, sm, fill=ForestGreen] (F1) [right of=C1] {};
        \node[label=above:\myblue{$B$}, sm, fill=RoyalBlue] (B3) [above right of=F1, node distance=1.5cm] {};
        \node[label=below:\mypurple{$C$}, sm, fill=mypurple] (C2) [below right of=F1, node distance=1.5cm] {};
        \node[label=above:\myorange{$A$}, sm, fill=myorange] (A2) [right of=F1, node distance=2.2cm] {};

        \draw[->] (A1) -- (B1);
        \draw[->] (A1) -- (B2);
        \draw[->] (B1) -- (C1);
        \draw[->] (B2) -- (C1);
        \draw[->] (C1) -- (F1);
        \draw[->] (F1) -- (B3);
        \draw[->] (F1) -- (C2);
        \draw[->] (B3) -- (A2);
        \draw[->] (C2) -- (A2); 

        \node (unrav) [right of=F1, node distance=2.5cm] {\Large $\xrightarrow{\text{quasi unrav.}}$};

        \node[label=above:\myorange{$A$}, sm, fill=myorange] (A1u) [right of= unrav, node distance=3cm] {};
        \node[label=above:\myblue{$B$}, sm, fill=RoyalBlue] (B1u) [above right of=A1u] {};
        \node[label=below:\myblue{$B$}, sm, fill=RoyalBlue] (B2u) [below right of=A1u] {};
        \node[label=above:\mypurple{$C$}, sm, fill=mypurple] (C1u) [right of=B1u] {};
        \node[label=below:\mypurple{$C$}, sm, fill=mypurple] (C2u) [right of=B2u] {};
        \node[label=right:\mygreen{$F$}, sm, fill=ForestGreen] (F1u) [below right of=C1u] {};
        \node[label=above:\myblue{$B$}, sm, fill=RoyalBlue] (B3u) [above right of=F1u, node distance=1.5cm] {};
        \node[label=below:\mypurple{$C$}, sm, fill=mypurple] (C3u) [below right of=F1u, node distance=1.5cm] {};
        \node[label=above:\myorange{$A$}, sm, fill=myorange] (A2u) [right of=B3u] {};
        \node[label=below:\myorange{$A$}, sm, fill=myorange] (A3u) [right of=C3u] {};

        \draw[->] (A1u) -- (B1u);
        \draw[->] (A1u) -- (B2u);
        \draw[->] (B1u) -- (C1u);
        \draw[->] (B2u) -- (C2u);
        \draw[->] (C1u) -- (F1u);
        \draw[->] (C2u) -- (F1u);
        \draw[->] (F1u) -- (B3u);
        \draw[->] (F1u) -- (C3u);
        \draw[->] (B3u) -- (A2u);
        \draw[->] (C3u) -- (A3u);
    \end{tikzpicture}
    \caption{An interpretation and its quasi-unravelling.}
    \label{fig:quasiunravelling}
\end{figure}
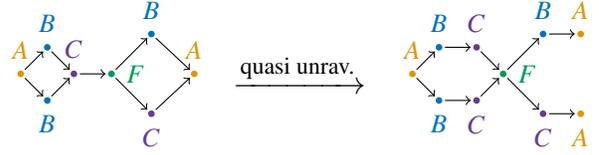
 
\begin{definition}[Quasi\=/unravelling] The quasi\=/unravelling of a~rooted interpretation $(\intp, a)$ over a~single transitive role $t$ is the rooted interpretation $(\widetilde\intp,a)$ over $t$ obtained as follows. The domain of $\widetilde\intp$ is the union of  $\Delta_F^\intp\cup\{a\}$ and the set of pairs of the form $(K,p)$ where $K$ is either $K^\intp_a$ or a~critical cluster of \intp, and $p$ is a~non\=/empty path in \intp that begins in a~successor of an element of $K$ and visits  
elements from $\Delta^\intp \setminus (\Delta_F^\intp\cup\{a\})$ only. For a~concept name $A$, we let  \[A^{\widetilde\intp}\ = \ \Big(A^\intp \cap\, \big(\Delta_F^\intp \cup \{a\}\big)\Big) \ \cup \ \big\{ (K,pu) \in \Delta^{\widetilde \intp} \bigm| u\in A^\intp \big\}\,.\] Finally, $t^{\widetilde\intp}$ is the transitive closure of 
\begin{gather*}
    \Big(t^\intp  \, \cap \:   \big(\big(\Delta_F^\intp \cup \{a\}\big)\times  \Delta_F^\intp  \big)\Big)\  \cup \\
    \cup \ \big\{\big(u,(K,p)\big) \bigm| (K,p)\in\Delta^{\widetilde\intp}, u \in K\big\} \ \cup \\
 \cup  \ \big\{\big((K,p), (K,pu)\big) \bigm| (K,p), (K,pu)\in\Delta^{\widetilde\intp}\big\}\ \cup \\
\cup \ \big\{\big( (K,pu), v \big) \bigm| (K,pu)\in\Delta^{\widetilde\intp}, (u,v) \in t^\intp, v\in\Delta_F^\intp  \big\} \,.
\end{gather*}
\end{definition}


The reader might be familiar with a different style of defining unravellings where copies of elements from $(\intp,a)$  are added successively to $(\widetilde\intp,a)$. In our definition, elements from $\Delta^\intp_F \cup\{a\}$ are copies of themselves, and each $(K,p)$ is a copy of the last element in path $p$.

The quasi\=/unravelling $(\widetilde \intp, a, \widetilde \lambda)$ of a~rooted $X$\=/labelled interpretation $(\intp, a, \lambda)$ is the quasi\=/unravelling $(\widetilde \intp, a)$ of $(\intp,a)$  with the labelling function $\widetilde\lambda$ defined as $\widetilde\lambda(u) = \lambda(u)$ for $u\in\Delta^\intp_F\cup\{a\}$ and 
$\widetilde\lambda((K,pu)) = \lambda(u)$ for $(K,pu) \in \Delta^{\widetilde\intp} \setminus \big(\Delta^\intp_F\cup\{u\}\big)$. 

Quasi-unravelling preserves the satisfaction of TBox constraints and conjunctive queries. Indeed, mapping each copy of an element of $\intp$ to its original gives a $t$-strong homomorphism from $\widetilde\intp$ to $\intp^*$.

\begin{restatable}{lmm}{lemUnravelling}
\label{lem:unravelling}
If $(\widetilde \intp, a, \widetilde \lambda)$ is the quasi\=/unravelling of $(\intp, a, \lambda)$ where $\intp$ is an interpretation over $t$ then there is a~$t$\=/strong homomorphism $h\from (\widetilde \intp, a, \widetilde \lambda) \to (\intp^*, a, \lambda)$. Moreover, if each element in $\intp$ is reachable from $a$, we can additionally require $h$ to be surjective.
\end{restatable}

From Lemmas~\ref{lem:strong} and~\ref{lem:unravelling} we get the following. 

\begin{restatable}{lemma}{CorUnravelling}
\label{cor:unravelling}    
Consider an interpretation $\intp$ over $t$ and an element $a$ of $\intp$ from which every other element of $\intp$ is reachable. Let
$(\widetilde \intp, a)$ be the quasi\=/unravelling of 
$(\intp, a)$.  Then 
\begin{itemize}
\item for every TBox $\TBox$ over $t$, $\widetilde \intp  \models \TBox$ iff   
$\intp^* \models \TBox$;
\item  for every UCQ $Q$ over $t$, $\widetilde \intp \models Q$ implies $\intp^* \models Q$.
\end{itemize}
\end{restatable}

Lemma~\ref{cor:unravelling} shows that $(\intp, a)$ is a~good representation of $(\widetilde\intp, a)$, but the second property is not sufficient to allow deciding if $\widetilde \intp \models Q$ by just examining $\intp$. Later we will show how to pick $(\intp, a)$  so that $\widetilde \intp \models Q$ iff $\intp \models Q$.
This will involve interpretations of a special shape,  described next.

\subsection{Elementary interpretations}

\begin{figure*}
    \centering
    \colorlet{singleton}{myorange}
\colorlet{cycle}{ForestGreen}
\colorlet{cluster}{RoyalBlue}

\begin{subfigure}{0.3\textwidth}
    \centering
    \begin{tikzpicture}[scale=0.6, sm/.style={scale=0.7, circle, fill=black, inner sep=0pt,minimum size=5pt, outer sep=2pt}]
        \node[label=above:$A$, sm] (singl) at (2.5,2) {};
        \fill [nearly transparent,singleton] (2.5,2) circle [x radius=0.3, y radius=0.3];

        \node[label=left:$A$, sm] (cyc-A) at (4,1) {};
        \node[label=right:$B$, sm] (cyc-B) at ($(cyc-A) + (1,0)$) {};
        \node[label=right:$C$, sm] (cyc-C) at ($(cyc-A) + (1,-1)$) {};
        \node[label=left:$D$, sm] (cyc-D) at ($(cyc-A) + (0,-1)$) {};
        \draw[->] (cyc-A) -- (cyc-B);
        \draw[->] (cyc-B) -- (cyc-C);
        \draw[->] (cyc-C) -- (cyc-D);
        \draw[->] (cyc-D) -- (cyc-A);
        \fill [nearly transparent,cycle] (4.5,0.5) circle [x radius=1, y radius=1];

        \node[label=above:$F$, sm] (cl-E) at (1,1) {};
        \node[label=below:$A$, sm] (cl-A) at ($(cl-E) + (-1,-1)$) {};
        \node[label=below:$B$, sm] (cl-B) at ($(cl-E) + (0, -1)$) {};
        \node[label=below:$C$, sm] (cl-C) at ($(cl-E) + (1,-1)$) {};
        \draw[->] (cl-E) -- (cl-A);
        \draw[->] (cl-E) -- (cl-C);
        \draw[->] (cl-A) -- (cl-B);
        \draw[->] (cl-C) -- (cl-B);
        \draw[->] (cl-B) -- (cl-E);
        \fill [nearly transparent,cluster] (1,0.4) circle [x radius=1.2, y radius=0.8];
        \end{tikzpicture}
    \caption{A singleton (orange), cycle (green) and cluster (blue).}
    \label{fig:elementary:first}
\end{subfigure}
\begin{subfigure}{0.3\textwidth}
    \centering
    \begin{tikzpicture}[scale=0.6, sm/.style={scale=0.7, circle, fill=black, inner sep=0pt,minimum size=5pt, outer sep=2pt}]
        \node[label=left:$A$, sm] (cyc-A) at (1,3.5) {};
        \node[label=left:$B$, sm] (cyc-B) at ($(cyc-A) + (0, -1)$) {};
        \node[label=left:$C$, sm] (cyc-C) at ($(cyc-B) + (0, -1)$) {};
        \node[label=left:$D$, sm] (cyc-D) at ($(cyc-C) + (0, -1)$) {};

        \draw[->] (cyc-A) -- (cyc-B);
        \draw[->] (cyc-B) -- (cyc-C);
        \draw[->] (cyc-C) -- (cyc-D);
        \draw[->] (cyc-D) to[bend left=20] (cyc-A);

        \fill [nearly transparent,cycle] (1,2) circle [x radius=0.5, y radius=1.7];

        \node[label=above:$A$, sm] (out-1) at ($(cyc-A) + (1.5, 0)$) {};
        \draw[->] (cyc-A) -- (out-1);
        \fill [nearly transparent,singleton] (2.5,3.5) circle [x radius=0.3, y radius=0.3];

        \node[label=above:$C$, sm] (out-2-a) at ($(cyc-B) + (1,0)$) {};
        \node[label=above:$D$, sm] (out-2-b) at ($(cyc-B) + (2,0)$) {};
        \node[label=right:$\cdots$] (out-2-contd) at ($(cyc-B) + (2.1,0)$) {};
        \draw[->] (cyc-A) -- (out-2-a);
        \draw[->] (out-2-a) to[bend left] (out-2-b);
        \draw[->] (out-2-b) to[bend left] (out-2-a);
=        \fill [nearly transparent,cycle] (2.5,2.5) circle [x radius=0.7, y radius=0.3];

        \node[label=above:$F$, sm] (out-3) at ($(cyc-C) + (1.5,0)$) {};
        \draw[->] (cyc-A) -- (out-3);
        \fill [nearly transparent,cluster] (2.5,1.5) circle [x radius=0.3, y radius=0.3];

        \node[label=above:$C$, sm] (out-4-a) at ($(cyc-D) + (1,0)$) {};
        \node[label=above:$D$, sm] (out-4-b) at ($(cyc-D) + (2,0)$) {};
        \node[label=right:$\cdots$] (out-4-contd) at ($(cyc-D) + (2.1,0)$) {};
        \draw[->] (cyc-A) -- (out-4-a);
        \draw[->] (out-4-a) to[bend left] (out-4-b);
        \draw[->] (out-4-b) to[bend left] (out-4-a);
        \fill [nearly transparent,cycle] (2.5,0.5) circle [x radius=0.7, y radius=0.3];
    \end{tikzpicture}
    \caption{A loop tree.}
    \label{fig:elementary:second}
\end{subfigure}
\begin{subfigure}{0.3\textwidth}
    \centering
    \begin{tikzpicture}[scale=0.6, sm/.style={scale=0.7, circle, fill=black, inner sep=0pt,minimum size=5pt, outer sep=2pt}]
        \node[label=above:$F$, sm] (cl-E) at (0,6) {};
        \node[label=below:$A$, sm] (cl-A) at ($(cl-E) + (-1,-1)$) {};
        \node[label=below:$B$, sm] (cl-B) at ($(cl-E) + (0, -1)$) {};
        \node[label=below:$C$, sm] (cl-C) at ($(cl-E) + (1,-1)$) {};
        \draw[->] (cl-E) -- (cl-A);
        \draw[->] (cl-E) -- (cl-C);
        \draw[->] (cl-A) -- (cl-B);
        \draw[->] (cl-C) -- (cl-B);
        \draw[->] (cl-B) -- (cl-E);
        \fill [nearly transparent,cluster] (0,5.4) circle [x radius=1.2, y radius=0.8];
    
        \node[label=right:$A$, sm] (cyc-A-left) at (-1,3.5) {};
        \node[label=right:$B$, sm] (cyc-B-left) at ($(cyc-A-left) + (0, -1)$) {};
        \node[label=right:$C$, sm] (cyc-C-left) at ($(cyc-B-left) + (0, -1)$) {};
        \node[label=right:$D$, sm] (cyc-D-left) at ($(cyc-C-left) + (0, -1)$) {};

        \draw[->] (cl-B) -- (cyc-A-left);
        \draw[->] (cyc-A-left) -- (cyc-B-left);
        \draw[->] (cyc-B-left) -- (cyc-C-left);
        \draw[->] (cyc-C-left) -- (cyc-D-left);
        \draw[->] (cyc-D-left) to[bend right=20] (cyc-A-left);

        \fill [nearly transparent,cycle] (-1,2) circle [x radius=0.5, y radius=1.7];

        \node[label=above:$A$, sm] (out-1-left) at ($(cyc-A-left) + (-1.5, 0)$) {};
        \draw[->] (cyc-A-left) -- (out-1-left);
        \fill [nearly transparent,singleton] (-2.5,3.5) circle [x radius=0.3, y radius=0.3];

        \node[label=above:$C$, sm] (out-2-a-left) at ($(cyc-B-left) + (-1,0)$) {};
        \node[label=above:$D$, sm] (out-2-b-left) at ($(cyc-B-left) + (-2,0)$) {};
        \node[label=left:$\cdots$] (out-2-left-contd) at ($(cyc-B-left) + (-2.1,0)$) {};
        \draw[->] (cyc-A-left) -- (out-2-a-left);
        \draw[->] (out-2-a-left) to[bend left] (out-2-b-left);
        \draw[->] (out-2-b-left) to[bend left] (out-2-a-left);
        \fill [nearly transparent,cycle] (-2.5,2.5) circle [x radius=0.7, y radius=0.3];

        \node[label=above:$F$, sm] (out-3-left) at ($(cyc-C-left) + (-1.5,0)$) {};
        \draw[->] (cyc-A-left) -- (out-3-left);
        \fill [nearly transparent,cluster] (-2.5,1.5) circle [x radius=0.3, y radius=0.3];

        \node[label=above:$C$, sm] (out-4-a-left) at ($(cyc-D-left) + (-1,0)$) {};
        \node[label=above:$D$, sm] (out-4-b-left) at ($(cyc-D-left) + (-2,0)$) {};
        \node[label=left:$\cdots$] (out-4-left-contd) at ($(cyc-D-left) + (-2.1,0)$) {};
        \draw[->] (cyc-A-left) -- (out-4-a-left);
        \draw[->] (out-4-a-left) to[bend left] (out-4-b-left);
        \draw[->] (out-4-b-left) to[bend left] (out-4-a-left);
        \fill [nearly transparent,cycle] (-2.5,0.5) circle [x radius=0.7, y radius=0.3];

        \node[label=left:$A$, sm] (cyc-A-right) at (1,3.5) {};
        \node[label=left:$B$, sm] (cyc-B-right) at ($(cyc-A-right) + (0, -1)$) {};
        \node[label=left:$C$, sm] (cyc-C-right) at ($(cyc-B-right) + (0, -1)$) {};
        \node[label=left:$D$, sm] (cyc-D-right) at ($(cyc-C-right) + (0, -1)$) {};

        \draw[->] (cl-B) -- (cyc-A-right);
        \draw[->] (cyc-A-right) -- (cyc-B-right);
        \draw[->] (cyc-B-right) -- (cyc-C-right);
        \draw[->] (cyc-C-right) -- (cyc-D-right);
        \draw[->] (cyc-D-right) to[bend left=20] (cyc-A-right);

        \fill [nearly transparent,cycle] (1,2) circle [x radius=0.5, y radius=1.7];

        \node[label=above:$A$, sm] (out-1-right) at ($(cyc-A-right) + (1.5, 0)$) {};
        \draw[->] (cyc-A-right) -- (out-1-right);
        \fill [nearly transparent,singleton] (2.5,3.5) circle [x radius=0.3, y radius=0.3];

        \node[label=above:$C$, sm] (out-2-a-right) at ($(cyc-B-right) + (1,0)$) {};
        \node[label=above:$D$, sm] (out-2-b-right) at ($(cyc-B-right) + (2,0)$) {};
        \node[label=right:$\cdots$] (out-2-right-contd) at ($(cyc-B-right) + (2.1,0)$) {};
        \draw[->] (cyc-A-right) -- (out-2-a-right);
        \draw[->] (out-2-a-right) to[bend left] (out-2-b-right);
        \draw[->] (out-2-b-right) to[bend left] (out-2-a-right);
        \fill [nearly transparent,cycle] (2.5,2.5) circle [x radius=0.7, y radius=0.3];

        \node[label=above:$F$, sm] (out-3-right) at ($(cyc-C-right) + (1.5,0)$) {};
        \draw[->] (cyc-A-right) -- (out-3-right);
        \fill [nearly transparent,cluster] (2.5,1.5) circle [x radius=0.3, y radius=0.3];

        \node[label=above:$C$, sm] (out-4-a-right) at ($(cyc-D-right) + (1,0)$) {};
        \node[label=above:$D$, sm] (out-4-b-right) at ($(cyc-D-right) + (2,0)$) {};
        \node[label=right:$\cdots$] (out-4-right-contd) at ($(cyc-D-right) + (2.1,0)$) {};
        \draw[->] (cyc-A-right) -- (out-4-a-right);
        \draw[->] (out-4-a-right) to[bend left] (out-4-b-right);
        \draw[->] (out-4-b-right) to[bend left] (out-4-a-right);
        \fill [nearly transparent, cycle] (2.5,0.5) circle [x radius=0.7, y radius=0.3];
    \end{tikzpicture}
    \caption{An elementary interpretation.}
    \label{fig:elementary:third}
\end{subfigure}
    \vspace{-1ex}
    \caption{The three levels of elementary interpretations. Critical elements indicated with concept name $F$.}
    \label{ex:interpretation-def}
\end{figure*}
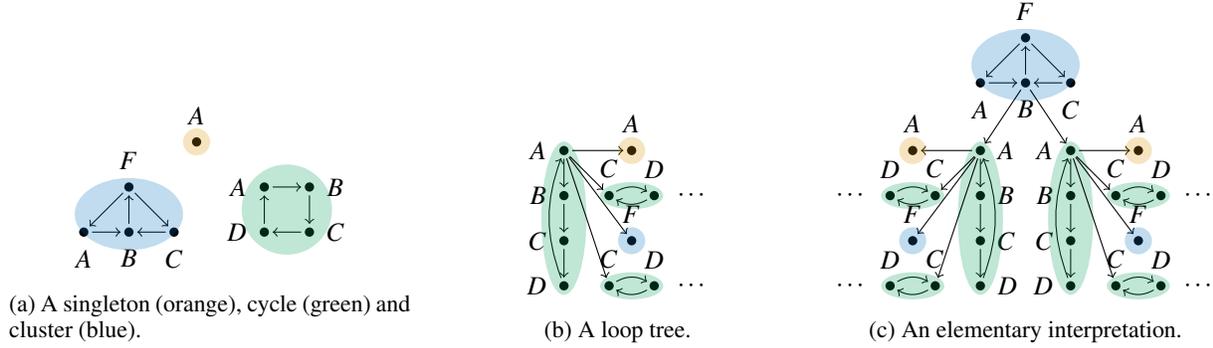


Elementary interpretations are defined hierarchically, with three levels of increasing complexity, illustrated in Figure~\ref{ex:interpretation-def}.

Three basic building blocks form the first level (Figure~\ref{fig:elementary:first}).  
We call an interpretation $\intp$ over $t$ a~\emph{singleton} if the associated graph $G_\intp$ is a~single isolated node; a~\emph{cycle} if $G_\intp$ is a~simple cycle; and a~\emph{critical cluster} if  $G_\intp$ consists of a~single SCC and $\intp$ contains a~critical element. 
Note that a~singleton interpretation is a~critical cluster as long as its unique element is critical.
We lift this terminology to rooted  interpretations naturally.

Loop trees, defined below, constitute the second level  (Figure~\ref{fig:elementary:second}). They are similar to \emph{cactuses} defined  e.g. in~\cite{danecki1984propositional}. The operation of \emph{attaching} an interpretation $(\JJ, v)$ to an element $u'$ of an interpretation $(\II, u)$ results in interpretation $(\II',u)$ where $\II'$ is $\II \cup \JJ$ with an additional edge from $u'$~to~$v$; note that  $\Delta^\II$ and $\Delta^\JJ$ need not be disjoint.

\begin{definition}
We define \emph{loop trees}  inductively as follows.
\begin{itemize}
\item Every rooted singleton interpretation is a~loop tree.
\item If $(\intp,u)$ is a~singleton or a~cycle and has no critical elements, and $(\intp_1, u_1), \dots, (\intp_n, u_n)$ are loop trees that share no non-critical elements with each other and with $(\intp,u)$, then the interpretation obtained from $(\intp,u)$ by attaching $(\intp_1, u_1), \dots, (\intp_n, u_n)$ to $u$ is a~loop tree. 
\end{itemize}
\end{definition}

Finally, we define elementary interpretations as trees built from  critical clusters and loop trees (Figure~\ref{fig:elementary:third}). 

\begin{definition}
We define  \emph{elementary interpretations} inductively as follows. 
\begin{itemize}
\item Every rooted critical cluster and every loop tree  is an elementary interpretation. 
\item If $(\intp,u)$ is a~loop tree and $(\intp_1, u_1), \dots, (\intp_n, u_n)$ are pairwise disjoint elementary interpretations such that $u_i$ is the only element shared by $(\intp,u)$ and $(\intp_i,u_i)$, and it is critical in both interpretations, then $(\intp\cup\intp_1\cup \dots\cup\intp_n, u)$ is elementary. 
\item If $(\intp,u)$ is a~critical cluster, $(\JJ, v)$ is an elementary interpretation without cycles visiting $v$, and the only element shared by $(\intp,u)$ and $(\JJ,v)$ is $v$, then $(\II \cup \JJ, u)$ is an elementary interpretation.
\end{itemize}
\end{definition}

Elementary interpretations will be used as finite representations of potential countermodels, via quasi-unravelling.  Not every countermodel has such a representation, but, as we show next,  each can be turned into one that does.

\subsection{From arbitrary to elementary interpretations}

We are now ready to realize the goal of this section: given an interpretation $\mathcal{I}$ with finitely many critical elements, we find a small elementary interpretation $\mathcal{E}$ with  a $t$-strong homomorphism from the quasi-unravelling of $\mathcal{E}$ to $\mathcal{I}$. 


\begin{theorem}
\label{thm:onerole}
For every transitive rooted $X$\=/labelled interpretation $(\intp,a,\iota)$ over a~ transitive role $t$ with finitely many critical elements there is an elementary rooted $X$\=/labelled interpretation $(\EE, a, \epsilon)$ over $t$ of size at most $(\ell+1)^{(\ell+1)^2}$  for $\ell=|\conceptsIn{\intp}|$ 
 whose quasi\=/unravelling $(\widetilde\EE, a, \widetilde\epsilon)$ maps to $(\intp,a,\iota)$ via a~$t$\=/strong homomorphism. 
\end{theorem}

To prove this theorem, we first observe that the length of simple paths that visit critical elements can be bounded.
%
%
%
%
%
%
We call a~path \emph{simple} if it never visits the same element twice, and  \emph{critical} if it visits only critical elements.  

\begin{restatable}{lmm}{lemShortPaths}
\label{lem:shortPaths}
Consider a~transitive interpretation  $\intp$ with finitely many critical elements and let $b\in\Delta^\intp$. Then there is an induced subinterpretation $\JJ$ of $\intp$ such that  $b\in\Delta^{\JJ}$, 
$\rch{t}{\JJ}(u) = \rch{t}{\intp}(u)$ for all $u\in\Delta^{\JJ}$, 
and every simple critical path in $\JJ$ that does not visit $b$ has length at most $|\conceptsIn{\intp}|$.
\end{restatable}

\begin{proof}[Proof of Theorem~\ref{thm:onerole}]

The \emph{critical depth} of a~rooted $X$\=/labelled interpretation $(\intp,a, \iota)$ is the maximal length of a~simple critical path in $\intp$ that does not visit the root cluster $K^\intp_a$. If $\II$ contains finitely many critical elements, then the critical depth of $(\intp,a, \iota)$ is finite. 
By Lemma~\ref{lem:shortPaths}, we can assume that the critical depth of $(\intp, a,\iota)$ is at most $\ell$. 
We can also assume that all elements of $\intp$ are reachable from $a$. 
By induction on $d$, we will construct for  $(\intp, a,\iota)$ of critical depth $d$ a~suitable  elementary $(\EE, a, \epsilon)$ of size $(\ell+1)^{ (\ell+1)\cdot(d+1)}$.

Consider $(\intp,a,\iota)$ of critical depth $d$. Let $(\widetilde \intp, a, \widetilde \iota)$ be the quasi\=/unravelling of $(\intp, a,  \iota)$. Note that the critical depth of $(\widetilde\intp,a)$ is $d$ as well. 

We call a~cluster $K$ in $\widetilde \intp$ \emph{non-root} if $K \neq K^{\widetilde\intp}_a$. Because all elements in $\intp$ are reachable from $a$, all critical clusters in $\widetilde\intp$ are reachable from $K^{\widetilde\intp}_a$. 
We call a~non\=/root critical cluster \emph{minimal} if it is not reachable from any other non-root critical cluster. 
Let $K_1, K_2, \dots, K_m$ be the minimal non-root critical clusters in $\widetilde\intp$ (if $d=0$, there are none). For each $i\in \{1, 2, \dots, m\}$ pick a~critical element  $c_i \in K_i$ and let $C=\{c_1, \dots, c_m\}$. 

For $c\in C$, let $(\widetilde\intp_{c}, \widetilde \iota_{c})$ be the $X$\=/labelled interpretation obtained by restricting $(\widetilde \intp, \widetilde \iota)$ to the domain consisting of $c$ and all elements reachable from $c$.
Clearly, the critical depth of $(\widetilde\intp_{c}, c, \widetilde\iota_c)$ is at most $d-1$. Let $(\EE_c, c, \epsilon_c)$ be the elementary interpretation of size at most $(\ell+1)^{(\ell+1)\cdot d}$  obtained from  $(\widetilde\intp_{c}, c,\widetilde\iota_c)$ by the induction hypothesis.
Without loss of generality we may assume that all elements of $\EE_c$ except $c$ are fresh; that is, they do not occur in the domain of any other interpretation. 

It remains to deal with the subinterpretation of $(\widetilde \intp, \widetilde\iota)$ induced by the set $\Delta^{\widetilde \intp} \setminus \bigcup_{c\in C} \Delta^{\widetilde \intp_c}$ of elements not reachable from $C$. By the definition of quasi\=/unravelling, this interpretation is the transitive closure of the union of the restriction $(\intp_a, \iota_a)$ of $(\intp, \iota)$ to $K^{\intp}_a$ and the restriction $(\widetilde\intp_S, \widetilde \iota_S)$ of $(\widetilde \intp, \widetilde \iota)$  to 
the set $S$ consisting of $a$ and all elements of the form $(K^{\intp}_a,p)$ from $\Delta^{\widetilde \intp}$. Note that all elements of the latter form are non-critical. Hence, $S \setminus \{a\}$ contains no critical elements.
Moreover, $\widetilde\intp_S$ is the transitive closure of a~tree with root $a$. We think of $S$ as if it were this tree and speak of nodes, leaves, subtrees, and children in $S$.  

We define the \emph{rank} of a~node $v$ in $S$ as the size of the set $\types{\widetilde\intp}{v} \cup \rch{t}{\widetilde\intp}(v)$ of concept names that occur in $v$ or elements of $\widetilde\intp$ (not necessarily in $S$) reachable from $v$. By the rank of a~subtree of $S$ rooted at $v$ we mean the rank of $v$.  Note that $\types{\widetilde\intp}{v} \cup \rch{t}{\widetilde\intp}(v) \supseteq \types{\widetilde\intp}{v'} \cup \rch{t}{\widetilde\intp}(v')$ for each descendant $v'$ of $v$, so ranks never increase as we go down $S$.

Below, we write  $(\widetilde\intp_{S\cup C}, \widetilde \iota_{S\cup C})$ for the restriction of $(\widetilde \intp, \widetilde \iota)$ to  $S\cup C$ and we say that a~homomorphism $h$ \emph{fixes} $C$ if $h(c)=c$ for each $c\in \dom(h)\cap C$.

We prove the following claim by induction on $q$.

\begin{claim}
\label{mainclaim}
For each  $v \in S\setminus\{a\}$ of rank $q$ there is a~loop tree $(\LL_v, u_v, \lambda_v)$ of size at most  $(\ell+1)^{q}$ 
%
and a~homomorphism \[h_v : \big(\widetilde\LL_v, u_v, \widetilde\lambda_v\big) \to \big(\widetilde \intp_{S\,\cup\,C}, v', \widetilde \iota_{S\,\cup\, C}\big)\] 
for some $v'$ such that 
\begin{itemize}
\item 
$v'$ is an element of $S$ reachable from $v$ and 
\begin{equation}
\label{eq:sameconcepts}
\types{\widetilde\intp}{v} \cup \rch{t}{\widetilde\intp}(v) = \types{\widetilde\intp}{v'} \cup \rch{t}{\widetilde\intp}(v') \,;
\end{equation}
\item
$h_v$ extends to a~$t$\=/strong homomorphism fixing $C$
\[\widehat h_v : \big(\widetilde \EE_v, u_v, \widetilde \epsilon_v\big) \to \big(\widetilde\intp, v', \widetilde \iota\big)\] 
where 
$\big(\EE_{v}, \epsilon_{v} \big)$
is the union of 
$\big(\LL_v, \lambda_v\big)$ and $\big(\EE_c, \epsilon_c)$ for all $c\in C$ such that $c \in \Delta^{\LL_v}$.
\end{itemize}
\end{claim}

\begin{proof}[Proof of Claim~\ref{mainclaim}]
To make the construction of $\big(\LL_v, u_v, \lambda_v\big)$ more uniform, for $c\in C$, we let $(\LL_{c}, u_c, \lambda_c)$ be the restriction of $(\widetilde\II, c, \widetilde\iota)$ to $\{c\}$ with all edges removed. Note that $(\LL_{c}, u_c, \lambda_c)$ satisfies the conditions in the statement of the claim with $c'=c$ and the identity function for $h_c$.

Consider first a~node $v \in S \setminus \{a\}$ of rank 0.
We define $(\LL_v, u_v, \lambda_v)$ as the restriction of $(\widetilde \intp, v, \widetilde \iota)$ to $\{v\}$. We let $v'=v$ and take the identity function for $h_v$. The condition in the first bullet holds trivially. For the second bullet, $\Delta^{\LL_v}=\{v\}$ and $v\notin C$, so $(\EE_v, u_v,\epsilon_v) = (\LL_v, u_v, \lambda_v)$ and $h_v$ itself is the required $t$-strong homomorphism from $(\widetilde \EE_v, u_v, \widetilde \epsilon_v)$ to $\big(\widetilde\intp, v', \widetilde \iota\big)$ 
because $\types{\widetilde\EE_v}{u_v} \cup \rch{t}{\widetilde\EE_v}(u_v)=\emptyset = \types{\widetilde\intp}{v} \cup \rch{t}{\widetilde\intp}(v)$.

Suppose now that we have shown the claim for nodes of rank at most $q-1$ and consider a~node $v \in S\setminus \{a\}$ of rank $q>0$. Let $S_v$ be the subtree of $S$ rooted at $v$ and let $R_v$ be the subset of $S_v$ consisting of all nodes of rank $q$. Note that $R_v$ forms a~\emph{prefix} of  $S_v$: if a~node from $S$ belongs to $R_v$, so do all its ancestors in $S$. We will think of $R_v$ as a~tree, too. Note also that every $v'\in R_v$ satisfies condition \eqref{eq:sameconcepts}.  We have two cases, depending on whether $R_v$ contains a leaf or not. 

\subsubsection{Case 1.} $R_v$ contains a~leaf $v'$; that is, all descendants of $v'$ in $S$  have lower rank. Pick a~minimal subset $W$ of elements of $S \cup C$ reachable from $v'$ in $\widetilde \intp$ 
such that 
\begin{equation}
\label{eq:case1}
\rch{t}{\widetilde\intp} (v') = \bigcup_{w \in W} 
\types{\widetilde\intp}{w} \cup \rch{t}{\widetilde\intp}(w)\,.  
\end{equation}
Each element $w \in W$ is an element of $S\setminus \{a\}$ of rank at most $q-1$ or an element of $C$. By the induction hypothesis and the initial step of the proof of the claim, $(\LL_{w}, u_{w}, \lambda_w)$ exists for all $w\in W$. 
We let $u_v = v'$ and define $(\LL_v, \lambda_v)$ as follows. 
We restrict $(\widetilde\intp,\widetilde\iota)$ to $\{v'\}$ 
and for each $w\in W$ we attach to $v'$ a~\emph{quasi\=/fresh} copy of the interpretation $(\LL_{w}, u_{w}, \lambda_w)$, where each non-critical element is replaced by a~fresh element.
Because $|W|\leq q \leq \ell$ and $|\Delta^{\LL_w}|\leq (\ell+1)^{q-1}$ for all $w \in W$, \[|\Delta^{\LL_v}|\leq 1+ q\cdot (\ell+1)^{q-1} \leq (\ell+1)^q\,.\]
A suitable  $h_v\colon \big(\widetilde\LL_v, u_v, \widetilde\lambda_v\big) \to \big(\widetilde \intp_{S\,\cup\,C}, v', \widetilde \iota_{S\,\cup\, C}\big)$ can be obtained by setting $h_v(u_v) = v'$ and combining homomorphisms $h_w$ for $w\in W$ as follows.  Consider a~path $p$ in the quasi-fresh copy of $\LL_w$ used in $\LL_v$, beginning in the copy of $u_w$. There is a~corresponding path $p'$ in $\LL_w$, beginning in $u_w$. We let $h_v(u_vp) = h_w(p')$. For every $c \in C \cap \Delta^{\widetilde\LL_v}$ we let $h_v(c)=c$.
It is routine to check that $h_v$ is a~homomorphism and that it extends to a~suitable $t$\=/strong homomorphism fixing $C$. Let us only verify that $\rch{t}{\widetilde\EE_v}(u_v) = \rch{t}{\widetilde\intp}(v')$. Recall that $u_v=v'$. By construction, \[\rch{t}{\widetilde\EE_v} (v') = \bigcup_{w \in W} 
\types{\widetilde\EE_w}{u_w} \cup \rch{t}{\widetilde\EE_w}(u_w)\,.\] 
By the induction hypothesis, for each $w \in W$, using the existence of a $t$-strong homomorphism from $\big(\widetilde \EE_w, u_w, \widetilde \epsilon_w\big)$ to $\big(\widetilde\intp, w', \widetilde \iota\big)$, followed by \eqref{eq:sameconcepts} for $w$, we get
\begin{align*}
\types{\widetilde\EE_w}{u_w} \cup \rch{t}{\widetilde\EE_w}(u_w) & = 
\types{\widetilde\intp}{w'} \cup \rch{t}{\widetilde\intp}(w') = \\
& =\types{\widetilde\intp}{w} \cup \rch{t}{\widetilde\intp}(w)\,.
\end{align*} 
We conclude by applying \eqref{eq:case1}.

\subsubsection{Case 2.} $R_v$ does not contain a~leaf; that is, each node in $R_v$ has a~child in $R_v$. We need a~similar, but stronger property for concept names. For a~subset $U$ of $S$, we call a~concept name $A$ \emph{dense in $U$} if each element in $U$ has a~proper descendant in $U$ that belongs to $A^{\widetilde \intp}$. We construct a~subtree $R'_v$ of $R_v$ such that each concept name occurring in $R'_v$ is dense in $R'_v$. 

Let $\types{\widetilde\intp}{v} \cup \rch{t}{\widetilde\intp}(v) = \{A_1, \dots, A_n\}$. We process $A_i$ one by one and eliminate those that are not dense by passing to a~smaller subtree of $R_v$, also without leaves. Let $R_v^i$ be the subtree of $R_v$ obtained after processing  $A_1, \dots, A_i$, with $R_v^0=R_v$. If $A_{i+1}$ is dense in $R_v^i$ or does not occur in $R_v^i$, we let $R_v^{i+1}=R_v^i$. If $A_{i+1}$ occurs  in $R_v^i$ but is not dense in $R_v^i$, there is an element $u$ of $R_v^i$ without proper descendants in $A_{i+1}^{\widetilde\intp} \cap R_v^i$. Since $R_v^i$ has no leaves, node $u$ has a~child $u'$ in $R_v^i$. We define $R_v^{i+1}$ as the  the subtree of $R_v^i$ rooted at $u'$. Note that $R_v^{i+1}$ is non-empty and has no leaves.
Moreover, $A_{i+1}$ does not occur in $R_v^{i+1}$, 
each concept name that does not occur in $R_v^{i}$ does not occur in $R_v^{i+1}$ either, and each concept name dense in $R_v^{i}$ is also dense in $R_v^{i+1}$. 
We let $R'_v = R_v^{n}$.  

Let $B_1,\ldots,B_k$ be the concept names that occur in $R'_v$. Because all these concept names are dense in $R'_v$, we can find an infinite sequence   
\[v^1_1, v^1_2, \dots, v^1_k, \; v^2_1, v^2_2, \dots, v^2_k, \; \dots,\; v^i_1, v^i_2, \dots, v^i_k, \; \dots\] of elements from $R'_v$ that form a~subsequence of a~branch in $R'_v$ ensuring that $v_{\! j}^{\,i} \in B_j^{\widetilde\intp}$ for all $i$ and $j$. 
Next, for each~$i$, we pick a~minimal subset $W_i$ of elements of $(S \setminus R'_v) \cup C$ reachable from $v^i_1$ in $\widetilde \intp$ such that \[\rch{t}{\widetilde\intp}(v^i_1) \setminus \{B_1, \dots,B_k\} =  \bigcup_{w\in W_i} 
\types{\widetilde\intp}{w} \cup \rch{t}{\widetilde\intp}(w)\,.\] 
Each element of $W_i$ is an element of $S \setminus \{a\}$ of rank at most $q-1$ or an element of $C$. By the induction hypothesis and the initial step of the proof of the claim,
for each $w\in \bigcup_{i=1}^\infty W_i$,  $(\LL_{w}, u_{w}, \lambda_w)$ is already defined and is an $X$\=/labelled loop tree of size at most $(\ell+1)^{q-1}$ using only concept names from $\conceptsIn{\widetilde\intp}$ and critical elements from $C$. Clearly, there are only finitely many such interpretations, up to an isomorphism fixing $C$. In consequence, 
there is an infinite sequence $\ell_1 < \ell_2 < \dots $ such that
\begin{itemize}
\item the sets $\big\{ (\LL_{w}, u_{w}, \lambda_w) \bigm| w \in W_{\ell_i}\big\}$ are equal for all $i$ (up to an isomorphism fixing $C$), 
\item the tuples $\big(\types{\widetilde\intp}{v^{\ell_i}_1},\dots,  \types{\widetilde\intp}{v^{\ell_i}_k}\big)$ are equal for all $i$, and
\item the tuples $\big(\iota(v^{\ell_i}_1),\dots,  \iota(v^{\ell_i}_k)\big)$ are equal for all $i$.
\end{itemize}
Let $u_v = v^{\ell_1}_1$ and construct $(\LL_v, \lambda_v)$  by taking the restriction of $(\widetilde\intp, \widetilde \iota)$ to $\{v^{\ell_1}_1, v^{\ell_1}_2, \dots, v^{\ell_1}_k\}$ with all edges dropped, 
adding edges $(v^{\ell_1}_1, v^{\ell_1}_{2}), \dots, (v^{\ell_1}_{k-1}, v^{\ell_1}_{k}), (v^{\ell_1}_k, v^{\ell_1}_{1})$, 
and attaching to $v^{\ell_1}_1$ a~quasi\=/fresh copy of $(\LL_{w}, u_{w}, \lambda_w)$ for each $w\in W_{\ell_1}$.
Then, \[|\Delta^{\LL_v}| \leq k + (q-k)\cdot (\ell + 1)^{q-1} \leq (\ell+1)^q\,.\]

Let us now define $h_v\colon\big(\widetilde\LL_v, u_v, \widetilde\lambda_v\big) \to \big(\widetilde \intp_{S\,\cup\,C}, v', \widetilde \iota_{S\,\cup\, C}\big)$ for $v'=v_1^{\ell_1}$. For $c \in C \cap \Delta^{\widetilde\LL_v}$,  let $h_v(c)=c$. The elements of $\widetilde\LL_v$ obtained by unravelling the cycle $v^{\ell_1}_1, v^{\ell_1}_2, \dots, v^{\ell_1}_k$ in $\LL_v$ are mapped to the elements $v^{\ell_1}_1, v^{\ell_1}_2, \dots, v^{\ell_1}_k, v^{\ell_2}_1, v^{\ell_2}_2, \dots, v^{\ell_2}_k, \dots $ in $\widetilde \II$: for $i\geq 1$ and $j\leq k$
we let \[\underbrace{\big(v^{\ell_1}_1 v^{\ell_1}_2 \cdots v^{\ell_1}_k\big) \cdots \big(v^{\ell_1}_1 v^{\ell_1}_2 \cdots v^{\ell_1}_k\big)}_{i-1} v^{\ell_1}_1 v^{\ell_1}_2 \cdots v^{\ell_1}_j \ \mapsto \ v^{\ell_{i}}_j\,.\]
For the remaining elements we rely on $h_w$ for $w\in W_{\ell_1}$. For every $w \in W_{\ell_1}$ and $i \geq 1$ there is $w^{(i)} \in W_{\ell_i}$ such that $(\LL_w, u_w, \lambda_w)$ is isomorphic to $\big(\LL_{w^{(i)}}, u_{w^{(i)}}, \lambda_{w^{(i)}}\big)$ via an isomorphism $\eta_{w,i}$ fixing $C$. Consider a~path $p$ in the quasi-fresh copy of $(\LL_{w}, u_{w}, \lambda_{w})$ used in $(\LL_v, u_v, \lambda_v)$, starting from the copy of $u_w$. Using the isomorphism $\eta_{w,i}$, we find a~corresponding path $p^{(i)}$ in $\big(\LL_{w^{(i)}}, u_{w^{(i)}}, \lambda_{w^{(i)}}\big)$, starting from $u_{w^{(i)}}$. 
We let 
\[\underbrace{\big(v^{\ell_1}_1 v^{\ell_1}_2 \cdots v^{\ell_1}_k\big) \cdots \big(v^{\ell_1}_1 v^{\ell_1}_2 \cdots v^{\ell_1}_k\big)}_{i-1} v^{\ell_1}_1 p \ \mapsto \ h_{w^{(i)}}(p^{(i)})\,.\]  
The element $v'$ is reachable from $v$ and satisfies \eqref{eq:sameconcepts}, because $v'\in R'_v \subseteq R_v$. Verifying that $h_v$ extends to a~suitable $t$\=/strong homomorphism is straightforward, based on the properties of $v^{\ell_1}_1, v^{\ell_1}_2, \dots, v^{\ell_1}_k, v^{\ell_2}_1, v^{\ell_2}_2, \dots, v^{\ell_2}_k, \dots $ and the inductive hypothesis. This completes the proof of Claim~\ref{mainclaim}.
\end{proof}

With $(\LL_v, u_v, \lambda_v)$ at hand for all $v\in S\setminus\{a\}$, we can define the elementary interpretation we are looking for. Let $W$ be a~minimal subset  of $S \cup C$ such that 
\[ \rch{t}{\widetilde\intp}(a) \setminus \rch{t}{\intp_a}(a) = \bigcup_{w\in W} \types{\widetilde\intp}{w} \cup \rch{t}{\widetilde\intp}(w) \,.\]
Let $V$ be a~minimal subset of $K^\intp_a$ such that $a\in V$ and $\rch{t}{\intp_a}(a) = \bigcup_{v\in V} \types{\intp_a}{v}$.
Let  $(\intp'_a, \iota'_a)$ be the restriction of $(\intp_a, \iota_a)$ to $V$.
Let $(\LL_a,\lambda_a)$ the $X$\=/labelled interpretations obtained from $(\intp'_a,\iota_a)$ by attaching to $a$ a~quasi\=/fresh copy of $(\LL_w, u_w, \lambda_w)$ for all $w\in W$. Then, 
\begin{align*}    
|\Delta^{\LL_a}| & \leq 1 + |\rch{t}{\intp_a}(a)| + \big(\ell - |\rch{t}{\intp_a}(a)|\big)\cdot(\ell+1)^\ell \leq \\
&\leq(\ell+1)^{\ell+1}\,.
\end{align*}
As in Case 1 in the claim, we can construct a~homomorphism $h_a : \big (\widetilde\LL_a, a, \widetilde \lambda_a\big) \to \big(\widetilde \intp_{S\,\cup\,C}, a, \widetilde \iota_{S\,\cup\, C}\big)$ that extends to a~$t$\=/strong homomorphism
$\widehat h_a: (\widetilde\EE_a, a, \widetilde \epsilon_a\big) \to \big(\widetilde \intp, a, \widetilde \iota\big)$ where 
$(\EE_a, \epsilon_a)$ is the union of $(\LL_a, \lambda_a)$ and $\big(\EE_c, \epsilon_c)$ for all $c\in C$ such that $c \in \Delta^{\LL_a}$.
Note that  \[|\Delta^{\EE_a}| \leq (\ell+1)^{\ell+1} \cdot (\ell+1)^{(\ell+1)\cdot  d} \leq (\ell+1)^{(\ell+1)(d+1)}\,.\] 
Hence, $(\EE_a, a, \epsilon_a)$ is the rooted $X$\=/labelled interpretation we seek.
\end{proof}

\section{Partially finite coloured blocking}\label{sec:PartiallyfiniteTheorem}

This section provides a~variant of the \emph{coloured blocking} theorem~\cite{Gogacz2018FiniteQA,DBLP:conf/ijcai/GogaczGIJM19}, suitable for partially finite models (over multiple roles). 
The original theorem offers a method to turn an infinite countermodel into a finite one,  under certain assumptions justified by the \emph{a priori} existence of a finite countermodel. Here, we only know that a partially finite countermodel exists, which does not give equally strong assumptions. In what follows we describe an alternative  construction that preserves countermodels and ensures partial finiteness when applied to an interpretation satisfying suitable weaker assumptions.

\subsection{Quotients}

Let $\intp$ be an~interpretation and $\sim$  an~equivalence relation on $\Delta^\intp$. We write $[x]_\sim$ for the equivalence class $\{y\in\Delta^\intp \bigm| y\sim x\}$ of $x\in\Delta^\intp$  
and $\intp/\sim$ for the  \emph{quotient interpretation} where 
\begin{itemize}
    \item $\Delta^{(\intp/\sim)} = \big\{[x]_\sim \bigm| x \in \Delta^\intp\big\}$,
    \item $A^{(\intp/\sim)} = \big\{[x]_\sim \bigm| x \in A^\intp\big\}$ for all concept names $A$,
    \item $r^{(\intp/\sim)} = \big\{\big([x]_\sim,[y]_\sim\big) \bigm| (x,y) \in r^\intp\big\}$ for all role names $r$. 
\end{itemize}

The quotient construction naturally induces the \emph{quotient homomorphism} $h_\sim\from \intp\to \intp/\sim$ defined as  $h_\sim(x)=[x]_\sim$.

\subsection{Distances}

An interpretation  $\intp$ can be viewed as an undirected graph with an~edge between $x\in\Delta^\intp$ and $y\in\Delta^\intp$ whenever $(x,y)$ or $(y,x)$ belongs to $r^\intp$ for some $r\in \RolNameSet$. By an \emph{undirected path} in $\intp$ we mean a path in this undirected graph. 
The \emph{distance in $\intp$} is the function $d_\intp\colon \Delta^\intp\times \Delta^\intp\to \N\cup\{\infty\}$ where $d_\intp(x,y)$ is the minimal length of an undirected path  in $\intp$ between $x$ and $y$, or $\infty$ if there is no such path. Naturally, $d_\intp$ satisfies the triangle inequality: $d_\intp(x,z) \leq d_\intp(x,y) + d_\intp(y,z)$ for all $x,y,z\in\Delta^\intp$.
In the case of a~single transitive role, the distance is always $0$, $1$, or $\infty$. But already for two transitive roles $d_\intp$ can take all possible values.

For $\distN\geq 0$ we define the \emph{$\distN$\=/neighbourhood} of an~element $x\in\Delta^\intp$ as the set $N_\distN(x)=\big\{y\in\Delta^\intp\mid d_\intp(x,y)\leq \distN\big\}$. We will often be interested in the induced subinterpretation $\intp\restr N_\distN(x)$, which might be infinite.

\subsection{Coloured interpretations}

Let us fix a set $C$ of  \emph{colours}. A~\emph{coloured interpretation} is an~interpretation $\intp$ along with a~function $\kappa\colon \Delta \to C$ for some  $\Delta\subseteq\Delta^\intp$ that assigns colours to some elements of $\intp$. 

If $\intp$ is coloured by $\kappa$ and $\intb$ is coloured by $\theta$ then a~homomorphism $h\colon \intp\to\intb$ \emph{preserves colours} if for every $x\in \Delta^\intp$ we have: $x\in\dom(\kappa)\Leftrightarrow h(x)\in\dom(\theta)$ and if $x\in\dom(\kappa)$ then $\kappa(x)= \theta(h(x))$.

Two interpretations $\intp$ and $\intb$ coloured by $\kappa$ and $\theta$ are \emph{homomorphically equivalent} if there is a~homomorphism from $\intp$ to $\intb$ and a~homomorphism from $\intb$ to $\intp$, both of them preserving colours. Homomorphic equivalence is reflexive, symmetric, and transitive, but  it does not guarantee an~isomorphism between the interpretations.

\subsection{\texorpdfstring{$\distN$}{\distN}-neighbourhood equivalence}

The original coloured blocking relies on the finite number of isomorphism types of $\distN$\=/neighbourhoods. As we cannot assume it here, we use the weaker notion of homomorphism equivalence. 

Consider $\distN\geq 0$ and an~interpretation $\intp$ coloured by $\kappa$.
For each $x\in\Delta^\intp$, the interpretation $\intp\restr N_\distN(x)$ is naturally coloured by $\kappa\restr (N_\distN(x)\cap \dom(\kappa))$.
We define \emph{$\distN$\=/neighbourhood equivalence} $\sim_\distN$ as the equivalence relation on $\Delta^\intp$ such that $x \sim_\distN y$ iff $x=y$ or $x,y\in \dom(\kappa)$, $\kappa(x)=\kappa(y)$, and the coloured interpretations $\intp\restr N_\distN(x)$ and $\intp\restr N_\distN(y)$ are homomorphically equivalent via homomorphisms that maps $x$ to $y$ and vice versa.
%
Note that $\sim_\distN$ merges only coloured elements. 

\subsection{The coloured blocking theorem}

For $\distN\geq 0$, we say that an~interpretation $\intp$ coloured by $\kappa$ is \emph{$\distN$\=/sparse} if for every $x\in\Delta^\intp$ and every two $y,y'\in N_\distN(x)\cap \dom(\kappa)$, if $y\neq y'$ then $\kappa(y)\neq\kappa(y')$.

\begin{restatable}{theorem}{thmColourBlocking}
\label{thm:pfColourBlocking}
Let $K=2\cdot\distN^3$ for $\distN\geq 0$. Consider a $K$\=/sparse interpretation $\intp$ coloured by $\kappa$ and a~finite interpretation   $\intq$ with $|\Delta^\intq|\leq \distN$.
Then $\intq$ maps homomorphically into $\intp$ iff $Q$ maps homomorphically into $\intp/{\sim_K}$.
\end{restatable}
 
A~complete proof of this result is given in Appendix~\ref{app:PartiallyfiniteTheorem}. Here we provide only a very high level sketch.

As the quotient homomorphism $h_{\sim_K}$ maps $\intp$ to $\intp/{\sim_K}$, the right-to-left implication is clear. For the converse, consider a homomorphism $h\from \intq\to \intp/{\sim_K}$. We need to lift it to a~homomorphism from $\intq$ to $\intp$. We begin by dividing $\intq$ into finitely many  \emph{pieces}. Initially, each piece contains at most two coloured elements. For such pieces we can lift the corresponding restriction of $h$ to a homomorphism to $\intp$, using $h_{\sim_K}$. Then, we iteratively glue pieces into larger ones, while preserving the existence of homomorphisms to $\intp$. This is possible because $\sim_K$-equivalent elements of $\intp$ have homomorphically equivalent neighbourhoods. Ultimately all pieces are glued together into a single piece $\intq$ with a corresponding homomorphism into $\intp$.

In the following section we apply Theorem~\ref{thm:pfColourBlocking} to an~interpretation of a~special shape, with suitably coloured critical elements, and prove that the resulting quotient interpretation is partially finite.

\section{Multiple roles allowed}\label{MultipleRoles}
\label{sec:multiple-roles}

In this section we show how to solve the entailment problem in the general case. The strategy is familiar: we characterize non-entailment in terms of structurally simple witnesses and show that the existence of such a witness is decidable. 

A rooted interpretation $(\intp,a)$ is \emph{piecewise single-role} if it is the union of a family of pairwise disjoint rooted interpretations called \emph{pieces}, that are either singleton interpretations or single-role interpretations over a transitive role, called \textit{transitive pieces}, with some additional edges between elements from different pieces.  Specifically, the pieces can be arranged into a potentially infinite tree such that 
\begin{itemize}
\item 
the root piece is a singleton interpretation with root $a$;
\item for every two adjacent pieces in the tree, $\intp$ contains exactly one edge from an element in the parent piece to the root of the child piece, and there are no other edges in $\intp$ between elements from different pieces.
\end{itemize}
Moreover, for every transitive role name $t$, if a piece contains a $t$-edge or its root has an incoming $t$-edge, then there are no $t$-edges originating in this piece and leading to other pieces. The latter condition ensures that there is at most one decomposition into pieces. 

An interpretation $(\GG, a)$ is \emph{piecewise elementary} if it is piecewise single-role and each non-singleton piece is a transitive piece. 

The unravelling of a piecewise single-role interpretation $(\GG, a)$ is a piecewise single-role interpretation $(\widetilde\GG,a)$  obtained as follows. Starting from the (singleton) root piece of $(\GG,a)$, repeat the following exhaustively: for each newly added  element $e'$ that is a copy of element $e$ from  $\GG$ (or $e$ itself), for each piece $(\JJ,b)$ of $(\GG,a)$  such that there is an $r$-edge from $e$ to $b$, add a fresh copy of the quasi-unravelling of $(\JJ,b)$, with an $r$-edge from $e'$ to the root.  Finally, close all transitive roles by transitivity. 

\begin{restatable}{theorem}{thmpiecewiseelementary}
\label{thm:piecewise-elementary}
For $\KB=\big(\TBox, \{A(a)\}\big)$, a concept name~$F$, and a UCQ $Q$, if $\KB\not\models_F Q$ then there is a piecewise elementary interpretation $(\GG,a)$ of piece size at most $(\ell+1)^{(\ell+1)^2}$ for $\ell=|\conceptsIn{\KB}|$ and degree at most $\|\TBox\|$ whose quasi-unravelling $(\widetilde\GG, a)$ satisfies
$\widetilde\GG \models \KB$ and $\widetilde\GG\not\models Q$. 
\end{restatable}

\begin{proof}[Proof sketch] Suppose that $\KB\not\models_F Q$. Then there exists a countermodel $\II$. Let $\JJ$ be classical unravelling of $\II$. $\JJ$ can be divided into connected components in a single role. For each component, we would like to find a subset that is the quasi\=/unravelling of an elementary interpretation. Theorem \ref{thm:onerole} gives one role substructures with desired bounds.

The problem arise when we have to decide the ordering of components to be shrunk. If we would do this in top down fashion, the final structure could end up empty. We are forced to do it in bottom tup fashion. But how do we do it in infinite tree of components?

We define initial fragments $\UU_k$ of $\JJ$ of depth $k$. Then we process $\UU_k$'s bottom up fashion getting structures $\KK_k$. We use the K\"onig lemma to get a single infinite structure out of structures $\KK_k$'s.
\end{proof}

A characterization of non-entailment could be obtained by proving the converse of Theorem~\ref{thm:piecewise-elementary}, but the condition $\widetilde \GG \not \models Q$ makes it hard to decide the existence of such witnesses. Instead, we replace this condition with a weaker one, that is easier to check and  still sufficient to prove the converse. 

We use a  blow-up operation, which can also be seen as partial unravelling. Given a piecewise elementary interpretation $(\GG,a)$ and a positive integer $n$, we define $(\GG_n,a)$ as follows. We process $\GG$ top-down, piece by piece. In each piece we blow-up the cycles in loop trees, also in the top-down order: we unravel each cycle to a~cycle $2n+1$ times longer. Each copy $e'$ of element $e$ from the original cycle gets its own fresh copy of the subinterpretation of $\GG$ induced by elements reachable from $e$ in $\GG$ without visiting any other elements on the cycle. 
When all the cycles are blown-up, we partially close the resulting interpretation by transitivity: we add all edges from the transitive closure except those between elements from the same blown-up cycle.  Finally, in each blown-up cycle  we transitively close the first $n$ of the $2n+1$ copies of the original cycle.

\begin{restatable}{lmm}{weakercriterion}
\label{lem:weaker-criterion}
For every piecewise elementary interpretation $(\GG, a)$ and positive integer $n$, the quasi-unravelling $(\widetilde{\GG_n}, a)$ of $(\GG_n, a)$ maps homomorphically into the quasi-unravelling $(\widetilde \GG, a)$ of $(\GG, a)$.
\end{restatable}

\begin{proof}[Proof sketch]
$\GG_n$ can be seen as a partial unravelling. The lemma formalise the intuition that partial unravelling followed by full unravelling is essentially the same as just full unravelling.
\end{proof}

By Lemma~\ref{lem:weaker-criterion}, Theorem~\ref{thm:piecewise-elementary} holds with  $\widetilde\GG\not\models Q$ replaced with $\widetilde{\GG_n}\not\models Q$ for any $n$. We can also prove the converse: a~partially finite counter-model is obtained by applying Theorem \ref{thm:pfColourBlocking} to $\widetilde \GG_n$ with suitably coloured critical elements. 

\begin{restatable}{lmm}{converse}
\label{lem:converse}
For $\KB=\big(\TBox, \{A(a)\}\big)$, a concept name $F$, and a UCQ $Q$, if there is a piecewise elementary interpretation $(\GG,a)$ of bounded piece size and bounded degree  such that 
$\widetilde\GG \models \KB$ and $\widetilde{\GG_n}\not\models Q$, then $\KB\not\models_F Q$.
\end{restatable}

\begin{proof}[Proof sketch]
We will use Theorem \ref{thm:pfColourBlocking}. In order to do it,  we have to define sparse colouring. Since the quasi\=/unravelling does not multiply critical elements, such colouring can be defined on $\GG_n$ instead on $\widetilde{\GG_n}$. Existence of such colouring follows from bounded  degree in $\GG_n$.
\end{proof}

Theorem~\ref{thm:piecewise-elementary}, Lemma~\ref{lem:weaker-criterion}, and Lemma~\ref{lem:converse} together establish 
an~equivalent criterion for non\=/entailment. 
%
%
But how does replacing $\widetilde \GG \not\models Q$ with $\widetilde \GG_n \not\models Q$ help with decidability? The construction ensures that  $\widetilde{\GG_n}$ and $\GG_n$ satisfy the same conjunctive queries of size at most $n$.

\begin{restatable}{lmm}{theweakestcriterion}
\label{lmm:the-weakest-criterion}
For a piecewise elementary interpretation $(\GG,a)$, a positive integer $n$, and a conjunctive query $Q$ with at most $n$ variables, $\GG_n\models Q$ iff $\widetilde{\GG_n}\models Q$.
\end{restatable}

\begin{proof}[Proof sketch]
The  only if direction follows from the fact that $\GG_n$ does not hvae short cycles. The if direction would be the standard unravelling property if not for the additional transitive clousure in the definition of quasi\=/unravelling. But those additional edges can be mimicked by partial transitive closure of loops in $\GG_n$
\end{proof}

Hence, we can replace the condition $\widetilde \GG_n \not \models Q$ in our criterion  with $\GG_n \not \models Q$ where $n$ is the maximal number of variables in a CQ from $Q$. The latter condition can be checked on the fly during a type-elimination procedure, which can be used to check the existence of a suitable witness. 

\begin{restatable}{lmm}{complexity}
\label{lem:complexity}
Given a KB $\KB=(\TBox, \{A(a)\})$, a concept name $F$, and a UCQ $Q$, one can decide in 2ExpTime whether there is 
a piecewise elementary interpretation $(\GG,a)$ of piece size at most $(\ell+1)^{(\ell+1)^2}$ for $\ell=|\conceptsIn{\KB}|$ and degree at most $\|\TBox\|$ such that
$(\widetilde\GG,a) \models \KB$ and $\GG_n\not\models Q$. 
\end{restatable}

Putting together all the ingredients, along with the  standard method for simplifying the ABox (see e.g. \cite{garcia-finite-s-arxiv}), we obtain our main result.

\begin{restatable}{theorem}{fullresult}
\label{thm:full-result}
Partially finite entailment of conjunctive queries for $\mathcal{S}$ is 2ExpTime-complete.   
\end{restatable}

The lower bound holds already for finite and unrestricted entailment \cite{garcia-finite-s-arxiv}.
\section{An application of partially finite entailment}\label{sec:application-to-closed-predicates}

As our last contribution, we present a result illustrating that partially finite reasoning occurs naturally 
in problems related to closed predicates. 
In the problem of  \emph{query containment in the presence of closed predicates}, we are given two Boolean CQs $q_1$, $q_2$, a TBox $\mathcal{T}$, and a~set $\mathcal{F}$ of concept names, and we have to decide
whether $q_1 \subseteq_ {\TT,\mathcal{F}} q_2$, that is, whether for every  ABox $\ABox$ and every interpretation $\mathcal{I}$ such that $\II \models (\mathcal{T}, \mathcal{A})$ and $A^\intp = \{a \ |\ A(a)\in\ABox\}$ for each $A\in\mathcal{F}$, if  $\intp\models q_1$ then $\intp\models q_2$.  Notice that in the formulation of the problem we essentially quantify universally over partially finite interpretations $\II$ with $F$ defined as the union of predicates from $\mathcal{F}$. And indeed, this problem can be solved via a Turing reduction to partially finite entailment. 

\begin{theorem}
    \label{thm:containment-to-entailement}
    Using an oracle for the partially finite entailment problem, one can decide query non-containment in the presence of closed predicates in nondeterministic polynomial time.
\end{theorem}


\begin{proof}
Let $q_1,q_2$ be conjunctive queries, $\mathcal{T}$ a~TBox, and $\mathcal{F}=\{F_1, F_2, \dots, F_n\}$ a set of concept names. We will show that  $q_1 \subsetneq_{\TT,\mathcal{F}} q_2$  iff  $(\mathcal{T},\ABox_{q'_1}) \not\models_F q_2$ for $F=F_1\sqcup F_2 \sqcup \dots \sqcup F_n$ and some ABox $\ABox_{q'_1}$ obtained from a homomorphic image $q'_1$ of $q_1$ by interpreting variables as individuals and atoms as assertions. 
%
This suffices, because the algorithm for query non-containment in the presence of closed predicates can then simply guess the correct homomorphic image of $q_1$ and test partially finite entailment.

Suppose that  $(\mathcal{T},\ABox_{q'_1}) \not\models_F q_2$ for some $\ABox_{q'_1}$ as above. Then, there is an interpretation $\mathcal{I}$ such that $\mathcal{I} \models (\TT,\ABox_{q'_1})$,  $\mathcal{I} \not \models q_2$, and $F^\mathcal{I}$ is finite. Consequently,  $q_1 \not\subseteq_{\TT,\mathcal{F}}q_2$ as witnessed by the ABox $\ABox = \bigcup_{i=1}^n\{F_i(a) \ | \ a \in F_i^\mathcal{I}\}$ and the interpretation $\mathcal{I}$.

    Conversely, suppose that $q_1 \not\subseteq_{\TT,\mathcal{F}}q_2$.  Then there is an ABox $\ABox$ and an interpretation $\mathcal{I}$ such that $\II \models (\mathcal{T}, \mathcal{A})$, $F_i^\intp = \{a \ |\ F_i(a)\in\ABox\}$ for each $i$,  $\intp\models q_1$, and $\intp\not\models q_2$.
    Clearly,  $F^\intp$ is finite. 
    Some homomorphic image $q'_1$ of $q_1$ maps to $\intp$ via an injective homomorphism $h$. Assuming that in $\ABox_{q'_1}$ variables of $q'_i$ are replaced with their images via $h$,
    $\intp \models (\TT,\ABox_{q'_1})$. Hence, $\intp$  witnesses that $(\TT,\ABox_{q'_1})\not\models_F q_2$. 
\end{proof}

    In the above argument we need to check all homomorphic images of query $q_1$ because we adopt the Standard Names Assumption in our partially finite entailment problem. Had we not done so, we could use a single ABox that could collapse to either image of the query and the algorithm would work in deterministic polynomial time. 

\section{Conclusion}
\label{sec:ending}
With this paper we aim to initiate a~systematic study of partially finite reasoning, where the reasoning tasks are performed under the assumption that some, explicitly defined, parts of the interpretations are necessarily finite. We have shown that in the case of query entailment this mode of reasoning, subsuming both query entailment and finite query entailment, does not come with an increased cost, retaining its $2$\=/$\mathrm{ExpTime}$ complexity. 
In the future, it would be interesting to inspect other logics without finite controllability, like $\mathcal{ALCIF}$, as well as other reasoning tasks including query containment, enumeration, or direct access.

\section*{Acknowledgements}
This work was supported by Poland’s NCN grant 2018/30/E/ST6/00042.

\bibliographystyle{named}
\bibliography{references}

@inproceedings{Murlak23,
  author       = {Filip Murlak},
  editor       = {Marco Console and
                  Boris Konev},
  title        = {Finite-Model Reasoning for Graph Queries and Description Logics},
  booktitle    = {Reasoning Web. Declarative Artificial Intelligence: Knowledge, Rules,
                  Logic - 19th International Summer School 2023 Oslo, Norway, September
                  21-24, 2023, Tutorial Lectures},
  series       = {Lecture Notes in Computer Science},
  volume       = {15400},
  pages        = {23--53},
  publisher    = {Springer},
  year         = {2023},
  url          = {https://doi.org/10.1007/978-3-031-80283-6\_2},
  doi          = {10.1007/978-3-031-80283-6\_2},
  bibsource    = {dblp computer science bibliography, https://dblp.org}
}

@inproceedings{Gogacz2018FiniteQA,
  title={Finite Query Answering in Expressive Description Logics with Transitive Roles},
  author={Tomasz Gogacz and Yazm{\'i}n Ang{\'e}lica Ib{\'a}{\~n}ez-Garc{\'i}a and Filip Murlak},
  booktitle={International Conference on Principles of Knowledge Representation and Reasoning},
  year={2018},
  url={https://api.semanticscholar.org/CorpusID:51954618}
}

@inproceedings{DBLP:conf/amw/AhmetajOS16,
  author       = {Shqiponja Ahmetaj and
                  Magdalena Ortiz and
                  Mantas Simkus},
  editor       = {Reinhard Pichler and
                  Altigran Soares da Silva},
  title        = {Polynomial Datalog Rewritings for Ontology Mediated Queries with Closed
                  Predicates},
  booktitle    = {Proceedings of the 10th Alberto Mendelzon International Workshop on
                  Foundations of Data Management, Panama City, Panama, May 8-10, 2016},
  series       = {{CEUR} Workshop Proceedings},
  volume       = {1644},
  publisher    = {CEUR-WS.org},
  year         = {2016},
  url          = {https://ceur-ws.org/Vol-1644/paper19.pdf},
  bibsource    = {dblp computer science bibliography, https://dblp.org}
}

@inproceedings{DBLP:conf/aaai/0002025,
  author       = {Lorenzo Marconi and
                  Riccardo Rosati},
  editor       = {Toby Walsh and
                  Julie Shah and
                  Zico Kolter},
  title        = {Consistent Query Answering over Existential Rules with Open and Closed
                  Predicates},
  booktitle    = {AAAI-25, Sponsored by the Association for the Advancement of Artificial
                  Intelligence, February 25 - March 4, 2025, Philadelphia, PA, {USA}},
  pages        = {15083--15091},
  publisher    = {{AAAI} Press},
  year         = {2025},
  url          = {https://doi.org/10.1609/aaai.v39i14.33654},
  doi          = {10.1609/AAAI.V39I14.33654},
  bibsource    = {dblp computer science bibliography, https://dblp.org}
}

@inproceedings{DBLP:conf/kr/GogaczLOS20,
  author       = {Tomasz Gogacz and
                  Sanja Lukumbuzya and
                  Magdalena Ortiz and
                  Mantas Simkus},
  editor       = {Diego Calvanese and
                  Esra Erdem and
                  Michael Thielscher},
  title        = {Datalog Rewritability and Data Complexity of {ALCHOIF} with Closed
                  Predicates},
  booktitle    = {Proceedings of the 17th International Conference on Principles of
                  Knowledge Representation and Reasoning, {KR} 2020, Rhodes, Greece,
                  September 12-18, 2020},
  pages        = {434--444},
  year         = {2020},
  url          = {https://doi.org/10.24963/kr.2020/44},
  doi          = {10.24963/KR.2020/44},
  bibsource    = {dblp computer science bibliography, https://dblp.org}
}

@inproceedings{DBLP:conf/kr/NgoOS16,
  author       = {Nhung Ngo and
                  Magdalena Ortiz and
                  Mantas Simkus},
  editor       = {Chitta Baral and
                  James P. Delgrande and
                  Frank Wolter},
  title        = {Closed Predicates in Description Logics: Results on Combined Complexity},
  booktitle    = {Principles of Knowledge Representation and Reasoning: Proceedings
                  of the Fifteenth International Conference, {KR} 2016, Cape Town, South
                  Africa, April 25-29, 2016},
  pages        = {237--246},
  publisher    = {{AAAI} Press},
  year         = {2016},
  url          = {http://www.aaai.org/ocs/index.php/KR/KR16/paper/view/12906},
  bibsource    = {dblp computer science bibliography, https://dblp.org}
}

@article{DBLP:journals/ngc/Katsuno90,
  author       = {H. Katsuno},
  title        = {Closed World Assumptions Having Precedence in Predicates},
  journal      = {New Gener. Comput.},
  volume       = {8},
  number       = {3},
  pages        = {185--209},
  year         = {1990},
  url          = {https://doi.org/10.1007/BF03037516},
  doi          = {10.1007/BF03037516},
  bibsource    = {dblp computer science bibliography, https://dblp.org}
}

@article{mccarthy1980circumscription,
  title={Circumscription—a form of non-monotonic reasoning},
  author={McCarthy, John},
  journal={Artificial intelligence},
  volume={13},
  number={1-2},
  pages={27--39},
  year={1980},
  publisher={Elsevier}
}

@article{gelfond1989relationship,
  title={On the relationship between circumscription and negation as failure},
  author={Gelfond, Michael and Przymusinska, Halina and Przymusinski, Teodor},
  journal={Artificial Intelligence},
  volume={38},
  number={1},
  pages={75--94},
  year={1989},
  publisher={Elsevier}
}

@inproceedings{DBLP:conf/ijcai/GogaczGIJM19,
  author       = {Tomasz Gogacz and
                  V{\'{\i}}ctor Guti{\'{e}}rrez{-}Basulto and
                  Yazm{\'{\i}}n Ib{\'{a}}{\~{n}}ez{-}Garc{\'{\i}}a and
                  Jean Christoph Jung and
                  Filip Murlak},
  editor       = {Sarit Kraus},
  title        = {On Finite and Unrestricted Query Entailment beyond {SQ} with Number
                  Restrictions on Transitive Roles},
  booktitle    = {Proceedings of the Twenty-Eighth International Joint Conference on
                  Artificial Intelligence, {IJCAI} 2019, Macao, China, August 10-16,
                  2019},
  pages        = {1719--1725},
  publisher    = {ijcai.org},
  year         = {2019},
  url          = {https://doi.org/10.24963/ijcai.2019/238},
  doi          = {10.24963/IJCAI.2019/238},
  bibsource    = {dblp computer science bibliography, https://dblp.org}
}

@inproceedings{DBLP:conf/ijcai/BienvenuB19,
  author       = {Meghyn Bienvenu and
                  Pierre Bourhis},
  editor       = {Sarit Kraus},
  title        = {Mixed-World Reasoning with Existential Rules under Active-Domain Semantics},
  booktitle    = {Proceedings of the Twenty-Eighth International Joint Conference on
                  Artificial Intelligence, {IJCAI} 2019, Macao, China, August 10-16,
                  2019},
  pages        = {1558--1565},
  publisher    = {ijcai.org},
  year         = {2019},
  url          = {https://doi.org/10.24963/ijcai.2019/216},
  doi          = {10.24963/IJCAI.2019/216},
  bibsource    = {dblp computer science bibliography, https://dblp.org}
}

@techreport{benedikt2018pspace,
  title={Pspace hardness of mixed world query answering for atomic queries under guarded tgds},
  author={Benedikt, Michael and Bourhis, Pierre},
  year={2018},
  institution={Technical report, University of Oxford}
}

@inproceedings{DBLP:conf/aaai/AndreselOS20,
  author       = {Medina Andresel and
                  Magdalena Ortiz and
                  Mantas Simkus},
  title        = {Query Rewriting for Ontology-Mediated Conditional Answers},
  booktitle    = {The Thirty-Fourth {AAAI} Conference on Artificial Intelligence, {AAAI}
                  2020, The Thirty-Second Innovative Applications of Artificial Intelligence
                  Conference, {IAAI} 2020, The Tenth {AAAI} Symposium on Educational
                  Advances in Artificial Intelligence, {EAAI} 2020, New York, NY, USA,
                  February 7-12, 2020},
  pages        = {2734--2741},
  publisher    = {{AAAI} Press},
  year         = {2020},
  url          = {https://doi.org/10.1609/aaai.v34i03.5660},
  doi          = {10.1609/AAAI.V34I03.5660},
  bibsource    = {dblp computer science bibliography, https://dblp.org}
}

@inproceedings{DBLP:conf/aaai/BajraktariOS18,
  author       = {Labinot Bajraktari and
                  Magdalena Ortiz and
                  Mantas Simkus},
  editor       = {Sheila A. McIlraith and
                  Kilian Q. Weinberger},
  title        = {Combining Rules and Ontologies into Clopen Knowledge Bases},
  booktitle    = {Proceedings of the Thirty-Second {AAAI} Conference on Artificial Intelligence,
                  (AAAI-18), the 30th innovative Applications of Artificial Intelligence
                  (IAAI-18), and the 8th {AAAI} Symposium on Educational Advances in
                  Artificial Intelligence (EAAI-18), New Orleans, Louisiana, USA, February
                  2-7, 2018},
  pages        = {1728--1735},
  publisher    = {{AAAI} Press},
  year         = {2018},
  url          = {https://doi.org/10.1609/aaai.v32i1.11565},
  doi          = {10.1609/AAAI.V32I1.11565},
  bibsource    = {dblp computer science bibliography, https://dblp.org}
}

@article{DBLP:journals/jair/BonattiLW09,
  author       = {Piero A. Bonatti and
                  Carsten Lutz and
                  Frank Wolter},
  title        = {The Complexity of Circumscription in DLs},
  journal      = {J. Artif. Intell. Res.},
  volume       = {35},
  pages        = {717--773},
  year         = {2009},
  url          = {https://doi.org/10.1613/jair.2763},
  doi          = {10.1613/JAIR.2763},
  bibsource    = {dblp computer science bibliography, https://dblp.org}
}

@article{DBLP:journals/tocl/DoniniNR02,
  author       = {Francesco M. Donini and
                  Daniele Nardi and
                  Riccardo Rosati},
  title        = {Description logics of minimal knowledge and negation as failure},
  journal      = {{ACM} Trans. Comput. Log.},
  volume       = {3},
  number       = {2},
  pages        = {177--225},
  year         = {2002},
  url          = {https://doi.org/10.1145/505372.505373},
  doi          = {10.1145/505372.505373},
  bibsource    = {dblp computer science bibliography, https://dblp.org}
}

@inproceedings{DBLP:conf/ijcai/Lifschitz91,
  author       = {Vladimir Lifschitz},
  editor       = {John Mylopoulos and
                  Raymond Reiter},
  title        = {Nonmonotonic Databases and Epistemic Queries},
  booktitle    = {Proceedings of the 12th International Joint Conference on Artificial
                  Intelligence. Sydney, Australia, August 24-30, 1991},
  pages        = {381--386},
  publisher    = {Morgan Kaufmann},
  year         = {1991},
  url          = {http://ijcai.org/Proceedings/91-1/Papers/059.pdf},
  bibsource    = {dblp computer science bibliography, https://dblp.org}
}

@article{DBLP:journals/ai/KnorrAH11,
  author       = {Matthias Knorr and
                  Jos{\'{e}} J{\'{u}}lio Alferes and
                  Pascal Hitzler},
  title        = {Local closed world reasoning with description logics under the well-founded
                  semantics},
  journal      = {Artif. Intell.},
  volume       = {175},
  number       = {9-10},
  pages        = {1528--1554},
  year         = {2011},
  url          = {https://doi.org/10.1016/j.artint.2011.01.007},
  doi          = {10.1016/J.ARTINT.2011.01.007},
  bibsource    = {dblp computer science bibliography, https://dblp.org}
}

@article{DBLP:journals/jacm/MotikR10,
  author       = {Boris Motik and
                  Riccardo Rosati},
  title        = {Reconciling description logics and rules},
  journal      = {J. {ACM}},
  volume       = {57},
  number       = {5},
  pages        = {30:1--30:62},
  year         = {2010},
  url          = {https://doi.org/10.1145/1754399.1754403},
  doi          = {10.1145/1754399.1754403},
  bibsource    = {dblp computer science bibliography, https://dblp.org}
}

@article{DBLP:journals/lmcs/LutzSW19,
  author       = {Carsten Lutz and
                  Inan{\c{c}} Seylan and
                  Frank Wolter},
  title        = {The Data Complexity of Ontology-Mediated Queries with Closed Predicates},
  journal      = {Log. Methods Comput. Sci.},
  volume       = {15},
  number       = {3},
  year         = {2019},
  url          = {https://doi.org/10.23638/LMCS-15(3:23)2019},
  doi          = {10.23638/LMCS-15(3:23)2019},
  bibsource    = {dblp computer science bibliography, https://dblp.org}
}

@inproceedings{DBLP:conf/aaai/LukumbuzyaOS20,
  author       = {Sanja Lukumbuzya and
                  Magdalena Ortiz and
                  Mantas Simkus},
  title        = {Resilient Logic Programs: Answer Set Programs Challenged by Ontologies},
  booktitle    = {The Thirty-Fourth {AAAI} Conference on Artificial Intelligence, {AAAI}
                  2020, The Thirty-Second Innovative Applications of Artificial Intelligence
                  Conference, {IAAI} 2020, The Tenth {AAAI} Symposium on Educational
                  Advances in Artificial Intelligence, {EAAI} 2020, New York, NY, USA,
                  February 7-12, 2020},
  pages        = {2917--2924},
  publisher    = {{AAAI} Press},
  year         = {2020},
  url          = {https://doi.org/10.1609/aaai.v34i03.5683},
  doi          = {10.1609/AAAI.V34I03.5683},
  bibsource    = {dblp computer science bibliography, https://dblp.org}
}

@incollection{existential-rules-cali,
  author       = {Andrea Cal{\`{\i}} and
                  Georg Gottlob and
                  Thomas Lukasiewicz},
  editor       = {Roberto De Virgilio and
                  Fausto Giunchiglia and
                  Letizia Tanca},
  title        = {Datalog Extensions for Tractable Query Answering over Ontologies},
  booktitle    = {Semantic Web Information Management - {A} Model-Based Perspective},
  pages        = {249--279},
  publisher    = {Springer},
  year         = {2009},
  url          = {https://doi.org/10.1007/978-3-642-04329-1\_12},
  doi          = {10.1007/978-3-642-04329-1\_12},
  bibsource    = {dblp computer science bibliography, https://dblp.org}
}

@article{cali-framework-2012,
title = {A general Datalog-based framework for tractable query answering over ontologies},
journal = {Journal of Web Semantics},
volume = {14},
pages = {57-83},
year = {2012},
note = {Special Issue on Dealing with the Messiness of the Web of Data},
issn = {1570-8268},
doi = {https://doi.org/10.1016/j.websem.2012.03.001},
url = {https://www.sciencedirect.com/science/article/pii/S1570826812000388},
author = {Andrea Calì and Georg Gottlob and Thomas Lukasiewicz},
}

@inproceedings{calvanese-ontologies-and-databases,
  author       = {Diego Calvanese and
                  Giuseppe De Giacomo and
                  Domenico Lembo and
                  Maurizio Lenzerini and
                  Antonella Poggi and
                  Mariano Rodriguez{-}Muro and
                  Riccardo Rosati},
  editor       = {Sergio Tessaris and
                  Enrico Franconi and
                  Thomas Eiter and
                  Claudio Gutierrez and
                  Siegfried Handschuh and
                  Marie{-}Christine Rousset and
                  Renate A. Schmidt},
  title        = {Ontologies and Databases: The DL-Lite Approach},
  booktitle    = {Reasoning Web. Semantic Technologies for Information Systems, 5th
                  International Summer School 2009, Brixen-Bressanone, Italy, August
                  30 - September 4, 2009, Tutorial Lectures},
  series       = {Lecture Notes in Computer Science},
  volume       = {5689},
  pages        = {255--356},
  publisher    = {Springer},
  year         = {2009},
  doi          = {10.1007/978-3-642-03754-2\_7},
}

@article{DBLP:journals/jcss/CalvaneseGLR12,
  author       = {Diego Calvanese and
                  Giuseppe De Giacomo and
                  Maurizio Lenzerini and
                  Riccardo Rosati},
  title        = {View-based query answering in Description Logics: Semantics and complexity},
  journal      = {J. Comput. Syst. Sci.},
  volume       = {78},
  number       = {1},
  pages        = {26--46},
  year         = {2012},
  url          = {https://doi.org/10.1016/j.jcss.2011.02.011},
  doi          = {10.1016/J.JCSS.2011.02.011},
  bibsource    = {dblp computer science bibliography, https://dblp.org}
}

@inproceedings{DBLP:conf/sdb/Perez-UrbinaMH08,
  author       = {H{\'{e}}ctor P{\'{e}}rez{-}Urbina and
                  Boris Motik and
                  Ian Horrocks},
  editor       = {Klaus{-}Dieter Schewe and
                  Bernhard Thalheim},
  title        = {Rewriting Conjunctive Queries over Description Logic Knowledge Bases},
  booktitle    = {Semantics in Data and Knowledge Bases, Third International Workshop,
                  {SDKB} 2008, Nantes, France, March 29, 2008, Revised Selected Papers},
  series       = {Lecture Notes in Computer Science},
  volume       = {4925},
  pages        = {199--214},
  publisher    = {Springer},
  year         = {2008},
  url          = {https://doi.org/10.1007/978-3-540-88594-8\_11},
  doi          = {10.1007/978-3-540-88594-8\_11},
  bibsource    = {dblp computer science bibliography, https://dblp.org}
}

@inproceedings{amendola-finite-control,
  author       = {Giovanni Amendola and
                  Nicola Leone and
                  Marco Manna},
  editor       = {J{\'{e}}r{\^{o}}me Lang},
  title        = {Finite Controllability of Conjunctive Query Answering with Existential
                  Rules: Two Steps Forward},
  booktitle    = {Proceedings of the Twenty-Seventh International Joint Conference on
                  Artificial Intelligence, {IJCAI} 2018, July 13-19, 2018, Stockholm,
                  Sweden},
  pages        = {5189--5193},
  publisher    = {ijcai.org},
  year         = {2018},
  doi          = {10.24963/IJCAI.2018/719},
  }

@inproceedings{DBLP:conf/ijcai/AmendolaLMV18,
  author       = {Giovanni Amendola and
                  Nicola Leone and
                  Marco Manna and
                  Pierfrancesco Veltri},
  editor       = {J{\'{e}}r{\^{o}}me Lang},
  title        = {Enhancing Existential Rules by Closed-World Variables},
  booktitle    = {Proceedings of the Twenty-Seventh International Joint Conference on
                  Artificial Intelligence, {IJCAI} 2018, July 13-19, 2018, Stockholm,
                  Sweden},
  pages        = {1676--1682},
  publisher    = {ijcai.org},
  year         = {2018},
  url          = {https://doi.org/10.24963/ijcai.2018/232},
  doi          = {10.24963/IJCAI.2018/232},
  bibsource    = {dblp computer science bibliography, https://dblp.org}
}

@article{DBLP:journals/dase/QinZYWFX21,
  author       = {Xiaoyu Qin and
                  Xiaowang Zhang and
                  Muhammad Qasim Yasin and
                  Shujun Wang and
                  Zhiyong Feng and
                  Guohui Xiao},
  title        = {{SUMA:} {A} Partial Materialization-Based Scalable Query Answering
                  in {OWL} 2 {DL}},
  journal      = {Data Sci. Eng.},
  volume       = {6},
  number       = {2},
  pages        = {229--245},
  year         = {2021},
  url          = {https://doi.org/10.1007/s41019-020-00150-0},
  doi          = {10.1007/S41019-020-00150-0},
  bibsource    = {dblp computer science bibliography, https://dblp.org}
}

@inproceedings{danecki1984propositional,
  title={Propositional dynamic logic with strong loop predicate},
  author={Danecki, Ryszard},
  booktitle={International Symposium on Mathematical Foundations of Computer Science},
  pages={573--581},
  year={1984},
  organization={Springer}
}

@article{gogacz-control-17,
title = {Converging to the chase – A tool for finite controllability},
journal = {Journal of Computer and System Sciences},
volume = {83},
number = {1},
pages = {180-206},
year = {2017},
issn = {0022-0000},
doi = {https://doi.org/10.1016/j.jcss.2016.08.001},
author = {Tomasz Gogacz and Jerzy Marcinkowski},
}

@inproceedings{figueira-control-20,
  author       = {Diego Figueira and
                  Santiago Figueira and
                  Edwin Pin Baque},
  editor       = {Diego Calvanese and
                  Esra Erdem and
                  Michael Thielscher},
  title        = {Finite Controllability for Ontology-Mediated Query Answering of {CRPQ}},
  booktitle    = {Proceedings of the 17th International Conference on Principles of
                  Knowledge Representation and Reasoning, {KR} 2020, Rhodes, Greece,
                  September 12-18, 2020},
  pages        = {381--391},
  year         = {2020},
  doi          = {10.24963/KR.2020/39},
}

@inproceedings{BednarczykK22,
  author       = {Bartosz Bednarczyk and
                  Emanuel Kieronski},
  title        = {Finite Entailment of Local Queries in the {Z} Family of Description
                  Logics},
  booktitle    = {Thirty-Sixth {AAAI} Conference on Artificial Intelligence, {AAAI}
                  2022, Thirty-Fourth Conference on Innovative Applications of Artificial
                  Intelligence, {IAAI} 2022, The Twelveth Symposium on Educational Advances
                  in Artificial Intelligence, {EAAI} 2022 Virtual Event, February 22
                  - March 1, 2022},
  pages        = {5487--5494},
  publisher    = {{AAAI} Press},
  year         = {2022},
  url          = {https://doi.org/10.1609/aaai.v36i5.20487},
  doi          = {10.1609/AAAI.V36I5.20487}
}

@article{garcia-finite-s-arxiv,
  author       = {Yazm{\'{\i}}n Ib{\'{a}}{\~{n}}ez{-}Garc{\'{\i}}a and
                  Jean Christoph Jung and
                  Vincent Michielini and
                  Filip Murlak},
  title        = {Revisiting Conjunctive Query Entailment for {S}},
  journal      = {CoRR},
  volume       = {abs/2511.07933},
  year         = {2025},
  doi          = {10.48550/ARXIV.2511.07933},
  eprinttype    = {arXiv},
  eprint       = {2511.07933},
}

@article{DBLP:journals/ai/Gutierrez-Basulto23,
  author       = {V{\'{\i}}ctor Guti{\'{e}}rrez{-}Basulto and
                  Yazm{\'{\i}}n Ib{\'{a}}{\~{n}}ez{-}Garc{\'{\i}}a and
                  Jean Christoph Jung and
                  Filip Murlak},
  title        = {Answering regular path queries mediated by unrestricted {SQ} ontologies},
  journal      = {Artif. Intell.},
  volume       = {314},
  pages        = {103808},
  year         = {2023},
  url          = {https://doi.org/10.1016/j.artint.2022.103808},
  doi          = {10.1016/J.ARTINT.2022.103808},
  bibsource    = {dblp computer science bibliography, https://dblp.org}
}

@inproceedings{eiter-unrestricted-s,
author = {Eiter, Thomas and Lutz, Carsten and Ortiz, Magdalena and \v{S}imkus, Mantas},
title = {Query answering in description logics with transitive roles},
year = {2009},
publisher = {Morgan Kaufmann Publishers Inc.},
address = {San Francisco, CA, USA},
booktitle = {Proceedings of the 21st International Joint Conference on Artificial Intelligence},
pages = {759–764},
numpages = {6},
location = {Pasadena, California, USA},
series = {IJCAI'09}
}

@article{shiq-infinite-s,
  author       = {Birte Glimm and
                  Carsten Lutz and
                  Ian Horrocks and
                  Ulrike Sattler},
  title        = {Conjunctive Query Answering for the Description Logic {SHIQ}},
  journal      = {J. Artif. Intell. Res.},
  volume       = {31},
  pages        = {157--204},
  year         = {2008},
  doi          = {10.1613/JAIR.2372},
}

@inproceedings{shoq-infinite-s,
  author       = {Birte Glimm and
                  Ian Horrocks and
                  Ulrike Sattler},
  editor       = {Diego Calvanese and
                  Enrico Franconi and
                  Volker Haarslev and
                  Domenico Lembo and
                  Boris Motik and
                  Anni{-}Yasmin Turhan and
                  Sergio Tessaris},
  title        = {Conjunctive Query Entailment for {SHOQ}},
  booktitle    = {Proceedings of the 2007 International Workshop on Description Logics
                  (DL2007), Brixen-Bressanone, near Bozen-Bolzano, Italy, 8-10 June,
                  2007},
  series       = {{CEUR} Workshop Proceedings},
  volume       = {250},
  publisher    = {CEUR-WS.org},
  year         = {2007},
}

@article{konig-lemma,
    author = {Dénes Kőnig},
    title = {Über eine Schlussweise aus dem Endlichen ins Unendliche.},
    journal = { Acta Sci. Math. (Szeged), 3(2-3):121–130},
    year = {1927}
}

@book{introduction-to-graph-theory,
    author = {Robin J. Wilson},
    title = {Introduction to Graph Theory, Wilson},
    publisher = {Longman},
    year = {2010}
}


\appendix

\clearpage
\section{Missing proofs for Section~\ref{sec:single-role}}

\subsection{Proof of Lemma~\ref{lem:strong}}

\lemStrong*
\begin{proof}
We first show that if $\intp' \models \TBox$ then $\intp \models \TBox$. 

By definition of TBox satisfaction, $\intp' \models \TBox$ if for each concept inclusion $C \sqsubseteq D$, $C^{\intp} \subseteq D^{\intp}$. We show that for all three normal form concept inclusion shapes, if $\intp' \models C \sqsubseteq D$ then $\intp \models C \sqsubseteq D$. 

\begin{itemize}
    \item If the inclusion is of the shape $A_1 \sqcap A_2 \sqcap \ldots \sqcap A_n \sqsubseteq B_1 \sqcup B_2 \sqcup \ldots \sqcup B_m$, then for all $a \in A_1^{\intp'} \sqcap A_2^{\intp'} \sqcap \ldots \sqcap A_n^{\intp'}$ there exists a concept $C$ such that $a' \in C^{\intp'}$ and $C = B_i$ for some $1 \leq i \leq m$. Let $a \in \Delta^{\intp}$ such that $h(a) = a'$. Since h is $t$-strong, we have that $\types{\intp}{a} = \types{\intp'}{a'}$, and so in particular $a$ must belong to $_1^{\intp} \sqcap A_2^{\intp} \sqcap \ldots \sqcap A_n^{\intp}$ and $C^{\intp}$. So the inclusion is satisfied for $\intp$. 

    \item If the inclusion is of the shape $A \sqsubseteq \forall t.B$, assume for contradiction that there exist two elements $a, b$ in $\Delta^{\intp}$ such that $a \in A^{\intp}, (a,b) \in t^{\intp}$ but $b \not\in B^{\intp}$. Let $a', b'$ be the two elements of $\Delta^{\intp'}$ such that $h(a) = a'$ and $h(b) = b'$. Since $h$ is a $t$-strong homomorphism, we know that $\types{\intp}{x} = \types{\intp'}{h(x)}$ for all $x \in \Delta^{\intp}$ and for all $(x, y) \in t^{\intp}$ it must also be the case that $(h(x), h(y)) \in t^{\intp'}$. In particular, we have that $a' \in A^{\intp'}, (a',b') \in t^{\intp'}$ and, since $\intp'$ satisfies the inclusion, $b'$ must be in $B^{\intp'}$. Since $b$ has the same type as $b'$, it must be in $B^{\intp}$. This contradicts our assumption so the inclusion is satisfied for $\intp$. 

    \item If the inclusion is of the shape $A \subseteq \exists t.B$, let $a \in \Delta^{\intp}, a' \in \Delta^{\intp'}$ be two elements such that $h(a) = a'$ and $a' \in A^{\intp'}$. Since $\intp'$ satisfies the inclusion, there exists an element $b' \in B^{\intp'}$ such that $(a',b') \in t^{\intp'}$. Note that this implies $B \in \rch{t}{\intp'}(a')$. Since $h$ is $t$-strong, $B$ must also be in $\rch{t}{\intp}(a)$. So there is a path via $t$-edges from $a$ to an element $b \in B^{\intp}$. Since $t$ is transitive, this path is also witnessed by an edge, so we have that $(a, b) \in t^{\intp}$ and the inclusion is satisfied for $\intp$.
\end{itemize}

We now show that if $\intp' \not\models Q$ then $\intp \not\models Q$. Without loss of generality we can assume that $Q$ is a single conjunctive query.
Assume for contradiction that $\intp \models Q$. Then there exists a match $\pi$ of $Q$ in $\intp$ that maps each variable of $Q$ to an element of $\Delta^{\intp}$. Let $\pi'$ be the function defined by $\pi'(x) = h(\pi(x))$. We show that if $\pi$ is a match of $Q$ in $\intp$ then $\pi'$ is a match of $Q$ in $\intp'$. Let $q$ be an atom of $Q$. 

If $q$ is of the shape $A(x)$ then there exists an element $a \in \Delta^{\intp}$ such that $\pi(x) = a$ and $a \in A^{\intp}$. Let $a' = \pi'(x)$. Since $h$ is $t$-strong, $a'$ must be in $A^{\intp'}$ and so the atom is also satisfied in $\intp'$.

If $q$ is of the shape $t(x,y)$ then there exist two elements $b, c \in \Delta^{\intp}$ such that $\pi(x) = b, \pi(y) = c$ and $(b,c) \in t^{\intp}$. Let $b' = \pi'(b)$ and $c' = \pi'(c)$. Since $h$ is a homomorphism, we know that $(h(b), h(c)) \in t^{\intp'}$ so $(b',c')\in t^{\intp'}$ and the atom is satisfied in $\intp'$. 

Since all atoms of $Q$ are satisfied in $\intp'$, we have shown that $\intp' \models Q$ which contradicts our initial assumption. 

Finally, we show that if If $h$ is surjective and $\intp\models \TBox$, then $\intp' \models \TBox$. We proceed as above, by showing that this holds for each normal-form concept inclusion shapes. 

\begin{itemize}
    \item If the inclusion is of the shape $A_1 \sqcap A_2 \sqcap \ldots \sqcap A_n \sqsubseteq B_1 \sqcup B_2 \sqcup \ldots \sqcup B_m$ then for all $a \in A_1^{\intp} \sqcap A_2^{\intp} \sqcap \ldots \sqcap A_n^{\intp}$ there exists a concept $C$ such that $a \in C^{\intp}$ and $C = B_i$ for some $1 \leq i \leq m$. Let $a' = h(a)$. Since $h$ is a $t$-strong homomorphism, the type of $a'$ is preserved and so in particular $a'$ must belong to $A_1^{\intp'} \sqcap A_2^{\intp'} \sqcap \ldots \sqcap A_n^{\intp'}$ and $C^{\intp'}$. So the inclusion is satisfied for $\intp'$. 
    
    \item If the inclusion is of the shape $A \sqsubseteq \forall t.B$, then for all pairs $(a,b) \in t^{\intp}$ such that $a \in A^{\intp}$, we have that $b \in B^{\intp}$. Let $(a',b')$ be a pair of elements of $t^{\intp'}$ such that $a \in A^{\intp'}$ but $b \not\in B^{\intp'}$. As $b$ and $b'$ do not have the same type, and since $h$ is a $t$-strong homomorphism, $b'$ cannot be the image of $b$ in $h$. Since $h$ is assumed to be surjective, there must exist another element $c \in \Delta^{\intp}$ such that $b' = h(c)$. As edges are preserved by homomorphism, we have that $(a,c) \in t^{\intp}$ and since $\intp \models \TBox$, $c \in B^{\intp}$. As $h$ is $t$-strong, this implies that $b' \in B^{\intp'}$ which contradicts our assumption so the inclusion is also satisfied in $\intp'$.

    \item If the inclusion is of the shape $A \sqsubseteq \exists t.B$, then for all $a \in A^{\intp}$, there exists a $b \in B^{\intp}$ such that $(a,b) \in t^{\intp}$. Let $a' = h(a)$ and $b' = h(b)$. Since $h$ is a $t$-strong homomorphism, we know that $\types{\intp}{x} = \types{\intp'}{h(x)}$ for all $x \in \Delta^{\intp}$ and for all $(x, y) \in t^{\intp}$ it must also be the case that $(h(x), h(y)) \in t^{\intp'}$. In particular we have that $a' \in A^{\intp'}, b' \in B^{\intp'}$ and $(a',b') \in t^{\intp'}$, So the inclusion is satisfied for $\intp'$.
\end{itemize}

\end{proof}

\subsection{Proof of Lemma~\ref{lem:unravelling}}

\lemUnravelling*

\begin{proof}
First, notice that since the construction of $\intp^*$ does not add new elements to the domain, $\Delta^{\intp^*} = \Delta^{\intp}$ so we refer to both domains without distinction. We will also use the fact that reachable concepts and element types are preserved by transitive closure.

We refer to elements of $\Delta^{\widetilde \intp}$ obtained from $\Delta^{\intp}_F \cup \{a\}$ as \emph{single-element} nodes, and to those of the form $(K, p)$ as \emph{cluster} nodes. 

Let $h$ be the function from $\Delta^{\widetilde \intp}$ to $\Delta^{\intp^{*}}$ defined as follows:
\begin{itemize}
    \item For $e \in \Delta^{\widetilde \intp}$ a single-element node, $h(e) = e$
    \item For $e \in \Delta^{\widetilde \intp}$ a cluster node of the shape $(K, pu)$ (with $p$ possibly empty), $h(e) = u$
\end{itemize}

We show that $h$ is 
\begin{enumerate}
    \item[(1).] a homomorphism from $(\widetilde \intp, a, \widetilde \lambda)$ to $(\intp^*, a, \lambda)$,
    \item[(2).] $t$-strong,
    \item[(3.)] surjective if all elements in $\intp$ are reachable from $a$.
\end{enumerate}

(1). Since, by definition of $h$, for all single-element nodes $h(e)=e$ and, by definition of quasi-unravelling, types are preserved for single-element nodes, we immediately have that $e \in A^{\widetilde \intp}$ implies $h(e) \in A^{\intp^{*}}$. For cluster nodes of the shape $(K, pu)$, the type is determined by the last element of the path $p$, i.e. $(K, pu) \in A^{\widetilde \intp}$ if $u \in A^{\intp}$. Since $h((K, pu)) = u$, types are preserved by $h$ for cluster nodes, and so also for all elements of the domain of $\widetilde \intp$.

Let $(e,e')$ be an edge in $\widetilde \intp$. By construction of $\widetilde \intp$, edges can occur when:
\begin{itemize}
    \item $e$ and $e'$ are both single-element nodes. In this case, $(h(e), h(e')) \in t^{\intp^{*}}$ by def of $t^{\widetilde \intp}$. 
    \item $e$ is a single-element node from some critical cluster $K$ and $e'$ is a cluster node of the shape $(K, pv)$. In this case, $h(e) = e$ and $h((K, pv)) = v$ and $v$ is reachable from the successor of some element in $K$ via $p$ in $\intp$. Since $K$ is a clique, all elements of $K$ are reachable from $e$, in particular the predecessor of $p$ in $K$, and so $v$ is also reachable from $e$. As $\intp^{*}$ is the transitive closure of the unique role $t$, there is a $t$-edge from $e$ to $v$, i.e. from $h(e)$ to $h((K, pv))$.
    \item $e$ are both cluster nodes of the shape $(K, pu)$ and $(K, puv)$ respectively. By construction of $\widetilde \intp$, the node $(K, puv)$ belongs to $\Delta^{\widetilde \intp}$ if $puv$ is a path of non-critical elements in $\intp$. In particular, there must be an edge between $u$ and $v$, and since $h((K, pu)) = u$, and $h((K, puv)) = v$, we have that $(h((K, pu)), h((K, puv)) \in t^{\intp^{*}}$.
    \item $e$ is a cluster node of the shape $(K, pu)$, $e'$ is a single-element node from some critical cluster $K'$ and $(u,v)$ is an edge in $\intp$. Since $h((K, pu)) = u$ and $h(v) =v$, and edges are preserved by transitive closure, we have that $(h((K, pu)), h(v) \in t^{\intp^{*}}$.
\end{itemize}

(2). Since $h$ is a homomorphism, we have that $\types{\widetilde \intp}{e} \subseteq \types{\intp^{*}}{h(e)}$ and $\rch{t}{\widetilde \intp}(e) \subseteq \rch{t}{\intp^{*}}(h(e))$. As concept names in $\widetilde \intp$ are directly inherited from $\intp$, we immediately have $\types{\widetilde \intp}{e} = \types{\intp^{*}}{h(e)}$. 

To see that $\rch{t}{\widetilde \intp}(a) = \rch{t}{\intp^{*}}(h(a))$, assume for contradiction that there is a concept name $C$ reachable in $\intp^{*}$ from some element $f$ that is not reachable in $\widetilde \intp$ for some $e$ such that $h(e) = f$. By definition of $\rch{}{}$, there exists an edge in $\intp^{*}$ from $f$ to an element $f'$ such that $f' \in C^{\intp^{*}}$. We show case-by-case that this leads to a contradiction.

\begin{itemize}
    \item If $f$ is a critical element of $\Delta^{\intp^{*}}$ then $e$ is a single-element node. \begin{itemize}
        \item If $f'$ is also a critical element of $\Delta^{\intp^{*}}$ then all $e' \in \Delta^{\widetilde \intp}$ such that $h(e') = f'$ are single-element nodes. Since $(f,f') \in t^{\intp^{*}}$, by definition of $t^{\widetilde \intp}$ we have that $(e,e') \in t^{\widetilde \intp}$. As $h(e') = f'$ and types are preserved under $h$ (as shown above), we have that $\types{\widetilde \intp}{e'} = \types{\intp^{*}}{f'}$ and, in particular, $e' \in C^{\widetilde \intp}$.
        \item If $f'$ is not a critical element, then let $K_f$ be the cluster of $f$ in $\Delta^{\intp^{*}}$. Since $f'$ is reachable from $f$, the pair $(K_f, f')$ is a legal node of $\Delta^{\widetilde \intp}$, and by definition of $t^{\widetilde \intp}$, there is an edge between $e$ and $(K_f, f')$. As $h((K_f, f')) = f'$ and types are preserved under $h$, we get that $(K_f, f') \in C^{\widetilde \intp}$.
    \end{itemize}

    \item If $f$ is not a critical element, them $e$ is a cluster node of the shape $(K, pf)$ for some $K$ and (possibly empty) $p$.

    \begin{itemize}
        \item If $f'$ is also not a critical element, then $(K, pff')$ is a legal cluster node of $\Delta^{\widetilde \intp}$, and by definition of $t^{\widetilde \intp}$, there is an edge $( (K, pf), (K, pff'))$ in $t^{\widetilde \intp}$. Once again by preservation of types under $h$, we get that $(K, pff') \in C^{\widetilde \intp}$.
        \item If $f'$ is a critical element, then all $e' \in \Delta^{\widetilde \intp}$ such that $h(e') = f'$ are single-element nodes. By definition of $t^{\widetilde \intp}$, since $(f, f') \in t^{\intp^{*}}$, $(K, pf) \in \Delta^{\widetilde \intp}$ and $f' \in \Delta_F^{\intp^{*}}$, there is an edge between $e$ and $e'$ in $t^{\widetilde \intp}$ and by preservation of types under $h$, $e' \in C^{\widetilde \intp}$.
    \end{itemize}
\end{itemize}

(3.) Assume for contradiction that all elements in $\intp$ are reachable from $a$ but $h$ is not surjective, i.e. there exists an element $f \in \Delta^{\intp^{*}}$ such that for all elements $e \in \Delta^{\widetilde \intp}, ~ h(e) \neq f$. Notice that since $\intp^{*}$ is transitively closed, whenever an element $a$ can reach an element $b$ in $\intp$, there is an edge $(a,b) \in t^{\intp^{*}}$. If $f$ is a critical element of $\intp$, then by definition of quasi-unravelling, $f$ is also an element of $\Delta^{\widetilde \intp}$ and, by definition of $h$, $h(f) = f$. Moreover, since $f$ is reachable from $a$, there is an edge $(a,f) \in t^{\widetilde \intp}$. If $f$ is not a critical element, then once again by definition of quasi-unravelling, $(K^{\intp}_a, f)$ is a valid node of $\Delta^{\widetilde \intp}$ and there is an edge $(a, (K^{\intp}_a, f)) \in t^{\widetilde \intp}$. By definition of $h$, $h((K^{\intp}_a, f)) = f$.
\end{proof}

\subsection{Proof of Lemma~\ref{cor:unravelling}}

\CorUnravelling*

\begin{proof}
This is a direct consequence of Lemmas~\ref{lem:strong} and ~\ref{lem:unravelling}. Indeed, by Lemma~\ref{lem:unravelling} and the hypothesis that all elements are reachable from $a$, we know that there exists a surjective $t$-strong homomorphism from $(\widetilde{\mathcal{I}}, a)$ to $(\mathcal{I}^{*}, a)$. By instantiating the implications of Lemma~\ref{lem:strong} with $(\widetilde{\mathcal{I}}, a)$ for $\mathcal{I}$ and $(\mathcal{I}^{*}, a)$ for $\mathcal{I}'$, we immediately get the above claim.
\end{proof}

\subsection{Proof of Lemma~\ref{lem:shortPaths}}

\lemShortPaths*

\begin{proof} We call an element $u$ of $\intp$ \emph{dispensable} if $\types{\intp}{u} \subseteq \rch{t}{\intp}(u)$, and \emph{indispensable} otherwise. {Because $t$ is transitive,} removing a~dispensable element does not affect concepts reachable from other elements: if $u \in \Delta^\intp$ is dispensable and  $\intp'$ is the  restriction of $\intp$ to $\Delta^\intp \setminus \{u\}$, then  $\rch{t}{\intp'}(v) = \rch{t}{\intp}(v)$ for all $v\in \Delta^{\intp'}$.

We claim that every simple path that only visits indispensable elements has length at most $|\conceptsIn{\intp}|$. Indeed, consider a~simple path $u_1 u_2 \dots u_m$ with $u_1, u_2, \dots, u_m$ indispensable. Because $u_{i+1}, u_{i+2}, \dots, u_m$ are successors of $u_i$ and are different from $u_i$, we have $\types{\intp}{u_i} \not\subseteq \bigcup_{j=i+1}^m \types{\intp}{u_j} \subseteq \rch{t}{\intp}(u_i)$ for all $i$. In consequence, 
\[\conceptsIn{\TBox} \supseteq \bigcup_{j=1}^m \types{\intp}{u_j} \supset \dots \supset \bigcup_{j=m}^m \types{\intp}{u_j} \supset \emptyset\,.\]
(The last strict inclusion holds because elements with empty types are dispensable.) This implies that $m \leq |\conceptsIn{\TBox}|$.


Now it suffices to ensure that $\JJ$ contains no dispensable critical elements different from $b$. We can construct a~suitable $\JJ$ by removing from $\intp$ dispensable critical elements different from $b$ one by one, until none are left.
\end{proof}
\clearpage
\section{Missing proofs for Section~\ref{sec:PartiallyfiniteTheorem}}
\label{app:PartiallyfiniteTheorem}
\subsection{Proof of Theorem \ref{thm:pfColourBlocking}}

Our goal is to show the following theorem.


\thmColourBlocking*

The rest of this section is devoted to a proof of this theorem.
{For $\distN = 0$ the statement is clear. Assume that $\distN > 0$ and }
that $\intp$ is coloured by $\kappa\from \dom(\kappa)\to C$ in a~$K$\=/sparse way.

Due to the existence of the {natural} quotient homomorphism $h_{\sim_K}\from \intp\to\intp/{\sim_K}$,
the ``only if'' part of the implication is 
clearly true. 
{ For the other direction, let } $h\from \intq \to\intp/{\sim_K}$ 
be a~homomorphism, we will show that there exists a~homomorphism from $\intq$ into $\intp$.

{ We construct the desired homomorphism inductively. Starting from single elements of $\intq$, we combine the partial homomorphism to bigger pieces, finally extending them to the whole interpretation. Before we describe the inductive step of the construction, we define some useful structures and operations used in the construction.}

\paragraph*{Pieces.}
We 
{ start with} turning $\intq$ into a~coloured interpretation by defining a~colouring function $\theta$ from a~subset of $\Delta^\intq$ into~$C${  in a way that agrees with both the homomorphism $h$ and colouring $\kappa$}.

For every element $e\in \Delta^\intq$ consider the value $h(e)$. If~$h(e)$ is a~non\=/coloured
element in $\intp/{\sim_K}$ { then, by definition of colour preserving,  $[h(e)]_{\sim_K} = \{h(e)\}$}
and $e$ is a~non\=/coloured element of $\intq$, i.e.~$e\notin\dom(\theta)$. On the other hand, if $h(e)$~is coloured then
$h(e)=[x]_{\sim_K}$ for some $x\in\dom(\kappa)$ and we define $\theta(e)$ as $\kappa(x)$ {, which makes}~$e$ a~coloured element of~$\intq$.

{ We} divide { the domain} $\Delta^\intq$ into two disjoint parts:
$\Del$~and $\Nel$, where $\Del \eqdef \dom(\theta)$ contains coloured elements and $\Nel~\eqdef~\Delta^\intq\setminus \Del$ contains non\=/coloured ones { in $\intq$}.

A~\emph{piece} is any subinterpretation $\intq'$ of $\intq$ such that if~$e~\in~\Delta^{\intq'}$ then for every concept name $A$ we have $e\in A^{\intq'}\Leftrightarrow e\in A^{\intq}$, i.e.~the unary types must be the same. Thus, a~piece is determined by its domain and the interpretations~$r^{\intq'}$ of 
all relevant role names $r$. 
{   Of course, since $\intq'$ is a subinterpretation of $\intq$, necessarily $r^{\intq'}\subseteq r^{\intq}$ for every role $r$. Moreover, we can extend the colouring $\theta$ to pieces of $\intq$ by restricting $\theta$ the respective domains. Hence, from now on, we will also assume that pieces of $\intq$ are coloured by~$\theta$.}

\paragraph*{Unions.}

During the induction process we will 
{combine pieces to expand the homomorphism in an iterative process. To combine pieces we will use a union operation.} 
A~union {$\intq_1\cup\intq_2$} of two interpretations $\intq_1$ and $\intq_2$ 
is the interpretation 
{ with}
\begin{itemize}\itemsep=0pt
    \item the domain $\Delta^{\intq_1\cup\intq_2} \eqdef \Delta^{\intq_1}\cup\Delta^{\intq_2}$, where
    \item for each concept name $A$, $A^{\intq_1\cup\intq_2}\eqdef A^{\intq_1}\cup A^{\intq_2}$,
    \item and for each role name $r$, $r^{\intq_1\cup\intq_2}\eqdef r^{\intq_1}\cup r^{\intq_2}$.
\end{itemize}
The domains $\Delta^{\intq_1}$ and $\Delta^{\intq_2}$ {are not necessarily} disjoint. However, if $\intq_1$ and $\intq_2$ are both pieces of $\intq$ then, by definition of a piece, $A^{\intq_i}=A^\intq\cap \Delta^{\intq_i}$ for $i=1,2$. 
Therefore, the unary types of elements are preserved when taking unions of pieces: {for $i=1,2$, for every concept name $A$, and every element $e \in \Delta^{\intq_i}$: $e \in A^{\intq_1\cup\intq_2}$ if and only if $e\in A^{\intq_i}$.}

We extend the above definition to arbitrary unions of interpretations { in the natural way}.

\paragraph*{Basic set.}

{Now we prepare the base step of the inductive construction.}
For this reason, we define a~certain set of \emph{basic pieces} of $\intq$, which will be used as a~base 
of the construction.

{For each coloured element $e\in \Del$ define a~trivial piece $\intq_e$ with domain $\{e\}$ and empty interpretation of every role name $r$, i.e. for every role name $r^{\intq_e} = \emptyset$.}
Intuitively, this piece consists of this single element and no roles.

Let $R$ be the set of pairs $(e,e')\in \Del\times \Del$ where $e,e'$ are coloured elements connected by an~edge (in any direction). For each $(e,e')\in R$ (possibly $e=e'$), define the~piece $\intq_{(e,e')}$ as $\intq\restr \{e,e'\}$. This piece contains both coloured elements $e$, $e'$ and all edges which connect them, i.e.~for every role name $r$, {$r^{\intq'}=r^{\intq}\,\restr\,{\{e,e'\}}$}.


{
For each non\=/coloured element $e\in \Nel$ define
$\intq_e$ as the interpretation $\intq\restr (\Nel(e)\cup \Del(e))$ \textbf{with all edges inside $\Del(e){\times}\Del(e)$ removed}, where
\begin{itemize}\itemsep=0pt
    \item $\Nel(e)$ is the set of all non\=/coloured elements $e'$ such $d(e,e')<\infty$ in $\intq\restr \Nel$, i.e.~there is a path between $e$ and $e'$ which only uses non\=/coloured elements of $\intq$;
    \item and $\Del(e)$ is the set of all coloured elements $f\in \Del$ such that there is an~element $e'\in \Nel(e)$ such that $f$ and $e'$ are connected by an~edge.
\end{itemize}

}


{
Note that for every pair of non\=/coloured elements $e_1,e_2$, either the neighbourhoods $\Nel(e_1)$ and $\Nel(e_2)$ are disjoint or $\Nel(e_1) \cap \Nel(e_2) \neq \emptyset$ and the pieces $\intq_{e_1}$, $\intq_{e_2}$ coincide, i.e. $\intq_{e_1} =\intq_{e_2}$.

The remaining types of basic pieces, i.e. $\intq_e$ for $e\in \Del$ and $\intq_{(e,e')}$, contain only coloured elements.
}

Consider the set of pieces
\[P=\{\intq_e\mid e\in\Delta^\intq\}\cup \{\intq_{(e,e')}\mid (e, e')\in R\}.\]
Here, the first set contains pieces of two types, depending on whether $e\in \Del$ or $e\in \Nel$, and the second set contains only the single\=/edge pieces. Clearly, $|P|\leq \distN + \distN^2\leq 2\cdot \distN^2$.

\paragraph*{Induction.}
Now we are ready to describe the inductive process of the construction. {We will process set $P$ iteratively until the pieces in $P$ have pairwise disjoint domains.
In each step we will replace some of the pieces in the set $P$ with their union.}
This will effectively reduce the size of $P$ each step $s=0,1,\ldots,N\leq |P|\leq  2\cdot \distN^2$ and will construct a sequence of sets of pieces $P=P_0,P_1,\ldots,P_N$. During these steps we will ensure the following invariants:
\begin{enumerate}
    \item The set $P_s$ has at most $|P|{-}s$ elements.
    \item Each piece in $P_s$ is connected as a~graph.
    \item If $\intq',\intq''\in P_s$ are distinct pieces then they do not share {non\=/coloured elements.} 
    \item The union $\bigcup P_s$ treated as an~interpretation equals $\intq$ (taking into account domain, concepts, and roles).
    \item For each piece $\intq'\in P_s$ there exists a~coloured homomorphism $h_{\intq'}\from \intq'\to\intp$.
    \item Consider $f_1,f_2\in \Del$ that are two coloured elements such that $f_1\in\Delta^{\intq_1}$ and $f_2\in\Delta^{\intq_2}$ with $\intq_1,\intq_2\in P_s$. If $\theta(f_1)=\theta(f_2)$ then $h_{\intq_1}(f_1)$ and $h_{\intq_2}(f_2)$ are $\sim_{\distN\cdot (|P|{-}s)}$\=/equivalent.
\end{enumerate}

All the conditions except the last one simply express how the pieces $P_s$ split $\intq$. The last condition is crucial: it preserves the ability to glue pieces whenever they share elements of the same colour $\theta(f_1)=\theta(f_2)$.

\paragraph*{Induction base.} We begin with the induction base.

\begin{lemma}
    The above invariants are met when $s=0$ (i.e.~for $P_0$).
\end{lemma} 

\begin{proof}
Clearly $P_0=P$ has $|P|-0$ elements. Each piece is a basic piece and is connected by the definitions of a basic piece. The fact that distinct pieces do not share non\=/coloured elements follows from the construction: only pieces of the form $\intq_e$ for $e\in \Nel$ contain any non\=/coloured elements, and as previously observed, if such pieces share a non\=/coloured element then they coincide.

We will now show that the union $\bigcup P$ of pieces in P treated as an interpretation is exactly $\intq$. It is clear that the domain of the union covers the whole domain of $\intq$. This is achieved already by pieces $\intq_e$ for $e\in \Delta^\intq$. Since unary types in $\intq$ and in the pieces coincide, this also covers the concepts realisations.
Thus, it remains to see that every role representation in $\intq$ is present in some piece in $P$. Indeed, every role representation is preserved, which depends on what kind of elements it connects:
\begin{itemize}
    \item for two non\=/coloured elements $e$, $e'$,  the role is in $\intq_e {=} \intq_{e'}$;
    \item for a non\=/coloured element $e$ connected to a~coloured element $f$, the role is in $\intq_e$;
    \item for two coloured  elements $f$ and $f'$, the role is in $\intq_{(f,f')}$.
\end{itemize}

It remains to show the last two {invariants}, which concern the homomorphisms $h_{\intq'}\from \intq'\to\intp$ for {pieces} $\intq'\in P_0=P$. The goal is to define homomorphisms $h_{\intq'}$ so they coincide with $h\from \intq \to\intp/{\sim_K}$. We have three types of basic pieces to consider.

We begin with pieces of the form $\intq_f$ for a~coloured element $f \in \Del$. Then, the domain of $h_{\intq_f}$ is $\{f\}$, $h(f)=[f']_{\sim_K}$ for some $f' \in \dom(\kappa)$, and
$\theta(f)=\kappa(f')$, as defined in the \emph{Pieces} paragraph. Since $h$ is a homomorphism, setting $h_{\intq_f}(f)=f'$ defines a~homomorphism from $\intq_f$ to $\intp$.

{
Now, consider a~piece of the form $\intq_{(f_1,f_2)}$ for $(f_1,f_2)\in R$. Then, the domain of $h_{\intq_{(f_1,f_2)}}$ is $\{f_1,f_2\}$.
To define the homomorphism take any edge connecting $f_1$ and $f_2$ in the interpretation $\intq$. As $f_1$ and $f_2$ are connected by an edge, by symmetry we can assume that $(f_1,f_2)\in r^{\intq}$ for some role name $r$. 
Now, since $h$ is a~homomorphism, $(h(f_1)$, $h(f_2)) \in r^{\intp/(\sim_K)}$. Hence, there are $f_1' \in [f]_{\sim_K}$  and $f_2' \in [f_2]_{\sim_K}$ such that $(f_1',f_2') \in r^{\intp}$. We call such elements \emph{witnesses} for the respective edge.

We define $h_{\intq_{(f_1,f_2)}}$ by map $f_1\mapsto f'_1$ and $f_2\mapsto f'_2$. This map clearly defines the desired homomorphism.
Since for $i=1,2$ holds that $[f'_i]_{\sim_K}=h(f_i)$, the unary types of $f'_i$ and $h(f_i)$ agree. Preservation of role $r$ follows directly form the construction.
For any other role name $s$ such that $(f_1,f_2) \in s^{\intq}$, since $h$ is a homomorphism, there are $f''_1$ and $f''_2$ such that $(f_1'',f_2'') \in s^{\intp}$ and are $\sim_K$ equivalent respectively to $f'_1$ and $f'_2$. Thus, by $\sim_1$ equivalence of $f'_1$ and $f''_1$ ($K \geq 1$), we know that there is $f'''_2$ which is connected by an $s$ edge to $f'_1$. Note that $\kappa(f'_2)=\kappa(f''_2)=\kappa(f'''_2)=\theta(f_2)$. It remains to see that $f'''_2=f'_2$, but this follows directly from the fact that $\intp$ is $K$\=/sparse and $K{>}0$.
}

%
%
%
%

What remains is to define the homomorphism for pieces of the form $\intq_e$ for non\=/coloured elements $e\in \Nel$. Here the domain is $\Del(e) \cup \Nel(e)$.
The part $h\restr \Nel(e)$ is already a~proper homomorphism $h'$ of $\intq\restr \Nel(e)$ into $\intp$ because there is no ambiguity here: if $e'\in \Delta^\intp\setminus \dom(\theta)$ then $[e']_{\sim_K}=\{e'\}$.

What remains is to define $h'$ on arguments in $\Del(e)$. Consider any $f\in \Del(e)$ and take any edge $(e',f)$ which witnesses that $f \in \Del(e)$ with $e'\in \Nel(e)$. Since $h$ is a~homomorphism, there is an~edge between $h(e')$ and $h(f)$ in $\intp/(\sim_K)$. Let $f_{e',f}\in \dom(\kappa)\subseteq \Delta^\intp$ be a~witness for this edge, i.e.~there is a~respective edge between $h'(e')$ and $f_{e',f}$ in $\intp$.

Our goal is to prove that whenever there are edges between $e'$ and $f$ and between $e''$ and $f$ for $e',e''\in \Nel(e)$, then the chosen witnesses are equal, i.e. $f_{e',f}=f_{e'',f}$. However, since $e'$ and $e''$ are both in $\Nel(e)$.  both $e'$ and $e''$ are connected by a~path in $\intq\restr\Nel(e)$. Since $|\Delta^\intq|=\distN$, we know that there must exist such path of length at most $\distN$ in $\intq$. Thus, due to the existence of the homomorphism $h'\restr \Nel(e)\from \intq\restr\Nel(e)\to \intp$, there must be a~path of length at most $\distN$ {connecting $e'$ and $e''$} in $\intp$. {This path is } of the form:
\[f_{e',f}\to h'(e') \to \cdots \to h'(e'') \to f_{e'',f}.\]
    
{Since $\intp$ is $K$\=/sparse, and $\kappa(f_{e',f}) = \kappa(f_{e'',f})$, it holds that $f_{e',f}=f_{e'',f}$.}

Now it remains to see that initially, i.e. when $s=0$, we have 
{$h_{\intq_1}(f_1) \sim_{\distN\cdot |P|}  h_{\intq_2}(f_2)$}
whenever $\theta(f_1)=\theta(f_2)$. However, the above definition guarantees that $[h_{\intq'}(f)]_{\sim_K} = h(f)$. Since $K=2\cdot \distN^3$, we know that 
{$K\geq \distN\cdot |P|$},
so whenever $\theta(f_1)=\theta(f_2)$ we necessarily have 
{$h_{\intq_1}(f_1) \sim_{\distN\cdot |P|}  h_{\intq_2}(f_2)$.}
\end{proof}

\paragraph*{Inductive step.}
{We will now describe how to perform the iterative step preserving
the invariants.}


Assume that $\intq_1$ and $\intq_2$ are two pieces in $P_s$ that share some element in the domain. This element is necessarily coloured, i.e.~it is $f\in \Del$ such that $f\in \Delta^{\intq_1}\cap\Delta^{\intq_2}$. Remove from $P_s$ the pieces $\intq_1$ and $\intq_2$ and put instead $\intq_1\cup \intq_2$, obtaining $P_{s+1}$. We need to check that all the invariants are preserved.

{
Before we discuss the preservation of the invariants, let us explain how to infer the theorem.
First, observe that the inductive step can be performed only if the domains of the pieces are not pairwise disjoint.
Thus, the final set of pieces $P_s$ necessarily consists of pieces with pairwise disjoint domains. Note that in that case, we can define the homomorphism from $\intq$ to $\intp$ as union of the homomorphisms $h_{\intq'}$ for $\intq' \in P_s$. Since the domains of the pieces in $P_s$ cover the domain of $\intq$ and the homomorphism $h_{\intq'}$ are defined on 
disjoint domains, their union is a proper homomorphism form $\intq$ to $\intp$ which concludes the implication ``if'' implication of the theorem. 
Thus, to end the proof of the theorem it is enough to show that the inductive step preserves the invariants.
}

Preservation of the first four invariants is \st{obvious} {clear}, because the only thing the union does is merging pieces which share an~element in their domain. We need to verify the last two invariants, in particular we need to define the homomorphism $h_{\intq_1\cup \intq_2}$.

{

Let $h_1\from \intq_1\to \intp$ and $h_2\from\intq_2\to \intp$ be the two homomorphisms and $S = \Delta^{\intq_1} \cap \Delta^{\intq_2}$ be the set of the common elements.
By the assumption $|S|>0$, so we can fix a common element $f \in S$, and set $f_1=h_1(f)$, $f_2=h_2(f)$.
By Invariant~6, $f_1$ and $f_2$ are $\sim_{\distN\cdot (|P|{-}s)}$\=/equivalent and thus $\sim_{\distN}$\=/equivalent, as clearly $s < |P|$. 
Hence, there exists a~coloured homomorphism $h'\from \intp\restr N_\distN(f_1)\to \intp\restr N_\distN(f_2)$. Moreover, since $N_\distN(f_2)$ is connected, graph $\intp$ is $K$\=/sparse and $h_1$, $h_2$ are colour preserving, $h'(h_1(f')=h_(f')$ for all $f' \in S$.

We define the homomorphism $h_{\intq_1 \cup \intq_2} \from \intq_1 \cup \intq_2 \to \intp$ as $h'\circ h_1 \cup h_2$.
Since $h'\circ h_1$ and $h_2$ agree on set $S$, this is indeed a well\=/defined function.
To show Invariant 5, we need to prove that it is also a coloured homomorphism.

For unary types map $h_{\intq_1 \cup \intq_2}$ preserves them as $h_1$ and $h_2$ are homomorphisms.
For any role name $r$ it is easy to check that if $(x,y) \in r^{\intq_1 \cup \intq_2}$ then either $(x,y) \in r^{\intq_1}$ or $(x,y) \in r^{\intq_2}$.
Thus, either $h_1$ or $h_2$ will preserve this edge.

}



Finally, it remains to see that the last invariant speaking about $\sim_{\distN\cdot (|P|-s)}$\=/equivalence is satisfied. The only place where the actual values of homomorphisms are changed is~$h'_1$. Thus, consider $f'_1\in\dom(h'_1)$ which is a~coloured element, any piece $\intq'\in P_{s+1}$, and any $f'_2\in\Delta^{\intq'}$ which is also coloured.
We assume that $\theta(f'_1)=\theta(f'_2)$ and we need to show that $y_2\eqdef h'_1(f'_1)$ and $y'_2\eqdef h_{\intq'}(f'_2)$ are $\sim_{\distN\cdot (|P|-s-1)}$ equivalent.

First note two obvious consequences of the definition of $\distN$\=/equivalence.

\begin{fact}
\label{ft:k-neightbour-monotone}
    For $\distN\geq \distN'\geq 0$, if $x\sim_\distN y$ then also $x\sim_{\distN'} y$.
\end{fact}

\begin{fact}
\label{fact:homo-distance}
    If $h\from \intp\to\intb$ is a~homomorphism and $x,y\in\Delta^\intp$ then $d_\intb(h(x),h(y))\leq d_\intp (x,y)$. 
\end{fact}

Observe that $y_1\eqdef h_1(f'_1)$ is $\sim_{\distN\cdot (|P|-s)}$ equivalent to $y'_2$ due to the inductive assumption about $h_1$ and $h_{\intq'}$. Fact~\ref{ft:k-neightbour-monotone} implies that $y_1$ and $y'_2$ are also $\sim_{\distN\cdot (|P|-s-1)}$ equivalent. Thus, due to transitivity of $\sim_{\distN\cdot (|P|-s-1)}$, it is enough to show that $y_1$ is $\sim_{\distN\cdot (|P|-s-1)}$ equivalent to $y_2$.

\begin{figure}[ht]
\colorlet{singleton}{myorange}
\colorlet{cycle}{ForestGreen}
\colorlet{cluster}{RoyalBlue}

\begin{tikzpicture}[scale=0.7, sm/.style={scale=0.7, circle, fill=black, inner sep=0pt,minimum size=5pt, outer sep=2pt}]

    \node[sm, label=above:$f_1 '$] (f1) at (-4,5) {};
    \node[sm, label=above:$f$] (f) at (-2,5) {};
    \node[sm, label=above:$f_2 '$] (f2) at (3,5) {};
    \draw[<->, decorate, decoration={coil,aspect=0, segment length = 2mm, amplitude = 0.6mm}] (f1) -- (f) node [midway, label=below:{\scriptsize{$\leq z$}}] {};
    
    \node (Q) at (-6,5) {\Large $Q$};
    \node (q1) at (-3, 6.5) {\color{cycle}$Q_1$};
    \fill [nearly transparent, cycle] (-3,5) circle [x radius = 1.5, y radius = 1];
    \node (q2) at (-1.6, 6.5) {\color{cluster}$Q_2$};
    \fill [nearly transparent, cluster] (-1.6,5) circle [x radius = 1.5, y radius = 1];
    \node (q') at (3,6.5) {\color{singleton}$Q'$};
    \fill [nearly transparent, singleton] (3,5) circle [x radius = 1.5, y radius = 1];

    \draw[dashed] (-6, 3.5) -- (5,3.5);

    \node[sm, label=below:$y_1$] (y1) at (-4.8, 2) {};
    \node[sm, label=below:$x_1$] (x1) at (-3.1, 2) {};
    \node[sm, label=below:$x_2$] (x2) at (-1, 2) {};
    \node[sm, label=below:$y_2 '$] (y2) at (3, 2) {};
    \draw[<->, decorate, decoration={coil,aspect=0, segment length = 2mm, amplitude = 0.6mm}] (y1) -- (x1) node [midway, label=above:{\scriptsize{$\leq z$}}] {};

    \node (I) at (-6, 2) {\Large $I$};

    \node at (-4,0.7) {\color{cycle} $h_1$};
    \filldraw [draw=cycle, semitransparent, pattern=north east lines, pattern color=cycle] (-4,2) circle [x radius = 1.4, y radius = 1]; 
    \draw[->, cycle] (f1) -- (y1);
    \draw[->, cycle] (f) -- (x1);

    \node at (-1, 0.7) {\color{cluster} $h_2$}; 
    \filldraw [draw=cluster, semitransparent, pattern=north east lines, pattern color=cluster] (-1,2) circle [x radius = 1.2, y radius = 1]; 
    \draw[->, cluster] (f) -- (x2);

    \node at (3, 0.7) {\color{singleton} $h_{Q'}$};
    \filldraw [draw=singleton, semitransparent, pattern=north east lines, pattern color=singleton] (3,2) circle [x radius = 1.2, y radius = 1]; 
    \draw[->, singleton] (f2) -- (y2);

    \draw[dashed] (-6, 0) -- (5, 0);


    \node[sm, label=below:$y_1$] at (-3, -1.5) {};
    \node[sm, label=below:{$x_1 = x_2$}] at (-1.5, -1.5) {};
    \node[sm, label=below:$y_2 '$] at (3, -1.5) {};
    

    \node (Im) at (-5.4,-1.5) {\Large $I$};

    \node at (-2.5, -3) {\color{cycle} $h_1'$};
    \filldraw [draw=cycle, semitransparent, pattern=north east lines, pattern color=cycle] (-2.5,-1.5) circle [x radius = 1.4, y radius = 1];
    \draw[->, thin, cycle] (-5.4, 2) -- (-3.9, -1.5); 
    \draw[->, thin, cycle] (-2.6, 2) -- (-1.1, -1.5); 

    \node at (-1.5, -3) {\color{cluster} $h_2'$};
    \filldraw [draw=cluster, semitransparent, pattern=north east lines, pattern color=cluster] (-1.5,-1.5) circle [x radius = 1.2, y radius = 1];
    \draw[->, thin, cluster] (-2.2, 2) -- (-2.7, -1.5); 
    \draw[->, thin, cluster] (0.2, 2) -- (-0.3, -1.5); 

    \node at (3, -3) {\color{singleton} $h'_{Q'}$};
    \filldraw [draw=singleton, semitransparent, pattern=north east lines, pattern color=singleton] (3,-1.5) circle [x radius = 1.2, y radius = 1];
    \draw[->, thin, singleton] (1.8, 2) -- (1.8, -1.5); 
    \draw[->, thin, singleton] (4.2, 2) -- (4.2, -1.5); 
\end{tikzpicture}
\caption{An illustration to the proof of Theorem \ref{thm:pfColourBlocking}.}
\end{figure}
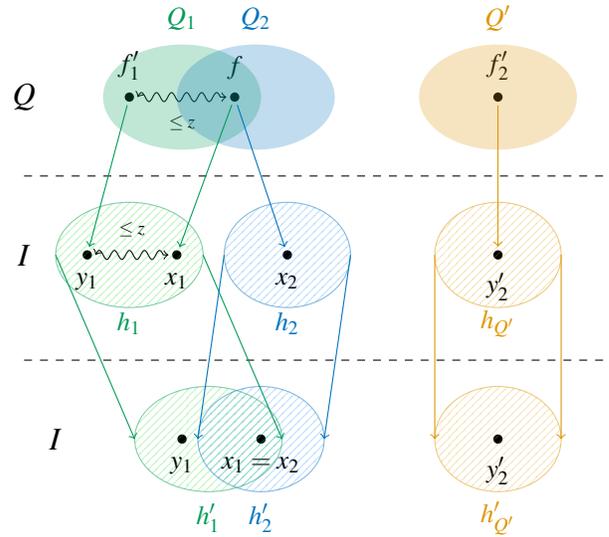

Note that both $f'_1$ and $f$ belong to $\Delta^{\intq_1}$ so due to the second invariant (and the assumption on the size of $\intq$) we know that $d(f'_1,f)\leq \distN$. This implies that for $x_1\eqdef f_1 = h_1(f)$ we have
\[d(y_1,x_1)=d(h_1(f'_1), h_1(f))\leq d(f'_1,f)\leq \distN.\]
Similarly, for $x_2=f_2=h_2(f)=h'_1(f)$ we have
\[d(y_2,x_2)=d(h'_1(f'_1), h'_1(f))\leq d(f'_1,f)\leq \distN.\]

Thus, we are in place to apply Lemma~\ref{lem:homo-decrease} as stated below, with:
\begin{itemize}
    \item $K=2\cdot \distN^3$ as fixed, $k=\distN\cdot(|P|-s)$, and $j=\distN\cdot(|P|-s-1)$;
    \item $y_1=h_1(f'_1)$, $y_2=h'_1(f'_1)=h'(y_1)$;
    \item $x_1=f_1=h_1(f)$, $x_2=f_2=h_2(f)=h'_1(f)$;
\end{itemize}

Thus, we know that $K\geq k$ because $|P|\leq 2\cdot \distN^2$. Clearly $k\geq j$ and $j\geq 0$ because $\distN>0$ and $|P|{-}s\geq 2$. Moreover, $x_1\sim_k x_2$ due to the inductive assumption about $h_1$ and $h_2$. The parameters are so that $k-j=\distN$, so we know that $d(y_1,x_1)\leq \distN$ and $d(y_2,x_2)\leq \distN$.

\paragraph*{Merging homomorphisms.}

Let $F=F^\intp$ be the set of coloured elements in $\intp$. First consider $x_1,x_2\in F$ which are two coloured
elements such that $x_1\sim_k x_2$. Let $\intp_1=\intp\restr N_\distN(x_1)$ and $\intp_2= \intp\restr N_\distN(x_2)$. The fact that $x_1\sim_k x_2$ means that $\kappa(x_1)=\kappa(x_2)$ and there exist homomorphisms $h_1\from \intp_1\to \intp_2$ and symmetrically $h_2\from \intp_2\to \intp_1$.

\begin{claim}
\label{cl:homo-iso}
Under the above assumptions, $(h_1\restr F)\from (\intp_1\restr F)\to(\intp_2\restr F)$ and $(h_2\restr F)\from (\intp_2\restr F)\to(\intp_1\restr F)$ are inverse isomorphisms between these two interpretations.
\end{claim}

\begin{proof}
First, coloured homomorphisms need to preserve being a~coloured
element, so the domains and co\=/domains of the homomorphisms agree.

Due to the assumption of $\distN$\=/sparseness, $x_1$ and $x_2$ are the unique coloured
elements in $\intp_1$ and $\intp_2$ of colour $\kappa(x_1)=\kappa(x_2)$. Thus, $h_1(x_1)=x_2$ and vice\=/versa $h_2(x_2)=x_1$.

Now, if $y\in \Delta^{\intp_1}\cap F$ is the unique coloured
element in $\intp_1$ whose colour is $\kappa(y)$ then $h_1(y)$ needs to satisfy $\kappa(h_1(y))=\kappa(y)$ and thus it is determined uniquely, as the only $\kappa(y)$ coloured element in $\intp_2$. This means that $(h_1\restr F)$ and $(h_2\restr F)$ are colour\=/preserving bijections. However, since they are both homomorphisms, it guarantees that they also need to preserve edges.
\end{proof}

The following lemma allows us to conclude the inductive step of the construction.

\begin{lemma}
\label{lem:homo-decrease}
Let $K\geq k\geq j\geq 0$ and $x_1\sim_k x_2$ be two coloured
elements in an~interpretation $\intp$ which is coloured by $\kappa\from F\to C$ in a~$K$\=/sparse way. If $y_1,y_2\in F$ are two coloured
elements such that $\kappa(y_1)=\kappa(y_2)$, $d(y_1,x_1)\leq k-j$, and $d(y_2,x_2)\leq k-j$ then $y_1\sim_j y_2$.
\end{lemma}

\begin{proof}
Let $\intp_1=\intp\restr N_k(x_1)$ and $\intp_2=\intp\restr N_k(x_2)$. Take $h_1\from \intp_1\to\intp_2$ and $h_2\from \intp_2\to\intp_1$ witnessing the fact that $x_1\sim_k x_2$. Let $y_1$ and $y_2$ be as in the statement. Due to Claim~\ref{cl:homo-iso} we know that $h_1(y_1)=y_2$ and $h_2(y_2)=y_1$. Note that since $d(y_1,x_1)\leq k-j$, we know that $N_j(y_1)\subseteq N_k(x_1)$, analogously for $y_2$.

Recall that below $\range(h)$ is the set of values of a~homomorphism $h$. 

Let $h'_1=(h_1\restr N_j(y_1))\from \intp_1\restr N_j(y_1) \to \intp_2$ and $h'_2=(h_2\restr N_j(y_2))\from \intp_2\restr N_j(y_2) \to \intp_1$ be the restrictions of the homomorphisms $h_1$ and $h_2$ to the neighbourhoods of $y_1$ and $y_2$ respectively. Due to Fact~\ref{fact:homo-distance}, the set of values $\range(h'_1)$ is contained in $N_j(h_1(y_1))=N_j(y_2)$ and dually for $\range(h'_2)$. This means that $h'_1$ and $h'_2$ are colour preserving homomorphisms between $\intp\restr N_j(y_1)$ and $\intp\restr N_j(y_2)$, which, together with the fact that $\kappa(y_1)=\kappa(y_2)$ guarantees that $y_1\sim_j y_2$.
\end{proof}

\clearpage
\section{Missing proofs for Section~\ref{sec:multiple-roles}}
\begin{figure}[ht]
    \centering
    \begin{tikzpicture}
        \node (Intp) at (0,0) {$\intp$};
        \node (G) at (2,0) {$\GG$};
        \node (Hn) at (4,0) {$\HH_n$};
        \node (Gn) at (6,0) {$\GG_n$};
        \node (Gbarn) at (8,0) {$\bar \GG_n$};
        
        \node (G-unrav) at (2,-1.5) {$\widetilde{\GG}$};
        \node (Hn-unrav) at (4,-1.5) {$\widetilde{\HH_n}$};
        \node (Gn-unrav) at (6,-1.5) {$\widetilde{\GG_n}$};

        \node (Hn-truncated) at (4, -3) {$\HH^n$};

        \draw[->] (Intp) to[bend left] node[above] {\scriptsize Theorem~\ref{thm:piecewise-elementary}} (G);
        \draw[->] (G) to[bend left] node[above] {\scriptsize Loop blow-up} (Hn);
        \draw[->] (Hn) to[bend left] node[above] {\scriptsize Loop trans. clos.} (Gn);
        \draw[->] (Gn) to[bend left] node[above] {\scriptsize Break cycles} (Gbarn);
        \draw[->] (G) to node[left] {\scriptsize Quasi-unrav.} (G-unrav);
        \draw[->] (Hn) to node[left] {\scriptsize Quasi-unrav.} (Hn-unrav);
        \draw[->] (Gn) to node[left] {\scriptsize Quasi-unrav.} (Gn-unrav);
        \draw[->] (Hn-unrav) to node[right] {\scriptsize Truncate loops} (Hn-truncated);
    \end{tikzpicture}
    \caption{Overview of the constructions used in proofs of Section~\ref{sec:multiple-roles}.}
    \label{fig:app-C-overview}
\end{figure}
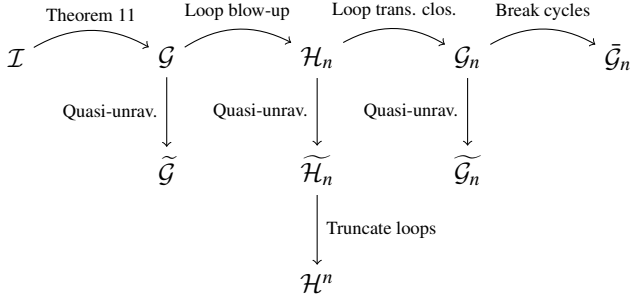

This Section of the Appendix will use multiple constructions to gradually transform an arbitrary counter model $\intp$ into a suitable piecewise elementary interpretation of exponential size. The diagram of Figure~\ref{fig:app-C-overview} shows an overview of this process together with the associated notation. 

\subsection{Proof of Theorem \ref{thm:piecewise-elementary}}
\thmpiecewiseelementary*

An overview of the proof of this Theorem is shown in Figure~\ref{fig:proof-thm-piecewise-elementary}. In relation to the diagram of Figure~\ref{fig:app-C-overview} it is a more detailed view of the first arrow (labelled ``Theorem 10''). 

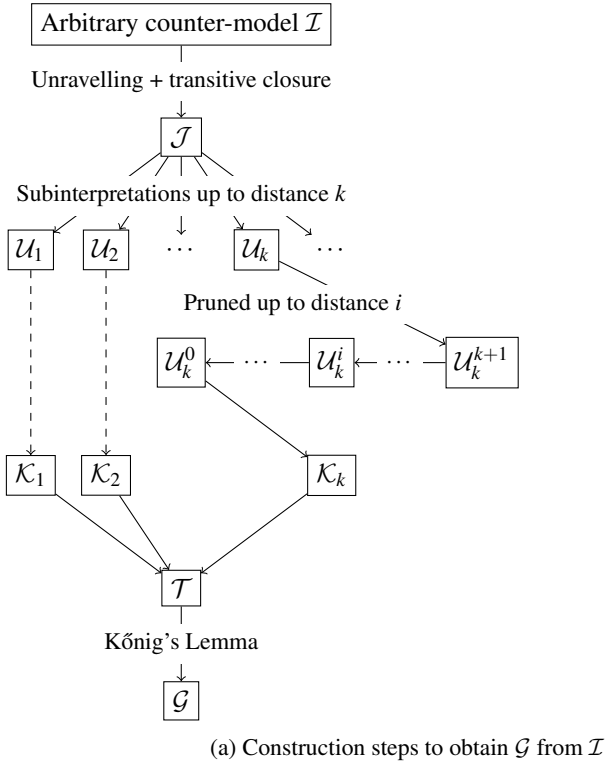
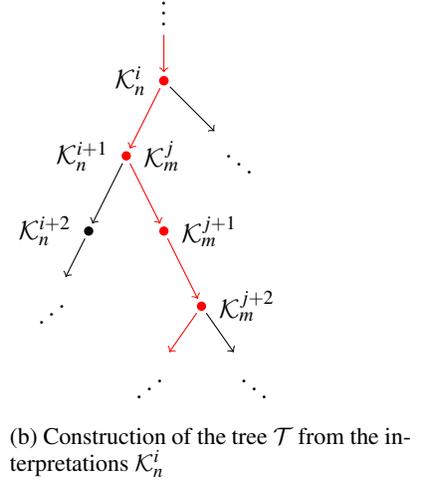
\begin{figure*}
    \centering
    \begin{subfigure}{0.6\textwidth}
        \begin{tikzpicture}
            \node[draw, rectangle] (I) at (0,0) {Arbitrary counter-model $\intp$};
            \node[draw, rectangle] (J) at (0,-1.5) {$\JJ$};
            \node[draw, rectangle] (U1) at (-2,-3) {$\UU_1$};
            \node[draw, rectangle] (U2) at (-1, -3) {$\UU_2$};
            \node (J-anchor-1) at (0, -3) {$\cdots$};
            \node[draw, rectangle] (Uk) at (1, -3) {$\UU_k$};
            \node (J-anchor-2) at (2, -3) {$\cdots$};
            \node[draw, rectangle] (Uk0) at (0, -4.5) {$\UU_k^0$};
            \node[draw, rectangle] (Uki) at (2, -4.5) {$\UU_k^i$};
            \node[draw, rectangle] (Ukk) at (4, -4.5) {$\UU_k^{k+1}$};
            \node[draw, rectangle] (K1) at (-2, -6) {$\KK_1$};
            \node[draw, rectangle] (K2) at (-1, -6) {$\KK_2$};
            \node[draw, rectangle] (Kk) at (2, -6) {$\KK_k$}; 
            \node[draw, rectangle] (T) at (0, -7.5) {$\TT$};
            \node[draw, rectangle] (G) at (0, -9) {$\GG$};

            \draw[->] (I) to node[pos=0.5, fill=white] {\footnotesize  Unravelling + transitive closure} (J);
            \draw[->] (J) -- (U1);
            \draw[->] (J) -- (U2);
            \draw[->] (J) -- (Uk);
            \draw[->] (J) -- (J-anchor-2);
            \draw[->] (J) to node[pos=0.5, fill=white] {\footnotesize  Subinterpretations up to distance $k$} (J-anchor-1); 
            \draw[->] (Uk) -- (Ukk);
            \draw[->] (Ukk) to node[pos=0.5, fill=white] {\footnotesize  $\cdots$} (Uki);
            \draw[->] (Uki) to node[pos=0.5, fill=white] {\footnotesize  $\cdots$} (Uk0);
            \path[draw=none] (Uk) to node[pos=0.5, fill=white] {\footnotesize  Pruned up to distance $i$} (Uki);
            \draw[->] (Uk0) -- (Kk);
            \draw[->, dashed] (U1) -- (K1);
            \draw[->, dashed] (U2) -- (K2);
            \draw[->] (K1) -- (T);
            \draw[->] (K2) -- (T);
            \draw[->] (Kk) -- (T);
            \draw[->] (T) to node[pos=0.5, fill=white] {\footnotesize Kőnig's Lemma} (G);
        \end{tikzpicture}
        \caption{Construction steps to obtain $\GG$ from $\intp$}
    \end{subfigure}
    \begin{subfigure}{0.3\textwidth}
        \begin{tikzpicture}[sm/.style={scale=0.7, circle, fill=black, inner sep=0pt,minimum size=5pt, outer sep=2pt}]
            \node (past) at (0,0) {$\vdots$};
            \node[sm, red, label=left:{$\KK_n^i$}] (Kni) at (0,-1) {$\cdot$};
            \node (alt) at (1, -2) {$\ddots$};
            \node[sm, red, label=left:{$\KK_n^{i+1}$}, label=right:{$\KK_m^j$}] (Kni+1) at (-0.5, -2) {}; 
            \node[sm, label=left:{$\KK_n^{i+2}$}] (Kni+2) at (-1, -3) {}; 
            \node (Kn++) at (-1.5, -4) {$\iddots$}; 
            \node[sm, red, label=right:{$\KK_m^{j+1}$}] (Kmj+1) at (0, -3) {};
            \node[sm, red, label=right:{$\KK_m^{j+2}$}] (Kmj+2) at (0.5, -4) {};
            \node (Kmj++) at (1.2, -5) {$\ddots$};
            \node (future) at (-0.2, -5) {$\iddots$};

            \draw[->, red] (past) -- (Kni);
            \draw[->, red] (Kni) -- (Kni+1);
            \draw[->] (Kni) -- (alt);
            \draw[->] (Kni+1) -- (Kni+2);
            \draw[->] (Kni+2) -- (Kn++);
            \draw[->, red] (Kni+1) -- (Kmj+1);
            \draw[->, red] (Kmj+1) -- (Kmj+2);
            \draw[->] (Kmj+2) -- (Kmj++);
            \draw[->, red] (Kmj+2) -- (future);
        \end{tikzpicture}
        \caption{Construction of the tree $\TT$ from the interpretations $\KK^i_n$}
    \end{subfigure}
    \caption{Overview of the proof strategy for Theorem~\ref{thm:piecewise-elementary}.}
    \label{fig:proof-thm-piecewise-elementary}
\end{figure*}
 
For convenience we will use at the beginning the classical notion of unravelling.

\begin{definition}[Unravelling] The unravelling of a~rooted interpretation $(\intp, a)$  is the rooted interpretation $(\widetilde\intp,a)$ obtained as follows. The domain of $\widetilde\intp$ is the set of paths in $\intp$ that start at $a$. For a~concept name $A$, we let  \[A^{\widetilde\intp}\ =  \big\{ pu \in \Delta^{\widetilde \intp} \bigm| u\in A^\intp \big\}\,.\] For a role $r$, let \[r^{\widetilde\intp}= \big\{ (pu,puv)  \bigm| (u,v)\in r^\intp \big\}\,.\]  
where $p$ is a path in $\intp$ that ends in a predecessor of $u$. 

\end{definition}

Start with an arbitrary counter model $\II$.  Let every critical element $f$ in $\II$ get a~unique concept name $Cr_f$. After introducing those new concepts, let $(\JJ,a)$ be the transitive closure of the classic unravelling of $\II$. The interpretation $\JJ$ homomorphically maps onto~$\II$.

Let $\UU_k$ be subinterpretation of $(\JJ,a)$ consisting of elements reachable from $a$ in at most $k$ steps. An element $v$ has height $h$ if $v\in \UU_h \setminus \UU_{h-1}$. Notice that due to transitivity, all elements of a maximally connected component over a transitive role $t$ that do not have predecessors in roles different than $t$ will belong to the same $\UU_k$.

We show by induction how to turn any $\UU_k$ into the quasi\=/unravelling of a piecewise elementary interpretation $\KK_k$. We will proceed bottom-up, from the $\UU_k$ that contains elements furthest away from the root to $\UU_1$. For each $\UU_k$, we consider elements for each height $i$ from $k$ to $1$ in descending order and remove some elements such that the following invariants are preserved at step $i$:

\begin{itemize}
    \item There exists a $t$-strong homomorphism from $\UU_k^i$ to $\UU_k$,
    \item All elements at height at most $i-1$ are identical in $\UU_k^i$ and $\UU_k$, 
    \item All outgoing edges from elements at height $i-1$ lead to the root of interpretation which is a quasi\=/unravelling of piecewise elementary interpretation of height $k-i$ 
\end{itemize}

(Base case) We start the induction at $i=k+1$ and define $\UU_k^{k+1}$ as $\UU_k$. All invariants are trivially true.

(Induction step) 
Let $C$ be a connected component of $\UU_k^{i}$. 

If $C$ is a non-transitive component, then it is necessarily a single element and so it forms a quasi\=/unravelling of piecewise elementary interpretation of height $k-i$. 



Consider now  $C$ is a component of a transitive role $s$ on level $i$. If $C$ contains multiple instances of the same critical element $f$, we can safely remove all copies but one and redirect all edges to and from the kept element. This preserves the mapping into $\II$, since all elements in the concept $Cr_f$ are mapped to the unique element of $\II$ of the same concept. Let $C'$ be the interpretation obtain from $C$ by applying this pruning. 

Let $v$ be an element in $C'$ at level $i$. The third invariant of the inductive hypothesis guarantees that all elements $w$ at level $i+1$ reachable from any $v$ are roots of a piecewise elementary interpretations of height $k-i-1$. Since there are only finitely many piecewise elementary interpretation of heigh $k-i-1$ we can obtain a labelling $\lambda$ with finite set of labels. Let $K$ be the elementary interpretation of size at most $(l+1)^{l+1^2}$ obtained from applying Theorem~\ref{thm:onerole} to $C'$ in which each element $v$ is labelled with the set of labels from $\lambda$ that correspond the the piecewise elementary interpretation of height $k-i-1$ reachable from $v$. 

From Theorem \ref{thm:onerole}, we know that the quasi-unravelling $\widetilde{K}$ of $K$ maps to $C'$ via a $t$-strong homomorphism. Moreover all elements of $\widetilde K$ that correspond to the same element in $K$ have isomorphic set of descendants, since this information was coded in the labelling. 

We replace $C'$ in $\UU_k^{i+1}$ by $\widetilde K$ and remove all unreachable elements. By materializing piecewise interpretations from $\lambda$ as children of $K$ we get suitable piecewise elementary interpretation describing  this part of $\UU_k^i$.

By definition $\UU_k^0$ is a quasi\=/unravelling of a piecewise elementary interpretation and we can take it as $\KK_k$.

We now build a tree $\TT$ from the obtained interpretations as follows. For each $n\in \mathbb{N}$, the interpretation $\KK_n$ and all of its truncations $\KK_n^i$ for $0 \leq i < n$ are nodes. Two nodes $K_n^i$ and $K_{n'}^{i+1}$ are connected if $K_n^i$ is a substructure of $K_{n'}^{i+1}$. We thus obtain an infinite tree of larger and larger interpretations. 



\begin{lemma}\label{lem:konig-lemma}[Kőnig's lemma]
    Let $G$ be a connected locally finite infinite graph. Then, for any vertex $v$ of $G$, there exists a one-way infinite path with initial vertex $v$ \cite{konig-lemma,introduction-to-graph-theory}.
\end{lemma}

Since, by previous construction, all interpretations $\KK_k$ have bounded degree, and by construction the tree $\TT$ is locally finite, by Lemma~\ref{lem:konig-lemma} we obtain that $\TT$ necessarily contains an infinite branch. Let $\GG$ be the limit interpretation of this branch. 

\begin{lemma}
 $(\widetilde{\GG},a)$ is a counter\=/model for $Q$
\end{lemma}
\begin{proof}
    If a CQ $q \in Q$ had a match in $\GG$, then it would have a match in a finite prefix of $\GG$. This prefix would be a subinterpretation for certain $\widetilde{\KK_k}$. But $\widetilde{\KK_k}$ maps into $\II$. Therefore $q$ would have a match in $\II$, which is not possible.
    
    Every element $g$ of $\GG$ of height $i$  is in a $\GG_k$ for $k>i$. Since all elements of $\GG_k$ have all necessary witnesses for elements except last layer,
    $g$ has all required witnesses in $\GG_k$. Hence also in $\GG$.
\end{proof}

The size bounds from Theorem \ref{thm:onerole} give necessary bounds in Theorem \ref{thm:piecewise-elementary}.


\subsection{Proof of Lemma \ref{lem:weaker-criterion}}
\weakercriterion*

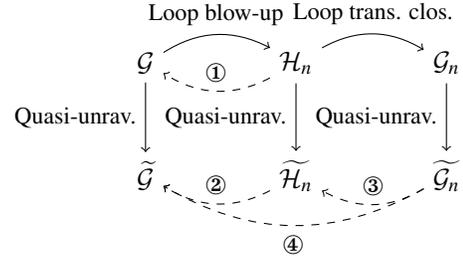
\begin{figure}
    \centering
    \begin{tikzpicture}
        \node (G) at (0,0) {$\GG$};
        \node (Hn) at (2,0) {$\HH_n$};
        \node (Gn) at (4,0) {$\GG_n$};
        
        \node (G-unrav) at (0,-1.5) {$\widetilde{\GG}$};
        \node (Hn-unrav) at (2,-1.5) {$\widetilde{\HH_n}$};
        \node (Gn-unrav) at (4,-1.5) {$\widetilde{\GG_n}$};

        \draw[->] (G) to[bend left] node[above] {\footnotesize Loop blow-up} (Hn);
        \draw[->] (Hn) to[bend left] node[above] {\footnotesize Loop trans. clos.} (Gn);
        \draw[->] (G) to node[left] {\footnotesize  Quasi-unrav.} (G-unrav);
        \draw[->] (Hn) to node[left] {\footnotesize  Quasi-unrav.} (Hn-unrav);
        \draw[->] (Gn) to node[left] {\footnotesize  Quasi-unrav.} (Gn-unrav);

        \draw[dashed, ->] (Hn) to[bend left] node[above] {\ding{172}} (G);
        \draw[dashed, ->] (Hn-unrav) to[bend left] node[above] {\ding{173}} (G-unrav);
        \draw[dashed, ->] (Gn-unrav) to[bend left] node[above] {\ding{174}} (Hn-unrav);
        \draw[dashed, ->] (Gn-unrav) to[bend left] node[below] {\ding{175}} (G-unrav);
    \end{tikzpicture}
    \caption{The proof overview for Lemma~\ref{lem:weaker-criterion}. Dashed arrows show homomorphisms between the structures.}
    \label{fig:proof-lem-weaker-criterion}
\end{figure}

\begin{proof}
A diagrammatic overview of this proof is presented in Figure~\ref{fig:proof-lem-weaker-criterion}.
Let $(\HH_n,a)$ be the rooted interpretation obtained from $(\GG_n,a)$ by skipping the edges included in the final step of the construction (i.e., without the transitive closure of the first $n$ copies of the original cycle). 
By construction, $(\HH_n,a)$ maps homomorphically into $(\GG,a)$ (represented as \ding{172} in the Figure). Consequently, the quasi-unravelling $(\widetilde{\HH_n},a)$ of $(\HH_n,a)$ maps homomorphically into the quasi-unravelling $(\widetilde \GG,a)$ of $(\GG,a)$ (represented as \ding{173} in the Figure). Moreover, the quasi-unravelling $(\widetilde{\GG_n},a)$ of $(\GG_n,a)$ maps homomorphically into $(\widetilde{\HH_n},a)$: when mapping paths in $\GG_n$ to paths in $\HH_n$ we simply replace each $t$-edge $(x,y)$ absent in $\HH_n$ with the shortest directed $t$-path from $x$ to $y$ in $\HH_n$, which exists because only edges following by transitivity are missing from $\HH_n$ (represented as \ding{174} in the Figure). We conclude by composing  the homomorphisms (respresented as \ding{175} in the Figure).
\end{proof}

\subsection{Proof of Lemma \ref{lem:converse}}
\converse*

Let $\intp$ denote $\widetilde{\GG_n}$. The interpretation $\intp$ will be quotiented into a partially finite countermodel by Theorem \ref{thm:pfColourBlocking}.

Since quasi\=/unravelling does not multiply critical elements, the sets of critical elements in both ${\GG_n}$ and $\widetilde{\GG_n}$ are identical. Moreover, distances between critical elements are preserved. Therefore, in order to define $\ell$-sparse colouring on $\intp$ we can just define the colouring on $\GG_n$. 

Let $\ell=2\cdot |\Delta^{Q}|^{3}$.
The colouring will be done inductively in a greedy fashion. Take any critical element $e$ of $\GG_n$ and assign it the smallest natural number not present in the $2\ell$-neighbourhood of $e$. $\GG_n$ has bounded degree, so the $2\ell$-neighbourhoods have a bounded size. Therefore, such a greedy procedure will use only finitely many colours. Moreover, it is $\ell$-sparse. Suppose that we have $y,y'\in N_\ell(x)$ such that $\kappa(y)=\kappa(y')$. Then $d(y,y')\leq 2\ell$, so whichever element get its colour second, it could not get the colour of first element. 

We can now apply Theorem \ref{thm:pfColourBlocking}. Let $\JJ=\intp\slash\sim_\ell$. Notice that in $\JJ$ there are as many critical elements as we have equivalence classes of $\sim_\ell$. So we just need to argue that there are finitely many equivalence classes of $\sim_\ell$ to finish the proof. Since $\GG_n$ has neighbourhoods of bounded size, it has finitely many different neighbourhoods even considering the colouring. Thus, partial finiteness of $\JJ$ follows from the following lemma.
\begin{lemma}
    Take any critical element $e$ in $\GG_n$. Then $\widetilde{N_\ell^{\GG_n}(e)}$ is homomorphically equivalent to ${N_\ell^{\intp}}(e)$
\end{lemma}
\begin{proof}
    For the homomorphism from right to left:
    
    Each element in $\intp$ is associated with a path $p$ in $\GG_n$. Consider the suffix of $p$ that goes only through the elements of $N_\ell^{\GG_n}(e)$. This is an element of $\widetilde{N_\ell^{\GG_n}(e)}$. This mapping is a homomorphism.
    
    For the homomorphism from left to right: 

    Take the root $r$ of  the elementary interpretation (or a node if this is the piece) that is $\ell$ levels above $e$. Let $p$ be a path from the root of $\GG_n$ to $r$. Prepend $p$ to all paths defining elements in $\widetilde{N_\ell^{\GG_n}(e)}$. We will get path in $\GG_n$, so elements of $\intp$. This mapping is a homomorphism. 
\end{proof}

\subsection{Proof of Lemma \ref{lmm:the-weakest-criterion}}
\theweakestcriterion*
We start with the following lemma.

\begin{lemma}\label{lem:gk_To_Hn}
$\widetilde{\GG_n}$ homomorphically maps into $\widetilde{\HH_n}$ where $\HH_n$ is the rooted interpretation obtained from $(\GG_n,a)$ by skipping the edges included in the final step of the construction (i.e., without the transitive closure of the first $n$ copies of the original cycle). 
\end{lemma}
\begin{proof}
    interpretations $\GG_n$ and $\HH_n$ have the same set of elements. Each transitive edge $e$ in $\HH_n$ which is not present in $\GG_n$ can be translated into a path $p$ whose transitive closure evaluates to $e$. Therefore the set of paths in $\GG_n$ (and thus $\widetilde{\GG_n}$) can be translated into the set of paths in $\HH_n$ (and thus $\widetilde{\HH_n}$).
\end{proof}

\begin{definition}
 Let $\HH^n$ denote a subinterpretation of $\widetilde{\HH_n}$ where we take only those elements whose construction paths traverse only the first $n$ copies of elements of $\GG$ in each loop.
\end{definition}
\begin{lemma}
    $\HH^n$ homomorphically maps into $\GG_n$.
\end{lemma}
\begin{proof}
Let $f$ be the function from $\widetilde{\HH_n}$ to $\HH_n$ defined by the natural mapping induced by the unravelling procedure. This function is not necessarily an homomorphism since some edges generated by the transitive closure might not be present in the original structure. Let $f'$ denote $f$ restricted to the elements of $\HH^n$. Since $\GG_n$ is, by definition, transitively closed on the image of $f'$, the function $f'$ naturally defines a homomorphism from ${\HH^n}$ to $\GG_n$.

\end{proof}

Suppose  $\widetilde\GG_n \models Q$. Then there is a homomorphism $h:Q\rightarrow\widetilde\GG_n$. Since, by Lemma~\ref{lem:gk_To_Hn}, $\widetilde {\GG_n}$ homomorphically maps to $\widetilde{\HH_n}$, we get a homomorphism $k:Q\rightarrow\widetilde{\HH_n}$. We will show that $k$ can be changed in such a way that $\Ima(k) \subseteq \HH^n$.

Given an element $a$, we denote by $sub_{G}(a)$ the subinterpretation rooted in $a$ in the graph $G$ and by $Anc(a)$ the set of ancestors of $a$ in $\widetilde{\HH_n}$. 

Consider $x\in var(Q)$ such that $k(x) \not \in \HH^n$. It means that $k(x)$ was removed from $\widetilde{\HH_n}$ by truncating a branch. Let $B$ be the highest branch  in $\widetilde{\HH_n} $ such that $B\cap Anc(k(x))\neq \emptyset$ and let $x'$ be a query variable such that $B\cap Anc(k(x'))\neq \emptyset$ is the smallest possible set, i.e. the path defined by  $Anc(k(x'))$ leaves the branch $B$ the earliest. 

Let $v$ be the first element in $B$ such that $k(x')$ and $v$ come from the same element in $\GG$. Subinterpretations rooted in $k(x')$ and $v$ are isomorphic due to the definition of quasi\=/unravelling, let this isomorphism be denoted $f$. We claim that if we move images of all variables from $sub_{\widetilde{\HH_n}}(k(x'))$ to $sub_{\widetilde{\HH_n}}(v)$  through $f$ we will get function $k':var(Q)\rightarrow \widetilde{\HH_n}$ which is a function from $Q$. We have to prove it is a homomorphism.

Since $f$ is an isomorphism, the property is trivially true for all pairs of elements within $sub_{\widetilde{\HH_n}}(v)$.  
The only edges  incoming  to $k(x')$ are in the same transitive relation $s$ as in $B$. Therefore, by construction of elementary interpretations, if there is an edge from an element $u$ to $k(x')$, then there is also an edge from $u$ to $v$. By construction of quasi\=/unravelling, all elements in $sub_{\widetilde{\HH_n}}(k(x'))$ are reachable from outside through $k(x')$ or through a critical element, and they have to be reachable by an $s$ edge. If they are reachable through $k(x')$, then they also reachable through $v$. Since $f$ is defined as the identity on critical elements and their rooted subinterpretations, all incoming edges to those elements stay the same.

By analogical reasoning we can fix $k'$ such that the variable $y$ is within the second biggest $B\cap Anc(k''(y))\neq \emptyset$ and gets mapped into the second loop unravelling in $B$ and so on. Finally, we get a homomorphism $h$ such that for each variable $z$ we have $B\cap Anc(h(z)) \subseteq \HH^n$. We can iterate above procedure to correct homomorphisms against any branch in $\HH_n$. Since there are only finitely many branches in all $Anc(h(x))$'s, after finitely many branch corrections of $h$, we get a homomorphism whose image is in $\HH^n$. 

Suppose now that $\GG_n\models Q$. Since, by construction, loops in $\GG_n$ are much longer than the number of variables in $Q$, the satisfaction of $Q$ by $\GG_n$ cannot be triggered by loops. More formally, for every loop $l$ in $\GG_n$ there is an element $e_l\in l\setminus \Ima(Q)$ such that after removing from $\GG_n$ all elements $e_l$ the resulting interpretation will not have directed cycles. Let us denote this interpretation $\bar\GG_n$. By definition $Q$ homomorphically maps into $\bar\GG_n$, so to prove that $\widetilde{\GG_n}\models Q$ it is enough to prove the following lemma.

\begin{lemma}
    $\bar\GG_n$ homomorphically maps into $\widetilde{\GG_n}$.
\end{lemma}
\begin{proof}
    We will find desired homomorphism $h$ by induction on the distance from the root. It will be done in such a way that every element in $\bar\GG_n$ will be mapped into one of its copy in $\widetilde{\GG_n}$.

    Let  $v$ be an element of $\bar\GG_n$, $u$ a successor of $v$ and $v'$ the image by $h$ of $v$ in $\widetilde{\GG_n}$. We will map $u$ into descendants of $v'\in \widetilde{\GG_n}$. If $u$ is not a root of a cycle in $\GG_n$, then we map $u$ to the corresponding successor of $v'$. If $u$ is a root of a cycle $c$ of relation $r$ in $\HH_n$, then we do the following. Let $B$ be the branch in $\widetilde{\GG_n}$ that is the unravelling of $c$ we map $u$ to the second copy of $u$ in $B$. The residue of $c$ in $\bar \GG_n$ is mapped into $B$ between first and third copy of $u$. The non $r$-children of $u$ will be mapped to descendants of the second copy of $u$ in $B$. The $r$-children of $u$ will be mapped to children of third copy of $u$,  instead of children of the second copy. Notice that due to transitivity, $r$-children of the third $u$ copy will be $r$ connected to the second $u$ copy as well.

    As for the base step of this procedure, we add a virtual root to both interpretations and apply the inductive step. That is,  the root of $\bar\GG_n$ (the successor of the virtual root) will be mapped to the  descendants of the virtual root of $\widetilde{\GG_n}$.

\end{proof}

\subsection{Proof of Lemma \ref{lem:complexity}}
\complexity*

\begin{proof}
There are pros and cons for trying to find existence of a  $\GG_n$ directly. On the pro side $\GG_n \models Q$ iff $\widetilde{\GG_n} \models Q$ so query testing is direct. On the con side $\GG_n$ is a result of a variant of an unravelling, so the interpretation is not local. Therefore, our approach will be somewhere in between. We will look for $G$, but will check query satisfiability not in pieces of $G$, but in pieces of $\GG_n$.

The algorithm is a variant of the type elimination procedure. We would like to build  tree $\GG$ iteratively level by level with a Safety Condition. To this end we need to turn query non satisfiability into local condition. 

Our building blocks will be elementary interpretations with ports (for transitive roles) and level one trees with ports (for non transitive roles). Those blocks will be annotated additionally with sets of subqueries of $\QQ$ at the root and at the ports. The intended semantics is that if a copy of a block $b$ will end up in $\GG$, then subqueries form the  annotation of the block root $v$  will not be satisfied in subtree of $\GG_n$ rooted in a copy of $v$, provided that subqueries from port annotations are not satisfied in their respected subtrees.

The first step will be finding the set of valid blocks. Elementary interpretations have exponential size, so there are doubly exponentially many of them. Set of subqueries is exponential in size of conjunctive query disjuncts and linear in the number of disjuncts. Therefore for a given  elementary interpretation there are doubly exponentially many possible annotations. This means there are doubly exponentially many candidates for valid blocks, each of exponential size. When blowing up the structure as in Lemma~\ref{lem:weaker-criterion}, we copy annotations accordingly.  Naive validation of a single block will suffice. We check if a query form the root annotation is not satisfied in the blown-up block by considering all possible mappings of query variables to the block. Each block candidate will be checked in doubly exponential time. Multiply that time by the number of candidates to check, and we get preprocessing done in doubly exponential time.

Let $B_0$ be the set of blocks obtained in the preprocessing. If we would use a block $b\in B_0$ to construct $\GG$, then ports of $b$ will have to be filled with some other blocks. In general it could not be possible. So we will reduce the set $B_0$ by a fixpoint procedure:

\begin{quote}
$B_{i+1}$ is the set of those $b\in B_i$ whose witness ports can be fulfilled by blocks from $B_i$. 
\end{quote}

This fixpoint procedure has to stop after $|B_0|$ many steps, which is doubly exponential. Each step can be done in polynomial time in the size of $B_i$. Hence, the whole fixpoint procedure will take doubly exponential time.

If we end up with an empty $B_i$, it means that $\GG$  cannot exist. If we end up with nonempty $B_i$, then $\GG$ can be constructed in the greedy fashion.
\end{proof}

\end{document}